\def\one{one}
\def\two{two}

\def\cols{one} %

\ifx\cols\one
\documentclass[runningheads]{article}
\usepackage{fullpage}
\fi
\ifx\cols\two
\documentclass[conference,compsoc]{IEEEtran}
\fi

\usepackage[hidelinks]{hyperref}
\usepackage{pgf}
\usepackage{pbox}
\usepackage{longtable}
\usepackage{dsfont}
\usepackage{latexsym}
\usepackage{orcidlink}
\usepackage{authblk}
\usepackage[nocompress]{cite}
\usepackage{amssymb}
\usepackage{amsthm}
\usepackage{amsmath}
\usepackage{amsfonts}
\usepackage{comment} 
\usepackage{graphicx}
\usepackage{amscd}
\usepackage{color}
\usepackage{xcolor}
\usepackage{tikz}
\usepackage{enumerate}
\usepackage{xurl}
\usepackage{float}

\usepackage[noend]{algpseudocode}
\usepackage{dsfont}
\usepackage{xspace}

\renewcommand{\sf}[1]{\mathsf{#1}}

\newcommand{\poly}{\mathrm{poly}}

\newcommand{\negl}{\mathsf{negl}}

\newcommand{\true}{\mathsf{true}}
\newcommand{\false}{\mathsf{false}}
\newcommand{\N}{\mathbb{N}}

\newcommand{\bits}{\{0,1\}}
\newcommand{\botbits}{\{0, 1, \bot \}}
\newcommand{\getsr}{{\:{\leftarrow{\hspace*{-3pt}\raisebox{.75pt}{$\scriptscriptstyle\$$}}}\:}}

\newcommand{\calO}{\mathcal{O}}
\newcommand{\calU}{\mathcal{U}}

\newcommand{\A}{\mathcal{A}}
\newcommand{\B}{\mathcal{B}}

\newcommand{\fail}{\texttt{FAIL}\xspace}

\newcommand{\eH}{H_{e}} %
\newcommand{\Cmax}{C_{\max}}
\newcommand{\Vmax}{V}
\newcommand{\green}{\mathcal{G}}
\newcommand{\red}{\mathcal{R}}
\newcommand{\thresh}{\rho}

\newcommand{\CGscheme}{\wat[\sf{PRC}]}
\newcommand{\PRC}{\sf{PRC}}
\newcommand{\hT}{\hat{T}}
\newcommand{\htau}{\hat{\tau}}
\newcommand{\RCG}{R^\sf{CG}_1}
\newcommand{\RCGZ}{R^\sf{CGZ}_1}
\newcommand{\RF}{R^\sf{FGJ+}_1}
\newcommand{\RK}{R^\sf{KTHL}_1}
\newcommand{\hp}{\hat{p}}
\newcommand{\keyspace}{\mathcal{K}}

\newcommand{\FP}{\mathsf{FP}}
\newcommand{\FPTrace}{\mathsf{FP.Trace}}
\newcommand{\FPGen}{\mathsf{FP.Gen}}
\newcommand{\ColSet}{C}
\newcommand{\Feasible}{F}

\newcommand{\FPKey}{\mathsf{tk}}
\newcommand{\FPAdv}{\mathcal{A}}
\newcommand{\ECC}{\mathsf{ECC}}

\newcommand{\Msg}{\mathcal{W}}
\newcommand{\msg}{m}
\newcommand{\MsgSetup}{\mathsf{KeyGen}}
\newcommand{\Encode}{\mathsf{Wat}} %
\newcommand{\Extract}{\mathsf{Extract}}

\newcommand{\channel}{\mathcal{E}}

\newcommand{\prompt}{Q}
\newcommand{\sk}{\mathsf{sk}}
\newcommand{\pk}{\mathsf{pk}}
\newcommand{\GenModel}{\mathsf{Model}}
\newcommand{\tokSet}{\mathcal{T}}
\newcommand{\wat}{\mathcal{W}}
\newcommand{\Setup}{\mathsf{KeyGen}}
\newcommand{\Wat}{\mathsf{Wat}}
\newcommand{\Detect}{\mathsf{Detect}}
\newcommand{\Users}{\mathcal{U}}
\newcommand{\Trace}{\mathsf{Trace}}
\newcommand{\numBlocks}{\mathsf{NumBlocks}}

\newcommand{\Blocks}{\mathsf{Blocks}}
\newcommand{\block}{\beta}
\newcommand{\isBlock}{\mathsf{block}}
\newcommand{\calQ}{\mathcal{Q}}
\newcommand{\calR}{\mathcal{G}}

\newcommand{\Ex}{\mathbb{E}}
\newcommand{\1}{\mathds{1}}

\newcommand{\apref}[1]{Appendix~\ref{#1}}

\newcommand{\figref}[1]{Figure~\ref{#1}}

\newcommand{\lemref}[1]{Lemma~\ref{#1}}
\newcommand{\clmref}[1]{Claim~\ref{#1}}
\newcommand{\secref}[1]{Section~\ref{#1}}
\newcommand{\thmref}[1]{Theorem~\ref{#1}} %

\newcommand{\defref}[1]{Definition~\ref{#1}}

\newcommand{\sub}[1]{\mathrm{#1}} %

\newcommand{\codeBox}[1]{\noindent\fbox{#1}}

\newtheorem{thm}{Theorem}[section]
\newtheorem{obs}[thm]{Observation}
\newtheorem{assm}[thm]{Assumption}
\newtheorem{defn}[thm]{Definition}
\newtheorem{lem}[thm]{Lemma}
\newtheorem{cor}[thm]{Corollary}
\newtheorem{clm}[thm]{Claim}
\newtheorem{rem}{Remark}

\newenvironment{assumption}{\begin{assm}}{\end{assm}}

\newenvironment{lemma}{\begin{lem}}{\end{lem}}
\newenvironment{definition}{\begin{defn}}{\end{defn}}
\newenvironment{theorem}{\begin{thm}}{\end{thm}}
\ifx\cols\one

\fi

\ifx\cols\two
\renewcommand{\paragraph}[1]{\smallbreak{\textit{#1.}}}
\fi

\title{Watermarking Language Models for Many Adaptive Users}
\ifx\cols\two
\author{Anonymous Author(s)}
\fi
\ifx\cols\one
\author{Aloni Cohen \orcidlink{0000-0002-3492-2447}}
\author{Alexander Hoover \orcidlink{0009-0003-9818-1419}}
\author{Gabe Schoenbach \orcidlink{0000-0002-4300-0356}}

\affil{University of Chicago}
\date{} %
\fi

\begin{document}

\maketitle
\begin{abstract}
    We study watermarking schemes for language models with provable guarantees. As we show, prior works offer no robustness guarantees against \emph{adaptive prompting}: when a user queries a language model more than once, as even benign users do. And with just a single exception \cite{EPRINT:ChrGun24}, prior works are restricted to \emph{zero-bit} watermarking:  machine-generated text can be detected as such, but no additional information can be extracted from the watermark. Unfortunately, merely detecting AI-generated text may not prevent future abuses.

We introduce \emph{multi-user watermarks}, which allow tracing model-generated text to individual users or to groups of colluding users, even in the face of {adaptive prompting}. We construct multi-user watermarking schemes from undetectable, adaptively robust, zero-bit watermarking schemes (and prove that the undetectable zero-bit scheme of \cite{EPRINT:ChrGunZam23} is adaptively robust).
Importantly, our scheme provides both zero-bit and multi-user assurances at the same time. It detects shorter snippets just as well as the original scheme, and traces longer excerpts to individuals. 

The main technical component is a construction of message-embedding watermarks from zero-bit watermarks.
Ours is the \emph{first generic reduction}
between watermarking schemes for language models. 
A challenge for such reductions is the lack of a unified abstraction for \emph{robustness} --- that marked text is detectable even after edits. 
We introduce a new unifying abstraction called \emph{AEB-robustness}. AEB-robustness provides that the watermark is detectable whenever the edited text ``approximates enough blocks'' of model-generated output.

\end{abstract}

\section{Introduction}\label{sec-intro}

Generative AI models can now produce text and images not easily distinguishable from human-authored content. 
There are many concerns about the undisclosed use of generative AI, whether for nefarious or banal purposes. Watermarking is one approach for detecting and tracing the provenance of generative AI outputs, and it raises challenging technical and policy questions \cite{Brookings, Lawfare}.

Watermarking for generative AI is on the cusp of deployment.
In July 2023, the White House secured commitments from seven industry leaders to manage and mitigate risks posed by AI: Amazon, Anthropic, Google,
Inflection, Meta, Microsoft, and OpenAI \cite{WhiteHouse,WhiteHouse2}.
Among the commitments is the development of watermarking and other provenance
techniques for the next generation of audio and image models.
A few months later, President Biden issued an executive order directing federal agencies
to ``establish standards and best practices for detecting AI-generated content,''
including guidance on the use of watermarking \cite{WhiteHouse3}.
Google already watermarks all content produced by the Lyria music generation
model \cite{SynthID}, and OpenAI has a working prototype for watermarking text produced by
ChatGPT \cite{Aar22}.

\medskip
Though many watermarking schemes are brittle, 
a burgeoning line of research shows that strong statistical and cryptographic guarantees (and negative results \cite{ARXIV:ZEFVAB23}) can be proved or heuristically justified for language models specifically
\cite{Aar22,ICML:KGWKMG23,EPRINT:ChrGunZam23,EPRINT:FGJMMW23,ARXIV:KTHL23, EPRINT:ChrGun24,ARXIV:GM24}.
These schemes envision a setting where a user queries a language model through some interface. For example, querying GPT-4 using ChatGPT or the API. The model deployer, e.g.,\ OpenAI, uses a secret watermarking key in conjunction with the underlying language model to produce watermarked text. The mark can be extracted from the text using the secret key.

To be useful, watermarking schemes must have four properties. Existing schemes achieve these guarantees, though the formalisms and assumptions vary greatly.
First is \emph{soundness}: marks must not be falsely detected in text not generated by the watermarked model (low Type I error).
Second is \emph{completeness}: verbatim outputs of the watermarked model are detectable (low Type II error).
Third is \emph{robustness} to edits: marks are detectable even after marked text is edited, say by deleting or rearranging phrases.
Finally, watermarking should not noticeably degrade the quality of the model's outputs. The strongest version of this is \emph{undetectability}, requiring that the outputs of the marked and unmarked models are  indistinguishable.

All existing works share two major drawbacks. First, none of them considers robustness when users \emph{adaptively prompt} the model.
Second, merely detecting watermarked text is not always enough. As we explain below, sometimes we need to \emph{trace text to a specific user} of the language model.

\paragraph{Adaptive prompting}
We introduce and build watermarking schemes that are \textbf{robust to adaptive prompting}. 
Existing works define security for a single text $T$ produced by the model in response to any single prompt $\prompt$. But no user issues just one prompt! 
Even benign users interact with the models, adaptively
refining the prompts and generations as they go.
Current definitions don't even imply that the resulting text is marked, let alone robust to edits. 
We give the first definitions and proofs of robustness when users interactively query a language model and derive text from the whole interaction.

\paragraph{Multi-user watermarks}
We introduce and build a \textbf{multi-user watermarking scheme} for language models. In such a scheme, model-generated outputs can be traced to specific users, even when users \emph{collude}.

Detecting watermarks isn't always enough to mitigate harms.
Consider a ChatGPT-powered bot carrying out romance scams to trick victims into transferring large sums of money. 
The gullible victims won't check for watermarks themselves.
And any single messaging platform detecting the watermarks and banning 
the bot won't help much --- the scammer can always switch to a different platform. 
After the scam is revealed, the watermarked text --- the bot's messages --- may be available to law enforcement, the messaging platform, or the language model provider.
But there doesn't seem to be any way to arrest the scammer or  protect future victims.

We'd like to trace the watermarked text to a specific user.
The bot is querying ChatGPT using the scammer's user account or API token. 
If the marked text revealed the user, one could directly cut off the scammer's access to the model or seek legal recourse.

We also want security when multiple users collude. 
Consider three users each using ChatGPT to collaboratively write a shoddy legal brief or S\&P submission. 
They might independently borrow from their respective model outputs, while making edits and adding text of their own.
Robust multi-user watermarking guarantees that at least one of the three is identified, but no innocent users are.

As a first attempt at building multi-user watermarking for $n$ users, one might use a different secret key for each user. But merely detecting the watermark would require checking each of the $n$ secret keys. This would be too slow for widespread watermark detection: ChatGPT reportedly has more than 180 million monthly active users.
In contrast, our multi-user watermarking scheme takes $O(\log n)$ time to detect watermarked text (omitting other parameters). Without collusions, tracing also takes $O(\log n)$ time, which is optimal as $\log n$ bits are needed to identify a user. With collusions, tracing takes $O(n\log n)$ time (performing a $O(\log n)$-time check for each user).

Adding the ability to trace individual users should not compromise the detection of watermarked text in contexts where tracing may not be needed (e.g., spam filtering).
Our multi-user construction leaves untouched the guarantees of the underlying watermarking scheme. \emph{You get the best of both robustness guarantees at the same time!} A short marked text is detectable as before, and a longer marked text is traceable to individuals.

\paragraph{From zero-bit to message-embedding watermarks, generically}
We give a {black-box} construction of multi-user watermarking from existing watermarking schemes.
Our high level approach to building multi-user watermarking is relatively simple. 
Suppose you have an \emph{${L}$-bit watermarking scheme}: one that embeds an $L$-bit message into generated text. 
Then, ignoring collusions, the obvious idea is to embed an ID unique to each user.
(Our scheme can be made robust to collusions using a cryptographic primitive called a {fingerprinting code}, though this doesn't work generically.)
So it suffices to build an $L$-bit scheme out of a so-called \emph{zero-bit watermarking scheme}, where text is simply viewed as ``marked'' or ``unmarked.''\footnote{
With one exception \cite{EPRINT:ChrGun24}, all the existing schemes we study are zero-bit schemes.}
We use a natural idea. We sample $2L$ secret keys $k_{i,b}$, one for each index $i$ and bit $b$. To embed a message $m$, we use the zero-bit scheme with the keys $k_{i,m[i]}$ for each index $i$. The result is text watermarked under $L$ different keys which together reveal the message.

The challenge is saying anything interesting about our construction while treating the underlying watermarking scheme as a {black box}. 
It's not even clear how to state the appropriate robustness guarantee, let alone prove it.
The issue is that every  watermarking construction has a bespoke formulation of completeness\footnote{%
    All seemingly-incomparable lower bounds on entropy: watermark potential \cite{ARXIV:KTHL23}, min-entropy per block \cite{EPRINT:FGJMMW23}, spike entropy \cite{ICML:KGWKMG23}, empirical entropy \cite{EPRINT:ChrGunZam23,EPRINT:ChrGun24,ARXIV:GM24}, and on-average high entropy / homophily \cite{ARXIV:ZALW23}.} 
(which verbatim model outputs are marked) and a bespoke formulation of robustness\footnote{%
    All requirements on substrings: equality \cite{EPRINT:ChrGunZam23,EPRINT:FGJMMW23},
     edit distance \cite{ARXIV:KTHL23,ARXIV:ZALW23},
    produced by a binary-symmetric channel \cite{EPRINT:ChrGun24}, produced by an 
    edit-bounded channel \cite{ARXIV:GM24}.
} 
(which edits preserve the mark), see Appendix~\ref{sec:block-appendix}. 
There are two ways forward: choose a specific scheme and tailor the results, or invent a unifying language for watermarking robustness and completeness.

We present a new framework for describing robustness and completeness guarantees of watermarking schemes, called \textbf{AEB-robustness}. AEB-robustness provides that text is watermarked if it \emph{{\bf A}pproximates {\bf E}nough {\bf B}locks} of model-generated text. Specifying the robustness condition amounts to defining ``approximates,'' ``enough,'' and ``blocks.'' All else equal, a scheme is more robust if looser approximations are allowed; fewer blocks are required; or blocks require less entropy.

Our black-box reductions only affect how many blocks are enough, not what constitutes a block nor an approximation thereof. Our results hold for any (efficiently-checkable) definition of a block and any definition of string approximation.

With the language of AEB-robustness, our theorems are easy to state informally. Let $\lambda$ be a cryptographic security parameter. Suppose $\wat$ is a (zero-bit) watermarking scheme that is undetectable, sound, and robust whenever one block is approximated ($R_1$-robust). Our $L$-bit scheme is undetectable, sound, and robust whenever $k = O(L\lambda)$ blocks are approximated ($R_k$-robust).
Our multi-user scheme is undetectable, sound, and robust for $n>1$ users and $c>1$ collusions whenever $k = O(c^2\log^2(c)\log(n)\lambda)$ blocks are approximated.
We adopt the cryptographic approach of analyzing security against all efficient adversaries. A byproduct is that our theorems only apply when the underlying watermarking scheme is cryptographically-secure.

\subsection{Our contributions}
We continue the study of watermarking schemes for language models with provable guarantees. Except where specified, our constructions require cryptographically-strong undetectability and soundness, and that the underlying scheme is  AEB-robust against adaptive prompting.
\begin{enumerate}
    \item We define robustness against \emph{adaptive prompting}. We prove the first robustness guarantees against users making more than one query: our constructions are robust against adaptive prompting, as is the zero-bit scheme of Christ, Gunn, and Zamir \cite{EPRINT:ChrGunZam23}.\footnote{An earlier version of this paper incorrectly claimed a generic reduction of adaptive robustness to non-adaptive robustness for certain zero-bit schemes (Section~\ref{sec:non-adaptive-to-adaptive}.)}
    \item We construct \emph{$L$-bit watermarks} with provable robustness from zero-bit watermarking schemes. 
    This is the first black-box reduction among watermarking schemes, using a new framework for describing the robustness of watermarking schemes, called \emph{AEB-robustness}. 
    Our construction is the first $L$-bit undetectable watermarking scheme that has both short keys and any provable robustness guarantee (let alone adaptive robustness).\footnote{This is true even without undetectability. We include undetectability to rule out trivial schemes whose output is completely independent of the underlying language model (e.g., an error-correcting code). The only existing $L$-bit scheme with any provable robustness guarantee is \cite{EPRINT:ChrGun24}. The watermarking keys for that scheme are as long as the maximum generation length of the model. Zamir suggests that a version of his undetectable $L$-bit steganography scheme scheme \cite{ARXIV:Zam24} can be made ``substring complete'' (the robustness guarantee of \cite{EPRINT:ChrGunZam23}), but leaves doing so to future work.}
    \item We define and construct \emph{multi-user} watermarking schemes, which allow tracing model-generated text to users, even with collusions. Our construction leaves unaffected the stronger robustness of the underlying zero-bit watermarking scheme, essentially allowing both schemes to be used at once. Surprisingly, we also show that instantiating our multi-user construction with a \emph{lossy} $L$-bit scheme yields better parameters.
\end{enumerate}

\subsection{Related work}

A recent flurry of work on the theory of watermarking language models was kicked off by \cite{Aar22} and \cite{ICML:KGWKMG23}. We directly build on this line of work, especially those with strong provable guarantees for undetectability \cite{EPRINT:ChrGunZam23}, robustness \cite{ARXIV:KTHL23,EPRINT:ChrGun24} or public detection \cite{EPRINT:FGJMMW23,EPRINT:ChrGun24}. See Appendix~\ref{sec:block-appendix} \ifx\cols\two in the full version of this paper\ \fi for an in-depth discussion of these schemes.
We adapt and extend the cryptographic-style definitions of
\cite{EPRINT:ChrGunZam23,EPRINT:FGJMMW23,ARXIV:Zam24,EPRINT:ChrGun24,ARXIV:GM24}, which enables security proofs against arbitrary efficient adversaries.

Our constructions require a watermarking scheme that is zero-bit, undetectable, sound, and robust against adaptive prompting. Christ, Gunn, and Zamir \cite{EPRINT:ChrGunZam23} were the first to construct a zero-bit, undetectable, sound watermarking scheme with any provable robustness. In Appendix \ref{ap:adaptive}, we prove that \cite{EPRINT:ChrGunZam23} is robust against adaptive prompting. More recently, \cite{EPRINT:ChrGun24} and \cite{ARXIV:GM24} also give zero-bit, undetectable, sound watermarking schemes with provable robustness. We conjecture that both schemes are robust against adaptive prompting. If \cite{ARXIV:GM24} is adaptively robust, then instantiating our construction with their scheme would yield the first $L$-bit scheme provably robust to deletions without additional assumptions on the model's output. In contrast, the $L$-bit scheme of \cite{EPRINT:ChrGun24} is (non-adaptively) robust to deletions under an assumption that roughly requires the entropy in the model's output to be uniformly distributed across the generation.

Other recent work  proves strong impossibility results 
against motivated and resourced adversaries \cite{ARXIV:ZEFVAB23, ARXIV:PHZS24}. But watermarking can still be useful ---
a little additional overhead can make a bad actor's job significantly harder.
Concurrently, many more applied works have advanced the practice of watermarking language models \cite{LLMWAT1,LLMWAT2,LLMWAT3,LLMWAT4,LLMWAT5,LLMWAT6} and other generative AI models \cite{StegoTraining,Hidden,TreeRing}.

Closely related work develops steganography for language models \cite{CCS:KJGR21,ARXIV:Zam24}. 
Steganography \cite{Hop04,EPRINT:Cachin00} and watermarking for language models are closely related, both embedding a message into a model's outputs. While steganography requires that the existence of the message be hidden, watermarking requires that the message persist even when generated text is modified. 
We highlight the recent work of Zamir \cite{ARXIV:Zam24}, who adapts the approach \cite{EPRINT:ChrGunZam23} to build an $L$-bit steganography scheme, but leaves robustness to future work.

We make extensive use of fingerprinting codes \cite{IEEE:BS98}.
We make black-box use of existing codes, particularly those that are robust to adversarial erasures. Asymptotically-optimal codes are given by \cite{ACM:Tar08,ACM:BKM10}, while \cite{SPRINGER:NFHKWOI07} focus on concrete efficiency.

\subsection{Outline}
The rest of the paper is structured as follows. \secref{sec-prelim} defines \emph{prefix-specifiable language models}, the type of models we watermark, and robust fingerprinting codes, which are used in our constructions. 
\secref{sec-watermark} gives definitions for watermarking language models. We define zero-bit,
$L$-bit, and multi-user watermarking schemes, along with their important properties: 
undetectability, soundness, completeness, and robustness.
Section~\ref{sec:block-by-block} defines \emph{block-by-block} robust watermarks, our framework for constructing and proving black-box watermarking constructions. 
\ifx\cols\one
Appendix~\ref{sec:block-appendix} shows that some existing watermarking schemes are block-by-block.
\fi
\ifx\cols\two
Appendix~\ref{sec:block-appendix} shows that the scheme from \cite{EPRINT:ChrGunZam23} is block-by-block and AEB-robust.\footnote{See the full version of this paper for proofs that other schemes are block-by-block and AEB-robust.}
\fi
\secref{sec-embedding} constructs an $L$-bit watermarking scheme from zero-bit watermarking schemes. We show that the $L$-bit 
scheme inherits undetectability, soundness, and robustness from the underlying 
zero-bit scheme. 
\secref{sec-multiuser} constructs a multi-user watermarking scheme from our $L$-bit scheme and robust fingerprinting codes.
As before, the multi-user scheme preserves the undetectability, soundness, and
robustness guarantees of the underlying watermarking scheme. We additionally
preserve the efficiency and utility of the original zero-bit scheme throughout our
black-box usage.
The above constructions all require the underlying watermarking scheme to be
cryptographically undetectable. In
\ifx\cols\one\secref{sec-detectable}, \fi
\ifx\cols\two\apref{sec-detectable}, \fi
we show that
it is possible to build robust multi-user watermarking schemes out of watermarking
schemes that are not undetectable, albeit with less impressive robustness parameters.

\section{Preliminaries}\label{sec-prelim}

\subsection{Notation}
For $n\in \mathbb{N}$, we denote by $[n]$ the set $\{1,2,\dots,n\}$.
A polynomial, sometimes denoted $\poly(\cdot)$, is some function for which there
is a constant $c$ with $\poly(n) = O(n^c)$.
A function is negligible if it is asymptotically smaller than any inverse
polynomial, and we use $\negl(\cdot)$ to denote an arbitrary negligible function.
Throughout the paper, we use $\lambda$ to denote the security parameter of the
scheme, which specifies the security level of the protocol.
We say an algorithm is \emph{efficient} or $\poly$-bounded if its runtime is
bounded above by some polynomial.

In our pseudocode, we use ``$x\gets y$'' to assign $x$ the value of $y$.
Likewise we use ``$x\getsr y$'' to denote sampling (uniformly)
from $y$ and assigning the value to $x$ (e.g. $x\getsr [n]$ denotes selecting a random
integer between $1$ and $n$ inclusive). The item $y$ may also be a randomized
function, which could just be viewed as running the function with fresh randomness.
We also sometimes use $\true$ and $\false$ in place of $1$ and $0$, respectively, to make
the semantics of an algorithm more clear.

For a finite set $\Sigma$ called an \emph{alphabet}, the set $\Sigma^*$ denotes the set of
all finite-length sequences of elements of $\Sigma$, called \emph{strings}. Language models
are typically defined with respect to an alphabet $\tokSet$, whose elements are
called \emph{tokens}. The length of a string $T$ is denoted $|T|$.
For a string $T$, we define $T[i]$ to be the token
at index $i$ in the string, for $1\le i \le |T|$.
The concatenation of string $S$ to string $T$ is denoted $T\|S$.\footnote{
    As explained in
    Sec.~\ref{sec:prelims:language-models}, we also use $\|$ to
    specify the prefix of a language model's output.
}
We use $\epsilon$ to denote the empty string (in contrast to $\varepsilon$ which
usually denotes a real number). The empty string satisfies $|\epsilon|=0$ and
$\epsilon\|T = T = T\|\epsilon$ for all $T$. We use $\Delta$ to denote the normalized
Hamming distance between strings of the same length. Specifically, for
$S,T\in \Sigma^L$, let $\Delta(S,T) := \frac{1}{L}\cdot \sum_{i=1}^L\1(S[i] \neq T[i])$.

For $0 \leq \delta \le 1$, $L \in \mathbb{N}$ and a string $y \in \bits^L$, we define
the \emph{$\delta$-erasure ball} $B_{\delta}(y)\subseteq \botbits^L$ to be the set of all strings
$z \in \{0,1,\bot\}^L$ where $z_i = \bot$ for at most $\lfloor \delta L \rfloor$
indices $i$, and otherwise $z_i = y_i$.
For $Y \subseteq \bits^L$, we define $B_\delta(Y) = \cup_{y \in Y} B_\delta(y)$.

\subsection{Language models}
\label{sec:prelims:language-models}

A \emph{language model} $\GenModel$ is a randomized algorithm that takes as input a
\emph{prompt} $Q$ and outputs a \emph{generation} $T$ in response.\footnote{
    Note that that our notation differs from prior work. Here, $\GenModel$ outputs a
    string in $\tokSet^*$ (like $\overline{\GenModel}$ in 
    \cite{EPRINT:ChrGunZam23,EPRINT:ChrGun24}), rather than a distribution over 
    the next token.
} Language models are
defined with respect to a set of \emph{tokens} $\tokSet$, where prompts and generations
are strings $Q,T \in \tokSet^*$.
For example, OpenAI's GPT-4 uses a set of 100,000 tokens, with the tokens `water,' `mark',
and `ing' composing the word `watermarking'. The token set contains a \emph{termination token}, which we denote $\bot \in \tokSet$. Every generation $T$ ends in $\bot$, and $\bot$ appears nowhere else in $T$.

We restrict our attention to \emph{prefix-specifiable} models. $\GenModel$ is
prefix-specifiable if for all prompts $Q$ and all prefixes $T$, one can efficiently
compute a new prompt denoted $Q\|T$ such that the distributions (i) $\GenModel(Q||T)$
and (ii) $\GenModel(Q)$
conditioned on $\GenModel(Q)_{\le |T|} = T$ are identical. This is a natural theoretical 
assumption to make, as most popular transformer
models have this property since they depend only on the prior tokens viewed. 
This assumption
was also made by prior work \cite{ICML:KGWKMG23,EPRINT:ChrGunZam23,EPRINT:ChrGun24}.
For a more complete discussion about prefix-specifiable models and their limitations,
we direct readers to \cite[Section 6]{EPRINT:ChrGunZam23}.

\subsection{Fingerprinting codes}
Fingerprinting codes allow for efficient tracing of pirated
digital content. In a canonical example, a film distributor is sending a movie
to several reviewers in advance of a screening. To combat any pre-release leaks,
the distributor embeds a distinct codeword in every copy, where
each letter in the codeword is embedded in a particular scene of the movie. If a reviewer
leaks the movie to the public, the distributor can extract the codeword from the
leaked copy and trace it back to the guilty party. The tracing task becomes more
difficult, however, if
the reviewers can collude: by picking and choosing scenes from each of their copies,
the reviewers may hope to leak a version of the movie that cannot be traced back 
to any colluding party. Fingerprinting codes guarantee that the distributor can
still identify a guilty
party so long there are at most $c$ colluders and they are restricted to picking
scenes from their own $c$ copies. In general, one cannot hope to identify more than one guilty party, as the colluders can always choose to leak a copy without any edits.

\begin{definition}[Fingerprinting codes -- syntax \cite{IEEE:BS98}]\label{def-fp}
    A \emph{fingerprinting code} is a pair of efficient algorithms
    $\FP = (\FPGen, \FPTrace)$ where:
    \begin{itemize}
        \item $\FPGen(1^{\lambda}, n, c, \delta) \to (X, \FPKey)$ is a randomized
        algorithm that takes as input a security parameter $\lambda$,
        a number of users $n$, a maximum number of colluders  $c$,
        and an erasure bound $0 \le \delta < 1$, and outputs a binary
        code $X \in \bits^{n\times L}$ of size $n$ and length $L$, and a tracing key $\FPKey$.
        
        \item $\FPTrace(y, \FPKey) \to S$ is a deterministic algorithm that takes as input any
        string $y \in \botbits^L$ and a tracing key $\FPKey$, and outputs a subset 
        $S \subseteq [n]$ of accused users.
    \end{itemize}
\end{definition}

For $u \in [n]$, we denote the codeword assigned to user $u$ by $X_u$, the $u^{\text{th}}$
row of $X$. For $\ColSet \subseteq [n]$, we denote the codewords assigned to the set
of colluders $\ColSet$ by $X_\ColSet$, the submatrix $(X_u)_{u\in \ColSet}$.
Throughout our paper, we assume that the length $L$ of the fingerprinting code
$\FP$ is a deterministic function of
$\lambda, n,c$ and $\delta$. Most fingerprinting codes, including the ones we
use \cite{ACM:BKM10,SPRINGER:NFHKWOI07}, have this property.

Fingerprinting codes provide a guarantee when a subset of users $\ColSet$ produce a
codeword $y$ by picking and choosing the individual bits of the codewords in $X_\ColSet$.
\emph{Robust} fingerprinting codes also allow a $\delta$-fraction of the bits to be adversarially erased. The \emph{feasible set} contains all strings $y$ that the colluding parties are able to create.

\begin{definition}[Feasible sets]\label{def-fp-feasible}
    For $X \in \bits^{n \times L}$, $\ColSet \subseteq [n]$, the
    \emph{feasible set} is 
    \begin{align*}
        \Feasible(X_{\ColSet}) := \{y \in \bits^L ~:~ \forall i \in [L],\;
        \exists x \in X_\ColSet,\; x[i] = y[i]\}.
    \end{align*}
    In particular, if every $x\in X_\ColSet$ has the same value $b$ at index $i$, then $y_i = b$.
    For $0\le \delta \le 1$, the \emph{$\delta$-feasible ball} is $\Feasible_\delta(X_\ColSet) := B_\delta(\Feasible(X_\ColSet))$.
\end{definition}

The goal of the colluders is to output some feasible word $y \in \Feasible_\delta(X_{\ColSet})$
such that $\FPTrace(y, \FPKey) = \emptyset$ (no user is accused) or that
$\FPTrace(y, \FPKey) \not\subseteq \ColSet$ (an innocent user is accused). The fingerprinting code is secure if this happens with negligible probability.

\ifx\cols\one
\begin{definition}[Fingerprinting codes -- robust security \cite{ACM:BKM10}]\label{def-fp-adv-security}
    A fingerprinting code $\FP$ is \emph{robust} if for all $0\le \delta \le 1$,
    $c\ge 1$, $n\geq c$, $\ColSet \subseteq [n]$ of size $|\ColSet| \le c$, and all efficient adversaries $\A$, the following event occurs with negligible probability:
    \begin{itemize}
        \item $y\in \Feasible_\delta(X_{\ColSet})$, AND
        \quad \texttt{// y is feasible}
        \item $\FPTrace(z, \FPKey) = \emptyset$ OR $\FPTrace(z, \FPKey) \not\subseteq\ColSet$
        \quad \texttt{// no or false accusation}
    \end{itemize}
    in the probability experiment defined by $(X, \FPKey) \gets \FPGen(1^{\lambda}, n, c, \delta)$ and $y \gets \FPAdv(X_{\ColSet})$.
\end{definition}
\fi
\ifx\cols\two
\begin{definition}[Fingerprinting codes -- robust security \cite{ACM:BKM10}]\label{def-fp-adv-security}
    A fingerprinting code $\FP$ is \emph{robust} if for all $0\le \delta \le 1$,
    $c\ge 1$, $n\geq c$, $\ColSet \subseteq [n]$ of size $|\ColSet| \le c$, and all efficient adversaries $\A$, the following event occurs with negligible probability:
    \begin{itemize}
        \item $y\in \Feasible_\delta(X_{\ColSet})$, AND \\
        \quad \texttt{// y is feasible}
        \item $\FPTrace(z, \FPKey) = \emptyset$ OR $\FPTrace(z, \FPKey) \not\subseteq\ColSet$ \\
        \quad \texttt{// no or false accusation}
    \end{itemize}
    in the probability experiment defined by $(X, \FPKey) \gets \FPGen(1^{\lambda}, n, c, \delta)$ and $y \gets \FPAdv(X_{\ColSet})$.
\end{definition}
\fi

Tardos's fingerprinting code \cite{ACM:Tar08} is asymptotically optimal, but not robust to adversarial erasures. The fingerprinting code of Boneh, Kiayias, and Montgomery \cite{ACM:BKM10} is based on the Tardos code 
but is robust. For $n > 1$ users, $c > 1$ colluders, and a
$\delta$ erasure bound, it has length $L = O\Bigl(\lambda (c \log c)^2 \log n / (1 - \delta)\Bigr)$,
but with very large constants.

\subsection{Balls and bins}
Our lossy watermarking schemes allow some of the mark to get erased. Our analysis will use the following lemma about throwing $k$ balls into $L$ bins uniformly at random. The lemma defines $k^*$, the number of balls needed to guarantee that at most $\delta L$ bins are empty except with probability $e^{-\lambda}$. The proof uses standard techniques and is deferred to
\ifx\cols\one Appendix~\ref{sec:lm-bnb-proof}.\fi
\ifx\cols\two the full version of this paper.\fi

\ifx\cols\one
\begin{lemma}\label{lm-bnb}
    For $\lambda,L \in \N$ and $0\le \delta <1$, define 
    \begin{equation}\label{eq:balls-bins}
        k^*(L,\delta) = \min\left\{L\cdot (\ln L + \lambda); \quad L\cdot \ln\left(\frac{1}{\delta - \sqrt{\frac{\lambda + \ln 2}{2L}}} \right) \right\}
    \end{equation}
    Then, after throwing $k\ge k^*(L,\delta)$ balls into $L$ bins, fewer than $\delta L$ bins are empty except with probability at most $e^{-\lambda}$.
\end{lemma}
\fi
\ifx\cols\two
\begin{lemma}\label{lm-bnb}
    For $\lambda,L \in \N$ and $0\le \delta <1$, define 
    \begin{equation}\label{eq:balls-bins}
        k^*(L,\delta) = \min\left\{L\cdot (\ln L + \lambda); \quad L\cdot \ln\left(\frac{1}{d(L,\delta)} \right) \right\}
    \end{equation}
    where $d(L, \delta) = \delta - \sqrt{(\lambda + \ln 2)/2L}$. Then, after throwing $k\ge k^*(L,\delta)$ balls into $L$ bins, fewer than $\delta L$ bins are empty except with probability at most $e^{-\lambda}$.
\end{lemma}
\fi

\section{Watermarks for language models}\label{sec-watermark}

This section defines the syntax and properties of watermarking schemes. 
We introduce \emph{multi-user watermarking schemes} that can trace watermarked text to an individual user, even in the presence of \emph{collusions}.

We require robustness/completeness to hold in the face of
\emph{adaptive prompting}.
All prior works only require completeness and robustness for a {single} text generated
using a {fixed} prompt. But this is not how generative
models are used in reality. Users interact with the models, adaptively
refining the prompts and generations until they are satisfied with the result.
In this setting, existing definitions don't make any guarantees at all!
We instead require that completeness and robustness hold for adversaries
that adaptively query the model, possibly selecting text from many responses.
In \secref{sec:non-adaptive-to-adaptive}, we will show that there is no general 
reduction between non-adaptive and adaptive robustness. In \apref{sec:CGZ}, we show that the scheme of 
Christ, Gunn, Zamir is in fact adaptively robust.

We consider three types of watermarking schemes: zero-bit, $L$-bit, and multi-user
watermarking. An \emph{$L$-bit} watermarking
scheme embeds into generated text an $L$-bit message which can later be extracted
using a secret key. In a \emph{zero-bit} watermarking scheme, text is viewed as
either marked or unmarked, but there is no message to be extracted from marked text.
The name stems from viewing zero-bit watermarking as a special case of $L$-bit
watermarking, with message space containing only one message. 
A \emph{multi-user} watermarking scheme allows for tracing
model-generated text back to the user (or group of users) who 
prompted the generation.

We call the robustness properties of the three types of schemes \emph{robust detection}, \emph{robust extraction}, and \emph{robust tracing}, respectively. We use \emph{robust} when the meaning is clear from context.
Definitions for zero- and $L$-bit watermarking are adapted from existing works, though we relax correctness to require only $(1-\delta) L$ bits of the message to be recovered. Our notation most closely follows \cite{EPRINT:ChrGunZam23,EPRINT:ChrGun24}.

\subsection{Zero- and $L$-bit watermarking syntax}
We give the syntax of  zero-bit watermarking and $L$-bit watermarking schemes. 
The difference is whether the watermark is binary --- marked or unmarked --- or encodes a message. Throughout the paper we focus on the secret key setting, where $\Detect$ and $\Wat$ share a key $\sk$ generated by $\Setup$.

For zero-bit watermarking, the algorithm $\Wat$ is the
watermarked version of the language model, with the same inputs and outputs.
The $\Detect$ algorithm indicates whether text is considered marked or unmarked.

\begin{definition}[Zero-bit watermarking -- Syntax]\label{def-wat}
A \emph{zero-bit watermarking scheme} for a
language model $\GenModel$ over $\tokSet$ is a
tuple of efficient algorithms $\wat=(\Setup,\Wat,\Detect)$ where:
\begin{itemize}
    \item $\Setup(1^{\lambda}) \to \sk$ is a randomized algorithm that takes
    a security parameter $\lambda$ as input and outputs a secret key $\sk$.
    \item $\Wat_{\sk}(\prompt) \to T$ is a keyed randomized algorithm that
    takes a prompt $\prompt$ as input and outputs a string
    $T\in \tokSet^{*}$.
    \item $\Detect_{\sk}(T) \to b$ is a keyed deterministic
    algorithm that takes a string 
    $T\in \tokSet^{*}$ as input and outputs a bit $b \in \bits$.
\end{itemize}
\end{definition}

An $L$-bit watermarking scheme allows a message $m \in \{0,1\}^L$
to be embedded in generated text and later extracted. 
The syntax is the natural
generalization of the above, though we rename $\Detect$ to $\Extract$,
reflecting its new semantics.
Observe that a $1$-bit scheme can be used to construct a
zero-bit scheme by taking $\Wat_\sk(\prompt) :=\Encode_\sk(1,\prompt)$ and
$\Detect_{\sk}(T) := \1\bigl(\Extract_{\sk}(T) = 1 \bigr)$. Likewise for $L>1$.
\begin{definition}[Watermarking syntax -- $L$-bit]\label{def-msg}
An \emph{$L$-bit watermarking scheme} for a language model $\GenModel$ over $\tokSet$ is a
tuple of efficient algorithms $\Msg=(\MsgSetup,\Encode,\Extract)$ where:
\begin{itemize}
    \item $\MsgSetup(1^{\lambda}) \to \sk$ is a randomized algorithm that takes
    a security parameter $\lambda$ as input and outputs a secret key $\sk$.
    \item $\Encode_{\sk}(\msg,\prompt) \to T$ is a keyed randomized algorithm that
    takes as input a message $\msg\in \bits^L$ and a prompt $\prompt \in \tokSet^*$
    and outputs a string $T \in \tokSet^*$.
    \item $\Extract_{\sk}(T) \to \hat{\msg}$ is a keyed deterministic
    algorithm that takes a string
    $T\in \tokSet^{*}$ as input and outputs a message $\hat{\msg} \in \botbits^L$.
\end{itemize}
\end{definition}

\subsection{Properties of watermarking schemes}

Watermarking schemes must have four properties. 
First is \emph{soundness}: marks must not be falsely detected (resp.\ extracted, traced) in text not generated by the watermarked model (low Type I error).
Second, watermarking should not noticeably degrade the quality of the model's outputs. The strongest version of this is \emph{undetectability}: the marked and unmarked models are indistinguishable.
Third is \emph{completeness}: the marks are detectable in verbatim outputs from the watermarked model (low Type II error).
Fourth is \emph{robustness} to edits: marks are detectable even after marked 
text is edited, say by deleting or rearranging phrases, or by pasting an excerpt into an unmarked document.
All the existing schemes offer some degree of all four guarantees, though the formalisms and assumptions vary greatly.

We now define these properties for $L$-bit watermarking. 
We defer to \ifx\cols\one\apref{ap:definitions} \fi\ifx\cols\two the full version of this paper \fi definitions of
undetectability, soundness, completeness, and robust detection for zero-bit watermarking,
as they are special cases of the $L$-bit definitions.
We require negligible probabilities of failure
against arbitrary efficient adversaries, in part because such strong guarantees lend
themselves to black-box reductions. 
Some prior works target weaker guarantees, in particular by
relaxing undetectability. 
In
\ifx\cols\one\secref{sec-detectable}, \fi
\ifx\cols\two\apref{sec-detectable}, \fi
we give some limited results
for schemes that are not undetectable.

\medskip
\emph{Soundness} requires that $\Extract$ returns $\bot^L$ except with negligible probability for all strings $T$ containing at most poly-many tokens.

\begin{definition}[Soundness -- $L$-bit]\label{def-soundness}
An $L$-bit watermarking scheme $\Msg=(\MsgSetup,\Encode,\Extract)$
is \emph{sound} if for all polynomials $\poly$ and all
strings $T\in \tokSet^{*}$ of length $|T| \le \poly(\lambda)$,
\[
    \Pr_{\sk \getsr \MsgSetup(1^{\lambda})}[\Extract_{\sk}(T) \ne \bot^L]
    < \negl(\lambda).
\]
\end{definition}

\emph{Undetectability}, a property pioneered and first realized by \cite{EPRINT:ChrGunZam23}\ifx\cols\one\ (\defref{def-zero-undetectability})\fi,
requires that outputs of the watermarked model
$\Wat$ must be computationally indistinguishable from outputs of the
underlying model $\GenModel$. In particular, this implies that any
computationally-checkable utility guarantees of $\GenModel$ also hold for $\Wat$.

\ifx\cols\two
\begin{definition}[Undetectability -- $L$-bit]\label{def-undetectability}
    Define the oracle $\GenModel'(\msg, \prompt) := \GenModel(\prompt)$.
    An $L$-bit watermarking scheme $\Msg=(\MsgSetup,\Encode,\Extract)$ for $\GenModel$ is
    \emph{undetectable} if for all efficient adversaries $\A$,
    \[
        \left|
        \Pr[\A^{\GenModel'(\cdot,\cdot)}(1^{\lambda}) = 1] -
        \Pr_{\sk \getsr \MsgSetup(1^{\lambda})}[
            \A^{\Encode_{\sk}(\cdot,\cdot)}(1^{\lambda}) = 1]
        \right|
    \]
    is at most $\negl(\lambda)$.
\end{definition}
\fi
\ifx\cols\one
\begin{definition}[Undetectability -- $L$-bit]\label{def-undetectability}
    Define the oracle
    $\GenModel'(\msg, \prompt) := \GenModel(\prompt)$.
    An $L$-bit watermarking scheme
    $\Msg=(\MsgSetup,\Encode,\Extract)$ for $\GenModel$ is
    \emph{undetectable} if for all efficient adversaries $\A$,
    \[
        \left|
        \Pr[\A^{\GenModel'(\cdot,\cdot)}(1^{\lambda}) = 1] -
        \Pr_{\sk \getsr \MsgSetup(1^{\lambda})}[
            \A^{\Encode_{\sk}(\cdot,\cdot)}(1^{\lambda}) = 1]
        \right| < \negl(\lambda).
    \]
\end{definition}
\fi

\paragraph{Completeness and robustness}
Formally defining completeness and robustness requires some care. We focus on robustness, as completeness is a special case.

Intuitively, robustness should guarantee something like the following: if $T \getsr \Wat_\sk(m,\prompt)$ and $\hT \approx T$, then $\Extract_\sk(\hT) = m$. But this requirement is too strong. Soundness requires that any fixed string, say the text of The Gettysburg Address, is unmarked with high probability. But then so must the text $T$ generated in response to the query
\ifx\cols\one
$\prompt = \texttt{What is the text of The Gettysburg Address?}$,
\fi
\ifx\cols\two
$\prompt = \texttt{What is the text of}\allowbreak \texttt{The Gettysburg Address?}$,
\fi assuming the model answers correctly.

To deal with this issue, existing works impose a requirement on the \emph{entropy} of $T$, for various notions of entropy.
In light of the above, robustness should guarantee something like the following. If $T \getsr \Wat_\sk(m,\prompt)$, then one of the following holds with high probability: (i) $T$ lacks sufficient entropy; (ii) $\hT \not\approx T$; or (iii) $\Extract_\sk(\hT) = m$.
Instantiating the definition requires specifying what exactly conditions (i) and (ii) mean.

Our definition of robustness is agnostic to this choice (though our constructions require additional structure, see Sec.~\ref{sec:block-by-block}). Instead, we define robustness relative to a generic \emph{robustness condition} $R$ which evaluates to 1 when both (i) $T$ has sufficient entropy, and (ii) $\hT \approx T$. Note that $R$ only defines a \emph{sufficient} condition for extracting the mark. It may sometimes hold that $\Extract_{\sk}(\hT) = m$
even when $R=0$.

Our definition of robust extraction is also parameterized by a number $0\le \delta \le 1$.
A scheme is \emph{lossy} if $\delta > 0$, and \emph{lossless} if $\delta = 0$.
We omit $\delta$ for lossless schemes, writing \emph{$R$-robust/complete}. Recall that
$B_\delta(m)$ is the set of all strings $\hat{m} \in \botbits^L$ that agree with
$m$ except for at most $\lfloor \delta L \rfloor$ indices.

\begin{definition}[Robustness condition]
    A \emph{robustness condition} is a (deterministic, efficient) function
    $R:(\lambda, (\prompt_i)_i, (T_i)_i, \hat{T}) \mapsto b,$
    where $\lambda \in \mathbb{N}$; $\prompt_i, T_i, \hat{T} \in \tokSet^*$
    for all $i$; and $b \in \bits$.
\end{definition}

\begin{definition}[$(\delta,R)$-Robust extraction -- $L$-bit, adaptive]\label{def-robustness}
    An $L$-bit watermarking scheme $\Msg=(\MsgSetup,\allowbreak\Encode,\Extract)$
    is \emph{adaptively $(\delta,R)$-robustly extractable} with respect to
    robustness condition $R$
    if for all messages $\msg \in \bits^L$, and all efficient adversaries $\A$,
    the following event \fail occurs with negligible probability:
    \begin{itemize}
        \item $R\left(\lambda,(\prompt_i)_i,(T_i)_i,\hat{T}\right) = 1$, AND 
        \ifx\cols\one\quad\fi\ifx\cols\two\\\fi \texttt{// robustness condition holds}
        \item $\hat{\msg} \not\in B_{\delta}(\msg)$ 
        \ifx\cols\one\quad\fi\ifx\cols\two\\\fi \texttt{// the mark is corrupted}
    \end{itemize}
    in the probability experiment defined by 
    \begin{itemize}
        \item $\sk \getsr \MsgSetup(1^{\lambda})$
        \item $\hat{T} \gets \A^{\Wat_\sk(\msg,\cdot)}(1^\lambda)$, denoting by
        $(\prompt_i)_i$ and $(T_i)_i$ the sequence of inputs and outputs of the oracle
        \item $\hat{\msg} \gets \Extract_{\sk}(\hat{T})$.
    \end{itemize}
\end{definition}
\noindent
\emph{Completeness} is almost identical, with the additional clause ``AND $\hat{T} \in (T_i)_i$'' added to \fail (see \apref{ap:definitions}).

\subsection{Multi-user watermarks}

We now define a {multi-user watermarking scheme}. We consider
a watermarking scheme which is deployed for a set of users $\Users$ and
generalize the notation for zero-bit watermarking. Our definition has
three functions that nearly match the syntax and semantics of zero-bit
watermarking. The only difference among the first three functions is that
$\Wat$ takes as input both a user and a prompt.
The new functionality of a multi-user watermarking scheme is $\Trace$, which
given some text $\hT$ can output the user that produced $\hT$.

Our syntax tracks the following intended usage. The watermarker initially
sets up their system by generating a secret key $\sk$ with $\Setup$.
For each user $u$, they provide oracle access to $\Wat_{\sk}(u,\cdot)$, fixing 
the first input $u$. This models, say, a signed-in user interacting 
with ChatGPT.

To detect only the {presence} of a watermark in a candidate text $\hT$,
the watermarker can run $\Detect_{\sk}(\hT)$ exactly as in a zero-bit 
watermarking scheme. If they want to determine {which user(s)} the
text belongs to, they can run $\Trace_{\sk}(\hT)$, which outputs
the (possibly empty) set of users whose watermarked outputs generated $\hT$.

\begin{definition}[Multi-user watermarking]\label{def-multi-wat}
    A \emph{multi-user watermarking scheme} for a
    model $\GenModel$ over a token alphabet $\tokSet$ and a set of users $\Users$
    is a tuple of efficient algorithms \ifx\cols\two\\\fi
    $\wat=(\Setup,\Wat,\Detect,\Trace)$ where:
    \begin{itemize}
        \item $\Setup(1^{\lambda}) \to \sk$ is a randomized algorithm that takes
        a security parameter $\lambda$ as input and outputs a secret key $\sk$.
        \item $\Wat_{\sk}(u, \prompt) \to T$ is a keyed randomized algorithm that
        takes as input a user $u \in \Users$ and a prompt $\prompt\in \tokSet^*$ and generates a
        response string $T\in \tokSet^{*}$.
        \item $\Detect_{\sk}(T) \to b$ is a keyed deterministic algorithm that takes a
        string $T \in \tokSet^*$ as input and outputs a bit $b \in \bits$.
        \item $\Trace_{\sk}(T) \to S$ is a keyed deterministic
        algorithm that takes a string
        $T\in \tokSet^{*}$ and outputs a set of accused users $S \subseteq \Users$.
    \end{itemize}
\end{definition}

Notice that detection could simply check if $\Trace$ outputs
at least one user, so our definition could remove the explicit function $\Detect$. However, detection alone may be much quicker or require less 
generated text compared to tracing. Both are true in our construction (\secref{sec-multiuser}): detection only requires checking if a partially-erased watermark is not empty, where extracting the partial mark takes $O(\log|\Users|)$ time. Tracing amounts to checking the partial watermark against each user one-by-one, and may fail if too much of the mark is erased.

Although it makes sense to separate these forms of detection,
we also hope that the two are \emph{consistent}. In other words, if there is
no detected watermark in the text, then we don't want to accuse any users of 
generating it. Alternatively, if we can find a user, then we should also detect 
that there is watermark.

\begin{definition}[Consistency]\label{def-multi-consistency}
    We say a multi-user watermarking scheme $\wat=(\Setup,\Wat,\Detect,\Trace)$
    is \emph{consistent} if for all $T \in \tokSet^*$ and all keys $\sk$,
    \begin{align*}
        \Detect_{\sk}(T) = 0 \implies \Trace_{\sk}(T) = \emptyset.
    \end{align*}
\end{definition}

As before, a multi-user watermarking scheme should be undetectable, sound,
and robust.\footnote{We do not define completeness for multi-user watermarking,
as it is just a special case of robustness, which we prove later in the paper.}
The first two properties are straightforward generalizations of
definitions for $L$-bit watermarking, and we defer them to 
\ifx\cols\one\apref{ap:multi-user-defs}.\fi\ifx\cols\two the full version of the paper.\fi

As for robustness, we can ask for both \emph{robust detection} and \emph{robust tracing}.
By consistency, $R$-robust tracing implies $R$-robust detection.
Motivated by the multi-user setting, we additionally consider robustness against \emph{collusions}. If a group of users $u_1,\ldots,u_c$,
each interacting with their own oracle, collude to produce some text
$\hT$, then we still wish that $\Trace$ can find one or more of the
users. Our definition below captures this idea by allowing the adversary to query the model as $c$ distinct users.

\ifx\cols\one
\begin{definition}[$R$-Robust tracing against $c$-collusions -- multi-user, adaptive]\label{def-multi-robustness}
    A multi-user watermarking scheme
    $\wat=(\Setup,\Wat,\Detect,\Trace)$ is
    \emph{$R$-robustly traceable against $c$-collusions} with respect to the robustness condition
    $R$ and collusion bound $c > 1$ if for all $C\subseteq \Users$ of
    size at most $|C|\le c$ and all efficient adversaries $\A$,
    the following event \fail occurs with negligible probability:
    \begin{itemize}
        \item $\forall i, u_i \in C$, AND 
        \quad \texttt{// only users in C collude}
        \item $R\left(\lambda,(\prompt_i)_i, (T_i)_i, \hat{T}\right) = 1$, AND 
        \quad \texttt{// robustness condition passes}
        \item $\bigl(S = \emptyset ~\lor~ S \not\subseteq C\bigr)$ 
        \quad \texttt{// no or false accusation}
    \end{itemize}
    in the probability experiment defined by 
    \begin{itemize}
        \item $\sk \getsr \MsgSetup(1^{\lambda})$
        \item $\hat{T} \gets \A^{\Wat_\sk(\cdot,\cdot)}(1^\lambda)$, denoting by
        $(u_i,\prompt_i)_i$ and $(T_i)_i$ the sequence of inputs and outputs of the oracle
        \item $S \gets \Extract_\sk(\hat{T})$.
    \end{itemize}
\end{definition}
\fi
\ifx\cols\two
\begin{definition}[$R$-Robustness against $c$-collusions -- multi-user, adaptive]\label{def-multi-robustness}
    A multi-user watermarking scheme 
    $\wat=(\Setup,\Wat,\Detect,\Trace)$ is
    \emph{$R$-robust against $c$-collusions} with respect to the robustness condition
    $R$ and collusion bound $c > 1$ if for all $C\subseteq \Users$ of
    size at most $|C|\le c$ and all efficient adversaries $\A$,
    the following event \fail occurs with negligible probability:

    \ifx\cols\one
    \begin{itemize}
        \item $\forall i, u_i \in C$, AND 
        \quad \texttt{// only users in C collude}
        \item $R\left(\lambda,(\prompt_i)_i, (T_i)_i, \hat{T}\right) = 1$, AND 
        \quad \texttt{// robustness \condition\ passes}
        \item $\bigl(S = \emptyset ~\lor~ S \not\subseteq C\bigr)$ 
        \quad \texttt{// no or false accusation}
    \end{itemize}
    \fi
    \ifx\cols\two
    \begin{itemize}
        \item $\forall i, u_i \in C$, AND \\
        \quad \texttt{// only users in C collude}
        \item $R\left(\lambda,(\prompt_i)_i, (T_i)_i, \hat{T}\right) = 1$, AND \\
        \quad \texttt{// robustness condition passes}
        \item $\bigl(S = \emptyset ~\lor~ S \not\subseteq C\bigr)$ \\
        \quad \texttt{// no or false accusation}
    \end{itemize}
    \fi
    
    in the probability experiment defined by 
    \begin{itemize}
        \item $\sk \getsr \MsgSetup(1^{\lambda})$
        \item $\hat{T} \gets \A^{\Wat_\sk(\cdot,\cdot)}(1^\lambda)$, denoting by
        $(u_i,\prompt_i)_i$ and $(T_i)_i$ the sequence of inputs and outputs of the oracle
        \item $S \gets \Extract_\sk(\hat{T})$.
    \end{itemize}

    We also say a scheme satisfying this definition is \emph{$R$-robustly
    traceable} (against $c$-collusions).
\end{definition}
\fi

\section{Block-by-block watermarks}\label{sec:block-by-block}

This section introduces the syntax for \emph{block-by-block} watermarking schemes
and \emph{AEB-robustness conditions}. These abstractions provide
a unified way to describe the robustness guarantees
of existing schemes that enables black-box reductions.

Informally, a block-by-block scheme views a generation $T$ as a sequence of 
\emph{blocks} $\Blocks(T)$, each of which has high-enough entropy. AEB-robustness guarantees that candidate text is watermarked whenever it 
\emph{{\bf A}pproximates {\bf E}nough {\bf B}locks} of model-generated text (see \figref{fig:blocks}).
In general, these blocks do not need to be copied verbatim and in fact only need to 
be \emph{approximated} in the candidate text.  The 
AEB-robustness condition $R_1$ requires one block to be approximated. $R_1$ is satisfied by $\hT$
if there exists a single block $\block \in \cup_i \Blocks(T_i)$ that is
approximated by some substring $\hat{\block}$ of $\hT$. For $k\ge 1$, the 
AEB-robustness condition $R_k$ requires approximating $k$ blocks. $R_k$ checks whether 
at least $k$ distinct blocks $\block_j \in \cup_i \Blocks(T_i)$ are approximated
by substrings $\hat{\block}_j$ of $\hT$. In our constructions, we will only require that
the underlying scheme is $R_1$-robustly detectable.
Using this underlying scheme, our multi-user construction (\secref{sec-multiuser})
will be $R_1$-robustly detectable and $R_k$-robustly traceable.

\begin{figure}[h!]
    \centering
    \ifx\cols\one
    \includegraphics[width=0.6\columnwidth]{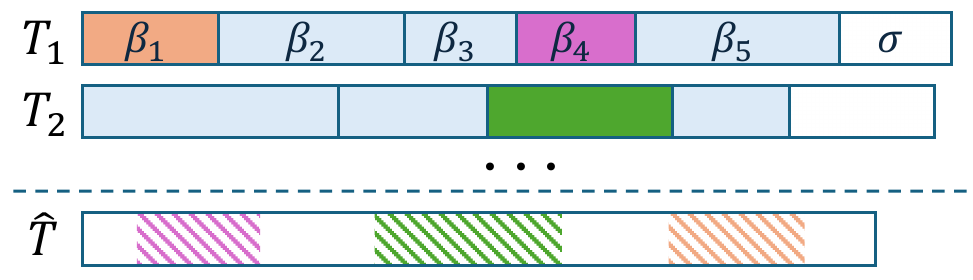}
    \fi
    \ifx\cols\two
    \includegraphics[width=0.9\columnwidth]{images/blocks.pdf}
    \fi
    \caption{Visualization of a string $\hT$ containing three approximate blocks
    from original generations $T_1$ and $T_2$. This $\hT$ would satisfy 
    the $R_3(\lambda, (\prompt_i)_i, (T_i)_i, \hT)$ robustness condition.
    }
    \label{fig:blocks}
\end{figure}

We now formalize these notions. To define the AEB-robustness condition
$R_1$, we need a binary function $\isBlock$ (which specifies whether a given
substring $\block \in T$ constitutes a block) and a binary relation on
strings $\bumpeq$ (which specifies when $\hat{\block}$ approximates $\block$)
denoted by $\hat{\block} \bumpeq \block$.
Note that all the functions defined below may take the security parameter
$\lambda$ as an additional input. We omit it to reduce notational clutter.

As already discussed, a useful measure of the entropy of a generation must be with
respect to the underlying query (and the language model). Thus we define $\isBlock$
to take two strings as input, the generation $T$ and the prompt $\prompt$.

\begin{definition}[Block]\label{def:is-block}
    Let $\isBlock:\tokSet^* \times \tokSet^* \to \bits$ be a (deterministic, efficient) function.
    For $\prompt \in \tokSet^*$, a \emph{block} with respect to $\prompt$ is a
    string $T \in \tokSet^*$ such that  $\isBlock(T;\prompt) = 1$.
    A block is \emph{minimal} if no prefix is a block.
\end{definition}

\begin{definition}[Block-by-block watermarking]
    \label{def:block-by-block-embedding}
    A \emph{block-by-block} watermarking scheme $\Msg$ is a watermarking scheme
    that  it has an additional block algorithm
    $\isBlock:\tokSet^* \times \tokSet^* \to \bits$,
    which may depend on $\lambda$ and $\GenModel$, but not the secret key $\sk$.
\end{definition}

It is easy to check that any generation of a prefix-specifiable model can be uniquely
parsed into a sequence of minimal blocks (possibly with a non-block suffix), which
we denote $\Blocks(T;\prompt)$.

\begin{definition}[$\Blocks(T;\prompt)$]\label{def:blocks}
    Let $\prompt, T \in \tokSet^*$ be strings and $\isBlock$ as above.
    We define $\Blocks(T;\prompt)$ to be the unique sequence
    $(\block_1, \block_2,\allowbreak \dots, \block_B)$ such that:
        (i) $T = \block_1 \| \block_2 \| \cdots \| \block_B\| \sigma$ for some string $\sigma$;
        (ii) $\block_i$ is a minimal block with respect to
        $\prompt\|\block_1\|\cdots\|\block_{i-1}$ for all $i\in [B]$; and
        (iii) no prefix of $\sigma$ is a block with respect to
        $\prompt\|\block_1\|\cdots\|\block_{B}$.
    We call $\Blocks(T;\prompt)$ the \emph{blocks of $T$ with respect to $\prompt$}.
\end{definition}

AEB-robustness conditions involve counting the number of
substrings of $\hat{T}$ that approximate distinct blocks in a generation $T$,
where the approximation is specified by the binary relation $\bumpeq$, which we typically
suppress throughout the paper for ease of notation.
The function $\numBlocks$ returns that count.

\begin{definition}[$\numBlocks$]\label{def:num-blocks}
    Let $\isBlock$ and $\Blocks$ as above, and let $\bumpeq:\tokSet^* \times \tokSet^* \to \bits$ be a binary relation on strings. The function $\numBlocks:(\hat{T};\prompt, T)\mapsto n$ on input strings $\hat{T}, \prompt, T \in \tokSet^*$ is defined to be the maximum $n\ge 0$ for which there
    exist substrings $\hat{\block}_1,\dots,\hat{\block}_n \in \hat{T}$ and  distinct
    blocks $\block_1,\dots,\block_n \in \Blocks( T;\prompt)$ such that
    $\hat{\block}_i \bumpeq \block_i$ for all $i \in [n]$.
\end{definition}

$R_k$ holds when
at least $k$ substrings are present in a string $\hT$ that approximate
(according to some $\bumpeq$) blocks contained in generations $(T_i)_i$
with respect to prompts $(Q_i)_i$.

\begin{definition}[AEB-robustness condition $R_k$]\label{def:R-k}
    Fix $\numBlocks$ as above. For $k \ge 1$, the \emph{AEB-robustness condition} $R_k$ is
    \[
        R_k(\lambda,(\prompt_i)_i,(T_i)_i,\hat{T}) := \1\biggl(\sum_i \numBlocks(\hat{T};\prompt_i,T_i) \ge k \biggr).
    \]
\end{definition}

We can easily compare AEB-robust watermarking schemes to one another.
All else equal, a watermarking scheme is more robust if we relax any of ``approximate'', ``enough'', or ``blocks.''
We may relax $\bumpeq$, considering more strings as approximations of a given block. We may reduce $k$, requiring fewer approximate blocks. Or we may relax $\isBlock$, treating shorter strings as blocks.

Notice that we can trivially view any watermarking scheme as block-by-block and AEB-robust.
Specifically, by defining entire generations as a single block and taking the 
equality relation, any complete scheme would be $R_1$ robust. However, this perspective 
is very weak. As discussed in \apref{sec:block-appendix}, \ifx\cols\two we can show that \fi many existing watermarking schemes are block-by-block 
and AEB-robust for non-trivial block functions and relations.

\subsection{Non-adaptive versus adaptive robustness}
\label{sec:non-adaptive-to-adaptive}
Existing watermarking schemes provide robustness guarantees that are
\emph{non-adaptive} (\defref{def-robustness-zero-nonadaptive}): 
they only hold for a single generation $T$ produced in response to any single prompt $\prompt$.
In fact, the zero-bit scheme from \cite{EPRINT:ChrGunZam23} is provably
$R_1$-robust against adaptive prompting, for the equality relation (see \apref{sec:CGZ}).
The primary theorems in \cite{EPRINT:ChrGunZam23} are proven for non-adaptive robustness;
we extend the proof to the adaptive setting via
Lemmas \ref{lem:CGZ-direct} and \ref{lem:direct-to-adapt}.
We conjecture that \cite{EPRINT:ChrGun24} is also adaptively robust, but we do not prove it here.

Non-adaptive robustness does not imply adaptive robustness, even for undetectable schemes.
The following counter-example was suggested by Miranda Christ and Sam Gunn, who pointed
out a technical flaw in an earlier version of this paper.
Given any sound, non-adaptively $R_1$-robust zero-bit watermarking scheme $\wat=(\Setup,\Wat,\Detect)$,
we can construct a new scheme $\wat'=(\Setup,\Wat',\Detect)$ where
\[
    \Wat_{\sk}'(\prompt) =    \begin{cases}
                            \Wat_{\sk}(\prompt) &\text{if } \Detect_{\sk}(\prompt)=\false\\
                            \GenModel(\prompt) &\text{if } \Detect_{\sk}(\prompt)=\true
                        \end{cases}
\]
$\wat'$ is non-adaptively robust but not adaptively robust (or even adaptively complete). In particular, an adversary can just query $\Wat'$ once to get a
watermarked text, and then feed the response back to $\Wat'$ to get an output that is not watermarked.

\section{Zero-bit to $L$-bit watermarks}\label{sec-embedding}

In this section, we construct $L$-bit watermarking schemes from block-by-block,
zero-bit watermarking schemes.
Namely, if $\wat'$ is a zero-bit watermarking scheme that is undetectable, sound
and $R_1$-robust, the resulting $\wat$ is an $L$-bit scheme that is undetectable,
sound, and $R_{k}$-robust, with $k = O(L\lambda)$.
Our construction is black-box with respect to the robustness condition of $\wat'$,
requiring only that it is $R_1$-robust for some underlying functions $\isBlock$
and $\bumpeq$. The resulting scheme is $R_k$-robust where $R$ is induced by the same $\isBlock$ and $\bumpeq$ as $R_1$.

Section~\ref{sec:l-bit:construction} gives the construction and Section~\ref{sec-l-bit-robust}
proves robustness. Section~\ref{sec:l-bit:on-k} analyzes the resulting value of $k$.
It suffices to set $k = O(L\lambda)$ to losslessly recover the message. For
long messages ($L = \Omega(\lambda)$), recovering a constant $(1-\delta)$ fraction
of the messages requires only $O(L)$ blocks.

\subsection{Constructing $L$-bit watermarks}
\label{sec:l-bit:construction}

Our construction is given in \figref{fig-embedding}. Let
$\wat' \allowbreak = (\Setup', \Wat', \Detect')$ be a zero-bit watermarking scheme.
We construct an $L$-bit scheme $\wat = (\Setup, \allowbreak \Wat, \Extract)$ as follows. 
The secret key $\sk$ consists of $2L$ zero-bit keys $k_{i,b}\gets \Setup'(1^\lambda)$
sampled independently, for all $i \in [L]$ and $b\in \bits$.
The keys $k_{i,0}$ and and $k_{i,1}$ are used to embed the $i$th bit of a
message $\msg \in \bits^L$.
To do so, the watermarked model $\Wat_\sk(\msg,\prompt)$ repeatedly samples a new
block of text by calling $\Wat'$ with key $k_{i,\msg[i]}$ for uniformly random
index $i \getsr [L]$.
What constitutes a block of text is determined by $\wat'$, which is assumed to
be a block-by-block scheme.
The generated block is added to the current generation and the process is repeated. 
The loop exits when the call to $\Wat'$ fails to generate a full block of text.
To extract the message from $\hat{T}$, the algorithm $\Extract_\sk(\hat{T})$
runs $\Detect'(\hat{T})$ algorithm using every key $k_{i,b}$. The $i$th bit of
extracted message is determined by which of the keys $k_{i,0}$ or $k_{i,1}$ a
(zero-bit) mark was detected.

\begin{figure}
\begin{center}
\fbox{
\begin{minipage}[t]{4cm}
    \begin{tabbing}
        12\=123\=123\=123\=123\=\kill
        \underline{$\MsgSetup(1^{\lambda})$}\\
        For $i=1,\ldots,L$:\\
        \> $k_{i,0} \getsr \MsgSetup'(1^{\lambda})$\\
        \> $k_{i,1} \getsr \MsgSetup'(1^{\lambda})$\\
        Return $\sk = (k_{i,0},k_{i,1})_{i=1}^{L}$
    \end{tabbing}

    \begin{tabbing}
        12\=123\=123\=123\=123\=\kill
        \underline{$\Extract_{\sk}(\hat{T})$}\\
        For $k_{i,b} \in \sk$:\\
        \> $z_{i,b} \gets \Detect'_{k_{i,b}}(\hat{T})$\\
        For $i=1,\ldots,L$:\\
        \> $\hat{m}_i \gets
            \begin{cases}
                \bot &\text{if } z_{i,0} = z_{i,1} = \false\\
                0 &\text{if } z_{i,0} = \true\\ %
                1 &\text{otherwise}
            \end{cases}$\\
        Return $\hat{m} = \hat{m}_1 \hat{m}_2 \dots \hat{m}_L$
    \end{tabbing}
\end{minipage}
}
\fbox{
\begin{minipage}[t]{4cm}
    \begin{tabbing}
        12\=123\=123\=123\=123\=\kill
        \underline{$\Encode_{\sk}(\msg,\prompt)$}\\
        $T \gets \epsilon$\\
        While $\true$:\\
        \> $k \getsr \{k_{i,b} : i \in [L],b = \msg[i]\}$\\
        \> $T' \getsr \Encode'_{k}(\prompt\|T)$\\
        \> $(\block_1, \block_2,\dots) \gets \Blocks(T';\prompt \|T)$\\
        \> If $\block_1 = \bot$:~~Exit loop\\
        \> $T \gets T \| \block_1$\\
        $T \gets T \| T'$\\
        Return $T$
    \end{tabbing}
\end{minipage}
}
\caption{Pseudocode for $L$-bit watermarking scheme
of $\Msg=(\MsgSetup,\allowbreak \Encode,\Extract)$ from a
block-by-block zero-bit watermarking scheme
$\wat'=(\Setup',\Wat',\Detect')$.}
\label{fig-embedding}
\end{center}
\end{figure}

We now show that our scheme is  undetectable and sound. Robustness (which implies completeness) is more involved, and is deferred to Section~\ref{sec-l-bit-robust}. In the following, let $\wat'$ be a zero-bit watermarking scheme and $\wat$ be the $L$-bit watermarking scheme described in Figure~\ref{fig-embedding}.

\begin{clm}[$\Msg$ is undetectable]\label{clm-msg-undetectable}
If $\GenModel$ is prefix-specifiable and $\Msg'$ is undetectable,
then $\wat$ is undetectable.
\end{clm}
\begin{proof}
By the undetectability of $\wat'$, one can replace every call to $\Wat'_k(Q\|T)$
with $\GenModel(Q\|T)$, with negligible effect on the adversary's output distribution.
The result is a modified version of $\wat$ that generates the response to
prompt $\prompt$ by iteratively calling $\GenModel(\prompt \| T)$ with $T$ the prefix
generated so far. Because $\GenModel$ is prefix-specifiable, this is the same distribution as $\GenModel(\prompt)$.
\end{proof}

\begin{clm}[$\Msg$ is sound]\label{clm-msg-sound}
    Let $\Msg'$ be a block-by-block zero-bit watermarking scheme. If $\Msg'$ is 
    sound then $\Msg$ is sound.
\end{clm}
\begin{proof}
The soundness of $\Msg$ follows immediately from the soundness of
$\Msg' = (\MsgSetup',\Encode',\Extract')$. Specifically,  observe that
for a given $\sk = (k_{i,0},k_{i,1})_{i=1}^{L}$,
$\Extract_{\sk}(T) \ne \bot^L$ only when\\
$\Extract'_{k_{i,b}} \ne \bot$
for some $i\in [L]$ and $b\in \bits$.
So, for every polynomial $p$, $\lambda$, and $T$ with $|T| \le p(\lambda)$,

\ifx\cols\one
\begin{align*}
    \Pr_{\sk \getsr \MsgSetup(1^{\lambda})}&[\Extract_{\sk}(T) \ne (\bot)^L]
    \le 2L \Pr_{k \gets \MsgSetup'(1^{\lambda})}[\Extract_{k}'(T) \ne \bot]
    < \negl(\lambda),
\end{align*}
\fi
\ifx\cols\two
\begin{align*}
    \Pr_{\sk \getsr \MsgSetup(1^{\lambda})}&[\Extract_{\sk}(T) \ne (\bot)^L]\\
    \le \; &2L \Pr_{k \gets \MsgSetup'(1^{\lambda})}[\Extract_{k}'(T) \ne \bot]
    < \negl(\lambda),
\end{align*}
\fi

via a union bound over every call to $\Extract'$.
\end{proof}

\subsection{Robustness of $L$-bit watermarks}
\label{sec-l-bit-robust}
In this section, we prove robustness of our scheme (Theorem~\ref{thm-msg-undetectable}).
Namely,  our $L$-bit watermarking scheme is $(\delta,R_k)$-robust for any
$k\ge k^*(L,\delta)$ as defined in \lemref{lm-bnb}. The parameter
$k^* \le L \ln L + L\lambda = O(L\lambda)$ for all $\delta$
(see \secref{sec:l-bit:on-k} for discussion). Note that $k^*$ is independent
of the total amount of generated text seen by the adversary!

We actually prove a stronger result, \lemref{lem-msg-robust}, of which
Theorems~\ref{thm-msg-undetectable} and~\ref{thm-undetectable} are both corollaries.
In the theorems and proofs, we use the notation from \defref{def-fp-feasible}
to denote feasible sets.

\begin{lemma}[$\Msg$ recovers lossy descendants]\label{lem-msg-robust}
    Let $L \in \N$, $0\le \delta < 1$, $k \ge k^*(L, \delta)$, and let
    $M\subseteq \bits^L$.
    Suppose $\wat'$ is a block-by-block zero-bit watermarking scheme that is 
    undetectable, sound, and $R_1$-robustly detectable.
    Let $\Msg = (\MsgSetup, \Wat, \Extract)$ be the construction
    from \figref{fig-embedding} using $\wat'$.

    Then for all efficient $\A$, the following event \fail occurs with negligible probability:

    \ifx\cols\one
    \begin{itemize}
        \item $\forall i, m_i \in M$, AND
        \quad \texttt{// only queried messages in M}
        \item $R_k\left(\lambda,(\prompt_i)_i, (T_i)_i, \hat{T}\right) = 1$, AND
        \quad \texttt{// robustness condition passes}
        \item $\hat{\msg} \not\in \Feasible_\delta(M)$
        \quad \texttt{// extracted message unrelated to M}
    \end{itemize}
    \fi
    \ifx\cols\two
    \begin{itemize}
        \item $\forall i, m_i \in M$, AND\\
        \quad \texttt{// only queried messages in M}
        \item $R_k\left(\lambda,(\prompt_i)_i, (T_i)_i, \hat{T}\right) = 1$, AND\\
        \quad \texttt{// robustness condition passes}
        \item $\hat{\msg} \not\in \Feasible_\delta(M)$\\
        \quad \texttt{// extracted message unrelated to M}
    \end{itemize}
    \fi
    in the probability experiment defined by
    \begin{itemize}
        \item $\sk \getsr \MsgSetup(1^{\lambda})$
        \item $\hat{T} \gets \A^{\Wat_\sk(\cdot,\cdot)}(1^\lambda)$, denoting by 
        $(m_i,\prompt_i)_i$ and $(T_i)_i$ the sequence of inputs and outputs of the oracle
        \item $\hat{\msg} \gets \Extract_\sk(\hat{T})$.
    \end{itemize}
\end{lemma}
\begin{proof}
\ifx\cols\one
We will prove \lemref{lem-msg-robust} via three hybrid transitions, moving from
the experiment defined in the lemma statement to one
where all of the calls to the underlying zero-bit watermarking algorithm
$\wat'$ are replaced by calls to the robustness condition $R_1$ and to
$\GenModel$. In Hybrid 1, we use the soundness of $\wat'$ to remove all calls to 
$\Detect'$ that use keys unrelated to any messages $\msg \in M$. In Hybrid 2, we 
use the $R_1$-robustness of $\wat'$
to remove the calls to $\Detect'$ that use the rest of the keys, this time
replacing them with calls to $R_1$. Finally, in Hybrid 3 we rely on the
undetectability of $\wat'$ to replace all of the calls
to $\Wat'$ with calls to $\GenModel$. We then use the definition of $R_k$-robustness and our choice of $k$ to complete the proof.
\fi

Let $p_\fail$ be the probability of the event $\fail$ in the experiment defined in
the lemma statement. 
\ifx\cols\two
We will prove \lemref{lem-msg-robust} via three hybrid transitions.
\fi
Throughout this proof, we condition on the events 
$(\forall i, \msg_i \in M)$ and $R_k(\lambda,(Q_i)_i, (T_i)_i, \hat{T}) = 1$, both 
of which are efficiently checkable by $\A$ and implied by \fail.

\newcommand{\Mkeys}{\mathcal{K}_M}
\newcommand{\MkeysC}{\overline{\mathcal{K}}_M}

\paragraph{Hybrid 1 (Soundness)}
Consider the set of keys $\Mkeys = \{k_{i,m[i]} : i \in [L], \msg \in M\}$, and its complement
$\MkeysC$. In Hybrid 1, we modify $\Extract_\sk$ to remove every call to
$\Detect'_{k_{i,b}}(\hat{T})$ for $k \in \MkeysC$, replacing it with $z_{i,b} \gets \false$.
Let $\hat{\msg}_1$ be the result of $\Extract_\sk(\hat{T})$ in Hybrid 1.
By construction, if $\hat{\msg}_1[i] \neq \bot$, then
$\hat{\msg}_1[i] = m[i]$ for some $m\in M$. In other words,
although there may be many $\bot$ entries, non-$\bot$ entries of $\hat{\msg}_1$
must agree with some element
of $M$, so $\hat{\msg}_1 \in \Feasible_{\gamma_1}(M)$ for some $\gamma_1\le 1$.

Observe that the keys
in $\MkeysC$ are never used by $\Wat_\sk$. %
Hence the view of $\A$ --- and in particular its output $\hat{T}$
--- is independent of the keys $\MkeysC$. By the soundness of $\wat'$, every
call to $\Detect'_k(\hat{T})$ for $k \in \MkeysC$ returns $0$ with high probability.
Notice that $\Extract_{\sk}$ does not change its behavior whether a call
to $\Detect'$ returns $0$ or is removed entirely (the only difference between the real game and Hybrid 1), since each $z_{i,b}$ is initialized to $\false$ in the hybrid.

Therefore, the output distribution of $\Extract_\sk$ in Hybrid 1 and the real execution
are statistically close (conditioned on $\forall i, \msg_i \in M$). Let $p_1$ be
the probability of the
event $\fail$ in Hybrid 1. As $\Extract_\sk$ changed only negligibly between
the hybrids, we have that
\[
    |p_1 - p_\fail| \le \negl(\lambda).
\]

\begin{figure}
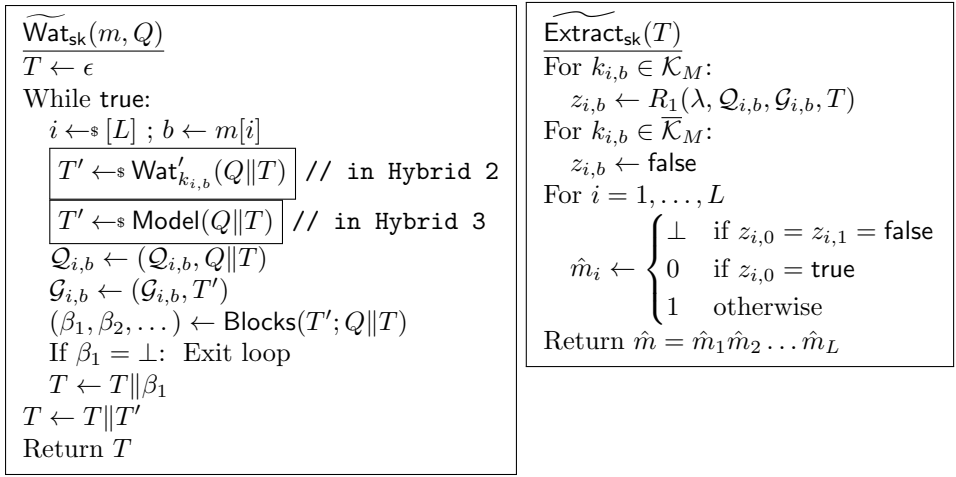

\begin{center}
\fbox{
\begin{minipage}[t]{4cm}
    \begin{tabbing}
        12\=123\=123\=123\=123\=\kill
        \underline{$\widetilde{\Encode}_{\sk}(\msg,\prompt)$}\\
        $T \gets \epsilon$\\
        While $\true$:\\
        \> $i \getsr [L]$ ; $b \gets \msg[i]$\\
        \> \codeBox{$T' \getsr \Encode'_{k_{i,b}}(\prompt\|T)$} \texttt{// in Hybrid 2}\\
        \> \codeBox{$T' \getsr \GenModel(\prompt\|T)$} \texttt{// in Hybrid 3}\\
        \> $\calQ_{i,b} \gets (\calQ_{i,b}, \prompt\| T)$ \\
        \> $\calR_{i,b} \gets (\calR_{i,b}, T')$ \\
        \> $(\block_1, \block_2,\dots) \gets \Blocks(T';\prompt \|T)$\\
        \> If $\block_1 = \bot$:~~Exit loop\\
        \> $T \gets T \| \block_1$\\
        $T \gets T \| T'$\\
        Return $T$
    \end{tabbing}
\end{minipage}
}
\fbox{
\begin{minipage}[t]{4cm}
    \begin{tabbing}
        12\=123\=123\=123\=123\=\kill
        \underline{$\widetilde{\Extract}_{\sk}(T)$}\\
        For $k_{i,b}\in \Mkeys$:\\
        \> $z_{i,b} \gets R_1(\lambda,\calQ_{i,b},\calR_{i,b},T)$\\
        For $k_{i,b} \in \MkeysC$:\\
        \> $z_{i,b} \gets \false$\\
        For $i=1,\ldots,L$\\
        \> $\hat{m}_i \gets
            \begin{cases}
                \bot &\text{if } z_{i,0} = z_{i,1} = \false\\
                0 &\text{if } z_{i,0} = \true\\ %
                1 &\text{otherwise }\\
            \end{cases}$\\
        Return $\hat{m} = \hat{m}_1 \hat{m}_2 \dots \hat{m}_L$
    \end{tabbing}
\end{minipage}
}
\caption{Intermediate versions of $\Encode$ and $\Extract$,
used to produce outputs that are independent of the keys,
used in the proof of \lemref{lem-msg-robust}.
The boxed lines are used only in the indicated hybrid. At setup, we additionally
initialize $\calQ_{i,b}, \calR_{i,b}$ to $(~)$ for all $i,b$.}
\label{fig-msg-hybrids}
\end{center}
\end{figure}

\paragraph{Hybrid 2 (Robustness)}
The pseudocode for Hybrid 2 is given in \figref{fig-msg-hybrids}. 
In this hybrid, we remove the remaining calls to $\Detect'_{k_{i,b}}\allowbreak (\hat{T})$ for
$k \in \Mkeys$, replacing each with a call to the robustness condition $R_1$ (which
requires the relevant set of queries and responses as input).
So, for each $k_{i,b}$, the corresponding call to  $R_1$ is run on $\calQ_{i,b}$ and
$\calR_{i,b}$, the sequences of \emph{queries} to and \emph{generations} from
$\Wat'_{k_{i,b}}$ in $\widetilde{\Wat}_\sk$.
The code of $\widetilde{\Wat}_\sk$ is edited to track $\calQ_{i,b}$ and $\calR_{i,b}$.

Let $\hat{\msg}_2$ be the result of $\widetilde{\Extract}_\sk(\hat{T})$ in Hybrid 2.
As in Hybrid 1, $\hat{\msg}_2 \in F_{\gamma_2}(M)$ for some $\gamma_2\le1$.
Moreover, $\gamma_2 \ge \gamma_1$, meaning that $\hat{\msg}_2$ can only contain more
$\bot$-entries than $\hat{\msg}_1$.
This is because $R_1(\lambda, \calQ_{i,b}, \calR_{i,b}, \hat{T}) = 1$
implies that $\Detect'_{k_{i,b}}(\hat{T}) = 1$ with high probability
(but not the converse), because $\wat'$ is $R_1$-robust.\footnote{Otherwise,
we construct adversary $\A'$ breaking the $R_1$-robustness of
$\Wat'_{k_{i,b}}$, as follows:
$\A'$ runs $\A$, internally simulating its $\widetilde{\Wat}$ oracle,
only querying its own $\Wat'$ oracle to simulate calls to $\Wat'_{k_{i,b}}$.
$\A'$ outputs the string $\hat{T}$ returned by $\A$.
Notice that this critically
uses the adaptivity of our robustness definition.}

Let $p_2$ be the probability of the event \fail in Hybrid 2. For
$\hat{\msg}_1 \in F_{\gamma_1}(M)$, $\hat{\msg}_2 \in F_{\gamma_2}(M)$, and
$\gamma_2 \ge \gamma_1$, we have
\ifx\cols\one
\begin{align*}
    \Pr[\hat{\msg}_2 \not \in F_\delta(M)]=\Pr[\gamma_2 > \delta]
    \ge \Pr[\gamma_1 > \delta]
    = \Pr[\hat{\msg}_1 \not \in F_\delta(M)].
\end{align*}
\fi
\ifx\cols\two
\begin{align*}
    \Pr[\hat{\msg}_2 \not \in F_\delta(M)]=\Pr[\gamma_2 > \delta]
    &\ge \Pr[\gamma_1 > \delta] \\
    &= \Pr[\hat{\msg}_1 \not \in F_\delta(M)].
\end{align*}
\fi
Hence, $$p_2 \ge p_1 - \negl(\lambda).$$

\paragraph{Hybrid 3 (Undetectability)}
The pseudocode for Hybrid 3 is given in \figref{fig-msg-hybrids}. 
In Hybrid 3, we remove all use of the watermarking scheme $\Msg'$ 
by replacing $\Wat'$ with $\GenModel$.
In these functions, we no longer sample
generations from $\Encode'$ and instead use $\GenModel$.
Observe that in Hybrid 3, the adversary's view is independent of the
indices $i$ sampled by $\widetilde{\Encode}$.

Let $p_3$ be the probability of \fail in Hybrid 3. By undetectability of $\Msg'$
$$|p_3 - p_2| \le \negl(\lambda).$$

\paragraph{At most $\lfloor\delta L \rfloor$ erasures}
Recall we are conditioning on \ifx\cols\two\\\fi
$R_k(\lambda, (\prompt_j)_j,(T_j)_j,\hat{T}) = 1$.
By definition, there are $k' \ge k$ substrings $\hat{\tau}\in \hat{T}$ that are
$\bumpeq$-close to disjoint blocks $\block \in \cup_j \Blocks(T_j;\prompt_j)$. Each
such block $\block$
has an associated index $i_\block$: the index that was sampled by
$\widetilde{\Encode}_\sk$ in the iteration that generated $\block$.
By construction, a bit is extracted at each of these indices: $\hat{\msg}[i_\block] \neq \bot$.

The indices $i_\block$ are uniform over $[L]$ and independent of one another.
Hence, the number $N_\bot$ of indices $j$ where $\hat{\msg}[j] = \bot$ is distributed as
the number of empty bins remaining after throwing $k'$ balls into $L$ bins uniformly at random.
By \lemref{lm-bnb} and the hypothesis that $k' \ge k \ge k^*(L, \delta)$, we have
that $N_\bot \le \lfloor \delta L\rfloor$ except with probability $e^{-\lambda}$.
We know that $\hat{\msg}\in F_1(M)$ because
$\hat{\msg}_1\in F_{\gamma_1}(M)$ (with high probability) and
the two games are negligibly different.
Combining this fact with our argument about $N_{\bot}$ shows that
we have  $\hat{\msg} \in F_\delta(M)$, with high probability. Hence
\[
    p_3 \le \negl(\lambda) \implies p_\fail \le \negl(\lambda)\qedhere
\]
\end{proof}

\begin{theorem}[$\Msg$ is robust]\label{thm-msg-undetectable}
    Suppose $\Msg'$ is a block-by-block zero-bit watermarking scheme that is 
    undetectable, sound, and $R_1$-robustly detectable.
    Then the $\Msg$ construction from
    \figref{fig-embedding} is an $L$-bit watermarking scheme that is
    sound, undetectable, and $(\delta,R_k)$-robustly extractable for $k\ge k^*(L,\delta)$.
\end{theorem}
\begin{proof}
    By Claims~\ref{clm-msg-sound} and~\ref{clm-msg-undetectable}, $\Msg$ is sound and undetectable.
    Robustness is an an immediate corollary of \lemref{lem-msg-robust}, 
    fixing the subset $M\subseteq \bits^L$ to be a singleton set. If an adversary
    can only query a fixed message $\msg\in M$, then the event bounded
    in \lemref{lem-msg-robust} is exactly the definition of
    $(\delta,R_k)$-robustness. 
\end{proof}

\subsection{How good is $R_{k^*}$?}
\label{sec:l-bit:on-k}

Theorem~\ref{thm-msg-undetectable} states that our scheme is $(\delta,R_{k^*})$-robust where:
\[
    k^*(L,\delta) = \min\left\{L\cdot (\ln L + \lambda); \quad L\cdot \ln\left(\frac{1}{\delta - \sqrt{\frac{\lambda + \ln 2}{2L}}} \right) \right\}
\]
This means that our watermarking scheme embeds $L$-bit messages into model-generated
text $T$ such that at least a $(1-\delta)$-fraction of the embedded message can be recovered
from any text $\hat{T}$ containing at least $k^*(L,\delta)$ approximate blocks from $T$.
We remark that as our construction doesn't depend on $\delta$, it satisfies
$(\delta,R_{k^*})$-robustness for all $0\le \delta < 1$ simultaneously.

We'd like $k$ to be as small as possible, for two reasons. First, because smaller $k$ means that $\hat{T}$ can be farther from $T$ --- guaranteeing stronger robustness. Second, because $k$ blocks are required to extract the mark even in the
absence of an adversary! Language models have variable-length outputs, and too-short
$T$ are not marked. Smaller $k$ means that more of model's outputs are marked. But we
cannot make $k$ too small without a very different approach. Any scheme that embeds each
bit into a distinct block of text requires $k\ge L(1-\delta) = O(L)$.

So how does $k^*$  compare to the $L(1-\delta)$ lower bound?
We consider two parameter regimes.
For $L < (\lambda + \ln 2)/2$, we have $k^* = L(\ln L + \lambda)$. 
As $L = \poly(\lambda)$, we have that $k^* = O(L\lambda)$ and is independent of $\delta$.
For $L > (\lambda + \ln 2)/2$, the minimum is achieved at $\delta > 0$. Taking
$c = \sqrt{2L/(\lambda + \ln 2)}$ and $\delta >1/c$, we have that
$k^* \le L \ln \left(\frac{1}{\delta -1/c}\right) = O(L)$.
(Even better, $k^* < L$ for $\delta >1/c + 1/e$.)

A possible approach for further improving the parameter $k$ is to use error correcting codes.
The idea is simple. Let $\ECC$ be an error correcting code  with block-length $L' > L$.
To losslessly embed a mark $\msg\in \bits^L$, embed the codeword $w = \ECC(\msg)$
using an $L'$-bit watermarking scheme.
If $\hat{w}$ can be extracted with at most
$\delta$ fraction of erasures, we can decode and recover $\msg$ in its entirety.
The result would be a lossless $R_{k'}$-robust scheme for
$k' = k^*(L',\delta)$. If $L' = o(L\lambda)$, this yields an asymptotic
improvement $k'= o(k) = o(L\lambda)$.

\paragraph{Comparison to \cite{EPRINT:ChrGun24}}
While most prior work constructs zero-bit watermarking, Christ and Gunn \cite{EPRINT:ChrGun24}
build both zero-bit and (lossless) $L$-bit watermarking schemes, from zero/$L$-bit
pseudorandom error-correcting codes respectively.
The $L$-bit scheme of \cite{EPRINT:ChrGun24} and our $L$-bit scheme instantiated with the zero-bit scheme of \cite{EPRINT:ChrGun24} have incomparable robustness guarantees.
Roughly speaking, their zero-bit scheme is $R_1$-robust for 
blocks that require
$O(\lambda)$ empirical entropy. Their $L$-bit scheme is also $R_1$-robust, but with ``longer''
blocks requiring $O(L+\lambda)$ empirical entropy.
Importantly this block of text has to be a single \emph{contiguous} block that was produced
in \emph{one generation}.

In contrast, our $L$-bit scheme is $R_k$-robust, requiring $k = O(L\lambda)$ of the
original, zero-bit blocks.
While it requires more empirical entropy overall, $R_k$-robustness allows these $k$ blocks to appear \emph{anywhere} in all of the generations ever seen by the adversary.

\section{Building multi-user watermarks}\label{sec-multiuser}

In this section, we construct a multi-user watermarking scheme 
using $L$-bit watermarking schemes and robust fingerprinting 
codes as black boxes. When instantiated
with our $L$-bit scheme from \secref{sec-embedding}, the result is robust against
colluding users as well.
To our knowledge, ours is the first watermarking scheme
that is secure against any sort of collusion. The black-box nature of our construction 
ensures that we can instantiate our multi-user schemes with improved parameters whenever 
the underlying watermarking schemes or fingerprinting codes improve.

\secref{sec:constructing-multi} gives the construction of our multi-user scheme
and \ifx\cols\one proves \fi\ifx\cols\two sketches proofs of \fi its undetectability, consistency, and soundness, using the undetectability 
and soundness of the underlying $L$-bit scheme. In \secref{sec-undetectable} we 
prove robustness of our multi-user scheme, which follows from the robustness of the
fingerprinting code and our own $L$-bit scheme (\thmref{lem-msg-robust}). Finally, in
\secref{sec:undetectable-bonus} we analyze the key features of our multi-user scheme
and compare our approach to other possible constructions.

\subsection{Constructing multi-user watermarks}\label{sec:constructing-multi}

\begin{figure}[h!]
\begin{center}
\fbox{
\begin{minipage}[t]{4cm}
    \begin{tabbing}
        123\=123\=123\=123\=123\=\kill
        \underline{$\Setup(1^{\lambda})$}\\
        $(X, \FPKey) \getsr \FPGen'(1^\lambda, \sf{pp})$\\
        $\sk \getsr \MsgSetup'(1^{\lambda})$\\
        Return $(X,\FPKey,\sk)$
        \\
    \end{tabbing}
    \begin{tabbing}
        123\=123\=123\=123\=123\=\kill
        \underline{$\Wat_{(X,\FPKey,\sk)}(u,\prompt)$}\\
        $T \getsr \Encode'_{\sk}(X_u,\prompt)$\\
        Return $T$
    \end{tabbing}
\end{minipage}
}
\fbox{
\begin{minipage}[t]{4cm}
    \begin{tabbing}
        123\=123\=123\=123\=123\=\kill
        \underline{$\Detect_{(X,\FPKey,\sk)}(\hT)$}\\
        $\hat{\msg} \gets \Extract'_{\sk}(\hT)$\\
        Return $\1(\exists i~\hat{\msg}_i \ne \bot)$
    \end{tabbing}
    \begin{tabbing}
        123\=123\=123\=123\=123\=\kill
        \underline{$\Trace_{(X,\FPKey,\sk)}(\hT)$}\\
        $\hat{\msg} \gets \Extract'_{\sk}(\hT)$\\
        If $\hat{\msg} = \bot^L$:\\
        \> Return $\emptyset$\\
        $C \gets \FPTrace'(\hat{\msg}, \FPKey)$\\
        Return $C$
    \end{tabbing}
\end{minipage}
}
\caption{Pseudocode for construction of
$\wat=(\Setup,\Wat,\Detect,\allowbreak\Trace)$
from fingerprinting code $\FP=(\FPGen',\FPTrace')$
and $L$-bit message embedding scheme $\Msg' = (\MsgSetup', \Encode', \Extract')$,
e.g. \figref{fig-embedding}. The construction is defined for any public parameters $\sf{pp} = (n, c, \delta)$, where $n, c > 1,$ and $0 \le \delta < 1$.}
\label{fig-multiuser}
\end{center}
\end{figure}

Our construction is given in \figref{fig-multiuser}. Let
$\Msg' \allowbreak = (\MsgSetup', \Encode',\allowbreak \Extract')$ be an $L$-bit watermarking scheme.
The multi-user watermarking scheme $\wat = (\Setup,\Wat,\Detect,\Trace)$ is 
constructed as follows. The secret key output by $\Setup$ consists of the 
fingerprinting codewords $X$ and the tracing key $\FPKey$, as well as the secret 
key $\sk$ from the $L$-bit scheme. In response to any prompt $\prompt$ from a user $u$,
$\Wat$ will use the $L$-bit watermarking algorithm to watermark the user's 
fingerprinting codeword $X_u$ into the response. At detection time, $\Detect$
will use the $L$-bit scheme to extract a message $\hat{\msg}$ from $\hT$ and return $1$
as long as $\hat{\msg} \neq \bot^L$. To trace a user, $\Trace$ will run the fingerprinting
code's tracing algorithm on $\hat{\msg}$ and return the set of accused users.

\ifx\cols\one
We now show that our multi-user scheme is consistent, undetectable, and sound. Robustness
to collusions (and hence completeness) is deferred to \secref{sec-undetectable}, because
it requires instantiating our multi-user scheme with our $L$-bit watermarking scheme.

\begin{clm}[$\wat$ is consistent]\label{clm-consistent}
    Let $L, n, c > 1$ be integers and $0 \le \delta < 1$.
    Let $\Msg'$ be an $L$-bit watermarking scheme and
    $\FP$ be a fingerprinting code.
    Then the $\wat$ construction from
    \figref{fig-multiuser} is a consistent multi-user watermarking scheme.
\end{clm}
\begin{proof}
The algorithm $\Detect_{(X,\FPKey,\sk)}(T)=0$ only if
$\Extract_{\sk}(T)$ returns $\bot^L$. The check in $\Trace$
will guarantee that in such a case,
$\Trace_{{(X,\FPKey,\sk)}}(T) = \emptyset$.\footnote{
    Note that we may instead want to check if $|\{i : s_i = \bot\}| > \delta L$,
    because in the formal fingerprinting games, $\FP$ is allowed to output
    anything on strings with this many $\bot$ entries.
    However, for non-contrived schemes, this is unnecessary.
}
\end{proof}

Whenever our scheme is instantiated with an undetectable $L$-bit watermarking scheme, the
overall output will also be undetectable, because $\Wat$ just returns the value
from the underlying $\Wat$ query. The proof is omitted, as it is essentially identical to the proof of \clmref{clm-msg-undetectable}. Together with \clmref{clm-msg-undetectable},
we obtain undetectable watermarking from a black-box undetectable zero-bit watermarking
scheme.

\begin{clm}[$\wat$ is undetectable]\label{clm-undetectable}
If $\GenModel$ is
prefix-specifiable and $\Msg'$ is an $L$-bit
watermarking scheme built from an undetectable zero-bit scheme, then the $\wat$ 
construction from \figref{fig-multiuser} is undetectable.
\end{clm}

Next we show that $\wat$ is sound, as long as the
underlying $L$-bit scheme is sound. We need this property to ensure that we do
not falsely detect a watermark when it is not present (Type I errors). Notice that because our scheme
is also consistent we will not falsely accuse users of generating unmarked text.

\begin{clm}[$\wat$ is sound]\label{clm-sound}
    Let $L, n, c > 1$ be integers and $0\le \delta < 1$.
    Let $\Msg'$ be a sound $L$-bit watermarking scheme and
    $\FP$ be a fingerprinting code of length
    $L$ with parameters $(\lambda,n,c,\delta)$.
    Then the $\wat$ construction from \figref{fig-multiuser} is a sound
    multi-user watermarking scheme.
\end{clm}
\begin{proof}
The soundness of $\wat$ follows immediately from the soundness of
$\Msg'=(\MsgSetup',\Encode',\Extract')$.
Specifically, observe that $\Detect_{(X,\FPKey,\sk)}(T) = 1$ implies 
$\Extract_{\sk}'(T) \ne \bot^L$. So, for a every $\poly$, $\lambda$,
and $T$ with $|T| \le p(\lambda)$
\ifx\cols\one
\begin{align*}
    \Pr_{(X,\FPKey,\sk) \getsr \Setup(1^{\lambda})}[\Detect_{(X,\FPKey,K)}(T) = 1]
    \le \Pr_{\sk' \getsr \MsgSetup'(1^{\lambda})}[\Extract'_{\sk'}(T) \ne \bot^L]
    < ~\negl(\lambda).\qedhere
\end{align*}
\fi
\ifx\cols\two
\begin{align*}
    \Pr_{(X,\FPKey,\sk) \getsr \Setup(1^{\lambda})}&[\Detect_{(X,\FPKey,K)}(T) = 1]\\
    \le & \Pr_{\sk' \getsr \MsgSetup'(1^{\lambda})}[\Extract'_{\sk'}(T) \ne \bot^L]\\
    < &~\negl(\lambda). \quad\quad\quad\quad\quad\quad\quad\quad\quad\quad\quad~~ \qedhere
\end{align*}
\fi
\end{proof}

\fi
\ifx\cols\two
In \apref{ap:multiuser-easy-props}, we provide formal claims that our multi-user scheme is consistent and inherits undetectability and soundness from its underlying scheme. The proofs follow directly from the definitions and can be found in the full version of the paper.
\fi

\subsection{Robustness against collusions}\label{sec-undetectable}
We now describe how our multi-user scheme from \figref{fig-multiuser}
can achieve robust collusion resistance under the right conditions.

Our main theorem for multi-user watermarking, \thmref{thm-undetectable}, requires 
that $\wat$ is built out of the $L$-bit watermarking scheme from \figref{fig-embedding},
which itself is built out of an undetectable zero-bit watermarking scheme. As 
shown in \lemref{lem-msg-robust}, our 
$L$-bit scheme will only ever extract (noisy) feasible messages of the set of 
the adversary's queried messages.

Since the messages that $\wat$ watermarks are codewords from a robust fingerprinting 
code, the $\hat{\msg}$ recovered will necessarily be in the $\delta$-feasible ball around
the set of the adversary's codewords. As long as enough blocks are included in
the adversarially generated text $\hat{T}$, we will be able to trace back to one 
of the colluding users.

\begin{theorem}[$\wat$ is robust]\label{thm-undetectable}
    Let $n, c > 1$ be integers and $0\le \delta < 1$.
    Let $\Msg'$ be the %
    $L$-bit watermarking scheme from \figref{fig-embedding}, built from a block-by-block
    zero-bit watermarking scheme that is undetectable, sound, and $R_1$-robustly detectable.
    Furthermore let $\FP$ be a robust fingerprinting code of length
    $L$ with parameters $(\lambda,n,c,\delta)$.

    Then, the $\wat$ construction
    from \figref{fig-multiuser} is a multi-user watermarking scheme that is
    consistent, sound, undetectable, and $R_{k^*}$-robust against $c$-collusions,
    for $k^*$ given by \lemref{lm-bnb}.
\end{theorem}
\begin{proof}
We have already shown that $\wat$
is consistent (\clmref{clm-consistent}), sound (\clmref{clm-sound}), and
undetectable (\clmref{clm-undetectable}).
Robustness is a corollary of \lemref{lem-msg-robust} when using
an appropriate fingerprinting code. Let $\ColSet$ be the set of (at most $c$) 
colluding users and apply \lemref{lem-msg-robust} to $M=\{X_u : u \in C\}$. Then 
we know that $\Extract'$ from $\Msg'$ will return some $\hat{\msg}$ with at most
$\lfloor \delta L \rfloor$ entries that are $\bot$. Therefore, with high probability 
we have $\hat{\msg} \in \Feasible_\delta(X_C)$. By the definition of a robust
fingerprinting code we know that $\FPTrace'$ will correctly accuse a colluding 
user with all but negligible probability. 
\end{proof}

\paragraph{Efficiency}
Notice that the construction in \figref{fig-multiuser} only requires a single
call to the underlying $L$-bit scheme for every generation and detection.
When tracing, it additionally only requires a single call to the fingerprinting
code's tracing algorithm. 
When instantiated with the $L$-bit scheme in \figref{fig-embedding}, detecting and tracing make $2L$ calls to its the underlying zero-bit scheme's detection algorithm. Fingerprinting code lengths scale
as the logarithm of the number of users, so our scheme is much faster
(and requires much less storage)
than one which generates keys for every single user, whose detection would require
linear time.

Unfortunately, existing fingerprinting codes require time linear
in the number of users to trace, so our tracing time still scales linearly. 
Fingerprinting tracing algorithms work by checking whether the extracted codeword is sufficiently close to each user's unique codeword, one by one. This means that $\Trace$ could be used to check any set of $c$ suspects in time linear in $c$.
It is an interesting open problem to improve the runtime of the fingerprint tracing, which would correspondingly improve our scheme.

\subsection{Other features of our multi-user watermarks}\label{sec:undetectable-bonus}
We discuss additional features of our construction that are not captured by the above definitions or theorems, but which offer practical improvements

\paragraph{Preserved zero-bit detection}
\thmref{thm-undetectable} proves that $\wat$ is $R_k$-robustly traceable
so long as the constructions in \figref{fig-embedding}
and \figref{fig-multiuser} use a
zero-bit scheme that is $R_1$-robustly detectable. Happily, the multi-user construction from
\figref{fig-multiuser} is \emph{also} $R_1$-robustly detectable! In particular,
the $\Detect$ function will detect if any single (approximate) block is present in a
generation. This means that our construction preserves the original robustness of 
the zero-bit watermarking scheme. The added benefit of finding users and resisting 
(unbounded!) collusions comes essentially for free: no cost to robust detection and only a
logarithmic slowdown in $\Detect$.

\paragraph{Recovering more users}
In our multi-user construction, we use the function $\Extract'$ from the $L$-bit scheme
as a black box, returning a single bitstring which is then fed into the fingerprinting 
code's tracing algorithm to accuse some set of users. However, a different construction 
may be able to recover even more colluding users. 
Slightly modifying the $\Extract'$ function from \figref{fig-embedding},
we could allow $\Extract'$ to return a special character ``$*$'' in the $i$th index
whenever both  $z_{i,0}$ and $z_{i,1}$ are $\true$
(since it recovered both bits in the same index). By
the soundness of the underlying zero-bit scheme, this should almost never happen 
when embedding a single message, as in the normal $L$-bit robustness game.
In the process of colluding, however, it is likely that users with different
codeword bits at index $i$ happen to include
(approximate) blocks in $\hT$ for both of their bits.

Therefore, the set of \emph{all strings} which could be created from the extracted
message $\hat{\msg} \in \{0, 1, \bot, *\}^L$ is a subset
of the $\delta$-feasible ball $\Feasible_\delta(X_{\ColSet})$ of the colluding users'
codewords. In practice, one may want to call $\FPTrace'$ on each of these
strings and return the union of all users returned, which will still (with high
probability) be a subset of the colluding users.
An interesting question is to design fingerprinting codes that allow faster tracing from pirate codewords $\hat{m} \in \{0,1,\bot,*\}^L$ than brute-force search, which requires time exponential in the number of $*$'s.

\paragraph{Robust fingerprinting codes reduce $k^*$}
Notice that our construction of multi-user watermarking uses a robust 
fingerprinting code with erasure bound $\delta$ in conjunction with an $L$-bit 
watermarking scheme that allows up to $\delta$ erasures. The length $L$ of the 
fingerprinting code grows with $\delta$. Ultimately,
the parameter we are most interested in is the number of (approximate)
blocks $k^* = k^*(L,\delta) = \Omega(L)$ needed to extract the watermark
from $\hat{T}$.\footnote{
    Other robustness parameters, like the $\bumpeq$-relation and block length are
    received in a black-box way from the underlying zero-bit scheme and therefore can be
    improved immediately as future work develops.
}

This raises the question of whether robustness $(\delta > 0)$ is helping at all,
or whether we would be better off using shorter fingerprinting codes without
adversarial robustness ($\delta = 0$).
Perhaps surprisingly, robustness yields an asymptotic\footnote{
    Showing a concrete improvement with $\delta > 0$ amounts to the same comparison
    between $k^*(L_0,0)$ and $k^*(L_\delta, \delta)$ as in the body, but
    with concretely optimal fingerprinting codes. It appears that the construction
    of \cite{SPRINGER:NFHKWOI07} is much more efficient than either of the
    asymptotically optimal codes we discuss. Like other fingerprinting codes,
    $L_\delta = L_0/\poly(1-\delta)$. Hence taking a constant
    $\delta \gg \sqrt{(\lambda + \ln 2)/2L_\delta}$ (say, $\delta = 1/2$)
    should yield $k^*(L_\delta,\delta) = O(L_\delta) = O(L_0)$, compared
    to $k^*(L_0,0) = O(L_0\lambda)$. However, we were unable to work out the
    details  of \cite{SPRINGER:NFHKWOI07} to our satisfaction.
} improvement in $k^*$.

The question boils down to comparing $k^*(L_{\delta},\delta)$ and $k^*(L_0, 0)$, where $L_0$ and $L_\delta$ are the lengths of the asymptotically optimal fingerprinting codes for $\delta = 0$ and $\delta > 0$, respectively.
For $n$ users, $c$ collusions, a $\delta$-fraction of adversarial erasures, and security parameter $\lambda$, the asymptotically optimal robust fingerprinting code of \cite{ACM:BKM10} has length
\[
    L_\delta = \frac{C(c \ln c)^2 \ln (n) \lambda}{1-\delta} 
\]
for some very large constant C.
Letting $W:= C(c\ln c)^2 \ln(n)$, we have $L_\delta = W\lambda/(1-\delta)$.
The asymptotically optimal non-robust fingerprinting code \cite{ACM:Tar08} has length 
$L_0 = 100 c^2 \ln (n) \lambda.$
Choosing $\delta = 1/2$, we get 
$k^*(L_{1 / 2},1/2) = O(c^2 \ln^2(c) \ln (n) \lambda)$,
whereas $k^*(L_0,0) = \Omega(c^2 \ln (n) \lambda^2)$. As  $c = \poly(\lambda)$,
we have $k^*(L_{1 / 2},1/2) = o(k^*(L_0, 0))$.

\ifx\cols\one
\section{Watermarking without undetectability}\label{sec-detectable}

Throughout the paper, we make heavy use of the \emph{undetectability} of watermarking schemes. 
While some watermarking schemes for language models are provably 
undetectable under cryptographic assumptions
\cite{EPRINT:ChrGunZam23,EPRINT:ChrGun24,ARXIV:GM24},
most are not (e.g. \cite{Aar22,EPRINT:FGJMMW23,ICML:KGWKMG23,ARXIV:KTHL23}).\footnote{
    Though \cite{EPRINT:FGJMMW23} is undetectable under a min-entropy assumption on the 
    model's output, it does not appear difficult for an adversary to create prompts undermining
    the assumption -- thereby violating our stringent undetectability definition.
}
Our constructions of $L$-bit and multi-user watermarking
(Fig.~\ref{fig-embedding},~\ref{fig-multiuser}) use an arbitrary zero-bit
watermarking scheme as a building block.

In this section, we explore the robustness of our schemes when instantiated with a
zero-bit scheme that is not undetectable.
We prove analogues of Theorem~\ref{thm-msg-undetectable} and
Theorem~\ref{thm-undetectable} with one very important difference
(Section~\ref{sec:detectable-main}). For undetectable schemes, our constructions
are $R_k$ robust for $k = O(L\lambda)$, where $k$ is the number of (modified)
blocks of model-generated text needed to extract the watermark. Without
undetectability, we require $k = \Omega(B)$, where $B$ is the \emph{total} number of
blocks of model-generated text that the adversary \emph{ever observed}.
To provably extract the watermark, the adversary's output must
contain essentially all the text produced by the model!

Still, the adversary has substantial freedom to modify the watermarked text
without destroying the mark. First, they can reorder blocks
arbitrarily and include extraneous unmarked text. Second, the adversary
can modify the blocks as allowed by the underlying scheme's
robustness guarantee ($\bumpeq$, e.g., bounded edit distance).

In Section~\ref{sec:detectable-discussion}, we discuss some possible
approaches for improving the robustness parameter $k$ without undetectability.
In practice, we believe that full undetectability may not be necessary for meaningful security.
Schemes that are not undetectable still offer some guarantees (e.g., bounded Renyi
divergence \cite{ARXIV:ZALW23} or undetectability for a single query \cite{ARXIV:KTHL23}).
In applications where these not-undetectable zero-bit watermarking schemes are
considered secure enough, we suspect our schemes would be too.

The theorems we prove in this section apply to {adaptively} robust
watermarking schemes that are not undetectable (though we know of no such schemes).

\subsection{Robustness for $k = \Omega(B)$}
\label{sec:detectable-main}
We need to prove an analogue of Lemma~\ref{lem-msg-robust} without undetectability.
As before, our theorems follow as corollaries.

Undetectability is only used in the third hybrid of the proof of
Lemma~\ref{lem-msg-robust}. The other two hybrids use only soundness and robustness.
Hybrid 3 uses undetectability to argue that which bits of the watermark are
embedded in which blocks of text is (computationally) independent from adversary's view.
The result is that any (modified) block in the adversary's output $\hat{T}$
embeds a uniformly random bit of the message $\msg\in \bits^L$.
By a balls-in-bins argument, $k = L(\ln L + \lambda)$ blocks suffice to recover the whole message.

Without undetectability, this argument breaks down.
An adversary who can perfectly distinguish blocks marked using the different keys can
choose which message index $i$ any (modified) block encodes, as each index
corresponds to a distinct pair of keys.
In the worst case, an adversary can create $\hat{T}$ using $k$ blocks that all
correspond to the same index $i$, and $\Extract$ would recover just one bit of the message.

\newcommand{\detR}{R_{\sf{DET}}}

To get around this, we need to require that $k$ depends on the
\emph{total number of blocks} $B$ of model-generated text seen by the adversary $\A$:
$$B := \sum_{i} |\Blocks(T_i; \prompt_i)|,$$ 
where $(\prompt_i)_i$ are $\A$'s queries and $T_i\gets \Wat_\sk(Q_i)$ are the
corresponding generations.
It is easy to see that if $\hat{T}$ included all $B$ blocks, then the
adversary's hands are tied. As long as $B \geq L (\ln L + \lambda)$, all
message bits will be extracted from $\hat{T}$ with high probability.

The index $i$ of the message that each
block encodes is sampled uniformly at random by $\Wat$.
The adversary is required to output $k$ of these blocks. 
We can slightly generalize the above argument to allow $\delta L$ erasures. 
Let $s(B,L,\delta) < \delta B$ be any high-probability upper bound on the total
load of the $\lfloor \delta L \rfloor$ smallest bins after throwing $B$
balls uniformly at random into $L$ bins.\footnote{E.g., for $B\ge L(\ln L + \lambda)$,
$s(B,L,\delta) \ge \lfloor \delta L\rfloor$, as every bin has has load at least
1 with high probability.}
($\Wat$ is throwing the balls, not $\A$.) 
Then if $B \ge L(\ln L + \lambda)$ and $k \ge B - s(B,L,\delta)$, the $\Extract$
algorithm recovers at least $(1-\delta)L$ bits of the watermark.

We now state the analogue of Lemma~\ref{lem-msg-robust}. Let
\ifx\cols\one
\begin{align*}
    \detR\left(\lambda,(\prompt_i)_i, (T_i)_i, \hat{T}\right) =
    \1\big(B \ge L (\ln L + \lambda)\big)
    \land~
    \1\big( \numBlocks(\hat{T};\prompt_i,T_i) \ge B - s(B,L,\delta) \big).
\end{align*}
\fi
\ifx\cols\two
\begin{align*}
    \detR&\left(\lambda,(\prompt_i)_i, (T_i)_i, \hat{T}\right) =
    \1\big(B \ge L (\ln L + \lambda)\big)\\
    \land~&
    \1\big( \numBlocks(\hat{T};\prompt_i,T_i) \ge B - s(B,L,\delta) \big).
\end{align*}
\fi

\begin{lemma}[$\Msg$ recovers lossy descendants -- not undetectable]\label{lem-msg-robust-detectable}
    Let $\lambda,L \in \N$, $0\le \delta < 1$ and $M\subseteq \bits^L$.
    Suppose $\wat'$ is a block-by-block, zero-bit watermarking scheme that is sound and $R_1$ robustly-detectable.
    Let $\Msg = (\MsgSetup, \Wat, \Extract)$ be the construction
    from \figref{fig-embedding} using $\wat'$.

    Then for all efficient $\A$, the following event \fail occurs with negligible probability:
    \ifx\cols\one
    \begin{itemize}
        \item $\forall i, m_i \in M$, AND
        \quad \texttt{// only queried messages in M}
        \item $\detR \left(\lambda,(\prompt_i)_i, (T_i)_i, \hat{T}\right) = 1$, AND
        \quad \texttt{// robustness condition passes}
        \item $\hat{\msg} \not\in \Feasible_\delta(M)$
        \quad \texttt{// extracted message unrelated to M}
    \end{itemize}
    \fi
    \ifx\cols\two
    \begin{itemize}
        \item $\forall i, m_i \in M$, AND\\
        \quad \texttt{// only queried messages in M}
        \item $\detR \left(\lambda,(\prompt_i)_i, (T_i)_i, \hat{T}\right) = 1$, AND\\
        \quad \texttt{// robustness condition passes}
        \item $\hat{\msg} \not\in \Feasible_\delta(M)$\\
        \quad \texttt{// extracted message unrelated to M}
    \end{itemize}
    \fi

    in the probability experiment defined by
    \begin{itemize}
        \item $\sk \getsr \MsgSetup(1^{\lambda})$
        \item $\hat{T} \gets \A^{\Wat_\sk(\cdot,\cdot)}(1^\lambda)$, denoting by
        $(m_i,\prompt_i)_i$ and $(T_i)_i$ the sequence of inputs and outputs of the oracle
        \item $\hat{\msg} \gets \Extract_\sk(\hat{T})$.
    \end{itemize}
\end{lemma}
\begin{proof}[Proof outline]
The proof exactly follows the proof \lemref{lem-msg-robust} for Hybrids 1 and 2.
Hybrid 3 is omitted.
We recover a message index $i$ whenever a (modified) block in $\hat{T}$ was generated
using $\Wat_{k_{i,b}}$ for some $b$.
Because $B \ge L(\ln L + \lambda)$ and the definition of $s$, any set of
$B - s(B,L,\delta)$ blocks were generated using a set of at least
$L - \lfloor\delta L \rfloor$ distinct indices $i$, with high probability.
\end{proof}

\begin{theorem}[$\Msg$ is robustly extractable, when not undetectable]\label{thm-msg-detectable}
    Suppose $\Msg'$ is a block-by-block, sound zero-bit watermarking scheme. 
    Then, the $\Msg$ construction from \figref{fig-embedding} is an $L$-bit
    embedding scheme that is $(\delta,\detR)$-robust.
\end{theorem}
\begin{proof}
This follows as an immediate consequence of \lemref{lem-msg-robust-detectable} following
the same reasoning used in \thmref{thm-msg-undetectable}.
\end{proof}

\begin{theorem}[$\wat$ is robustly traceable, when not undetectable]\label{thm-detectable}
    Let $n, c > 1$ be integers and $0\le \delta < 1$.
    Let $\Msg'$ be the $L$-bit embedding scheme from
    \figref{fig-embedding}, built out of a block-by-block, sound zero-bit 
    watermarking scheme. Furthermore let $\FP$ be a robust fingerprinting code 
    of length $L$ with parameters $(\lambda,n,c,\delta)$.
    Then, the $\wat$ construction from \figref{fig-multiuser} using $\Msg'$ and 
    $\FP$ is a consistent, sound, $(c,\detR)$-robust multi-user watermarking scheme.
\end{theorem}
\begin{proof}
This follows as an immediate consequence of \lemref{lem-msg-robust-detectable} following
the same reasoning used in \thmref{thm-undetectable}.
\end{proof}

\subsection{Can we do better?}
\label{sec:detectable-discussion}
We briefly describe two approaches to reducing the robustness parameter $k$
in the absence of undetectability.

\paragraph{Bounded or partial undetectability}
Our analysis allowed for a worst-case adversary who could perfectly distinguish
blocks marked under different keys. But even watermarking schemes that are not
undetectable are not so blatantly detectable.
For example, the scheme of \cite{ARXIV:KTHL23} is undetectable for any single
query (``distortion-free''), and the red-/green-list scheme of \cite{ARXIV:ZALW23}
guarantees that the Renyi divergence between the marked and unmarked
distributions for any single token is bounded.

It's not clear how to use these limited distinguishing guarantees to build
robust $L$-bit / multi-user watermarking schemes. The critical step in the proof of
Lemma~\ref{lem-msg-robust} is to bound the fraction $\delta$ of empty bins
after $k$ balls are thrown into $L$ bins. With undetectability, the balls are thrown uniformly.
One idea is to bound the distinguishing advantage of an adversary making $q$ queries
with $c$ marks and producing $B$ blocks of text as only growing
polynomially in $q$, $c$, or $B$. Then, use that result to conclude that the
induced distribution of balls-into-bins is not too far from uniform.
The Renyi divergence bounds in \cite{ARXIV:ZALW23} do not seem strong enough
to make this approach work, even for a single generation.
Even if this idea worked, proving adaptive robustness would still be challenging.

\paragraph{Heuristically duplicating keys}
Without undetectability, 
the adversary may be able to tell
whenever the same key is used to watermark a piece of text. This allows (in our construction)
the adversary to only include blocks watermarked under a few keys.
In our proof, this is prevented by undetectability. Our main results
($k = O(L\lambda)$) should hold so long as the adversary never sees two blocks
generated using the same key.

Towards that end, one could try key duplication. Generate poly-many keys (instead of 
just one) for each index-bit $(i,b)$ pair, sampling a random key from this set at
every iteration of $\Wat$. This will reduce the number of key collisions observed
by the adversary. Though it would not make the probability of collision negligible,
it would be possible to bound the number of collisions as a function of the number
of blocks $B$ observed. Combined with the previous approach, this may suffice.
Even if not, key duplication may improve practical security for
schemes that are particularly detectable.
However, key duplication comes with a proportional 
increase in the runtime of our detection algorithm, which checks every possible key.

\fi

\ifx\cols\one
\section*{Acknowledgements}
We thank Miranda Christ and Sam Gunn for pointing out an error in an earlier version of this paper (Section~\ref{sec:non-adaptive-to-adaptive}).
Aloni Cohen and Gabe Schoenbach were supported in part by the DARPA SIEVE program under Agreement No.\ HR00112020021. Any opinions, findings and conclusions or
recommendations expressed in this material are those of the authors and do not necessarily reflect the views of DARPA.
\fi

\ifx\cols\two
\bibliographystyle{IEEEtran}
\fi
\ifx\cols\one
\bibliographystyle{alpha}
\fi
\bibliography{others,../cryptobib/abbrev0,../cryptobib/crypto}

\ifx\cols\one
\appendix
\fi
\ifx\cols\two
\appendices

\fi

\section{Implications for existing LLM watermarking schemes}\label{sec:block-appendix}
\ifx\cols\one
As we show next, existing watermarking schemes can be viewed as block-by-block schemes. 
The only meaningful restriction imposed by Definition~\ref{def:block-by-block-embedding} 
is on the syntax of the robustness condition $R$.
Notice that one could in some trivial sense view all complete
schemes as block-by-block by considering an entire generation a block.
This triviality is unhelpful for our black-box constructions and
unnecessary for most schemes, but illustrates how broad our framework is in general.

\paragraph{Notation}
For a string $T = \tau_1 \tau_2 \dots \tau_{|T|}$, we define
$T_{i:j} := \tau_i \dots \tau_j$, $T_{\le k} := \tau_{1:k}$.
We write $\tau \in T$ to mean that $\tau$ is a \emph{substring} of $T$, i.e.
$\tau = T_{i:j}$ for some $i$, $j$.

\subsection{Intuition for how existing schemes work}\label{sec:intuition}
We give a brief intuition for the two dominant approaches to watermarking language models in prior work. Not all schemes fit into these categories, including \cite{EPRINT:FGJMMW23}, which we discuss 
\ifx\cols\one
below.
\fi
\ifx\cols\two
in the full version of the paper.
\fi

\paragraph{Derandomizing and measuring correlation}
One class of schemes work by derandomizing the language model using the secret key 
and then detecting the effects of this derandomization 
\cite{EPRINT:ChrGunZam23,Aar22,ARXIV:KTHL23,EPRINT:ChrGun24,ARXIV:GM24}. Because a probability distribution 
can be derandomized without being noticeably altered, these schemes enjoy some 
level of undetectability.

At a very high level, the derandomization schemes work as follows.
Language models generate text token-by-token. Let $\prompt$ be a prompt and $T_{< i}$ 
be an already-generated prefix.
In the unmarked model, the next token $\tau_{i}$ is sampled according to some 
distribution $p_\prompt(\cdot | \prompt\|T_{< i})$. 
The marked model is the same, except that a secret sequence 
$\sigma_1 \sigma_2 \dots \sigma_\ell$ of (pseudo-)random bits is used to 
derandomize the sampling of $\tau_{i}$ in a way that induces a correlation between 
$\sigma_{i}$ and $\tau_{i}$. We discuss how particular schemes derandomize $p_\prompt$
in Appendices \ref{CG-embedding} and \ref{ap:kthl}. Each of these schemes 
differ in how the secret sequence is derived, how it is used to derandomize the next token, 
and how the induced correlations are measured and used to detect.

\paragraph{Red/green lists}
Another class of statistical watermarking schemes bias the sampling of tokens using using so-called red and green lists, determined using a hash function \cite{ICML:KGWKMG23,ARXIV:ZALW23}. Tokens in the 
green list are sampled more often compared to the unmarked 
language model, and tokens in the red list are sampled less often. 
Depending on the scheme, these red and green lists can be solely determined by 
a secret key \cite{ARXIV:ZALW23} or also nearby tokens
\cite{ICML:KGWKMG23}.

To detect whether a watermark is present within text $\hT$, one can check the 
proportion of green list tokens in $\hT$. If the green list contains half the tokens sampled uniformly at random, say, then unmarked text should have close to 50\% green tokens. Marked text will have many more green tokens. Detection works by testing whether the proportion of green tokens in any substring is above a statistically significant threshold. Soundness and robustness are proved using concentration bounds on the expected number of green tokens in unmarked and marked text, respectively. However, these schemes are not undetectable. By design, green tokens are noticeably more likely in the watermarked model. With poly-many queries, an adversary could conceivably reconstruct the lists in full, though this seems very costly in practice.

\fi
\ifx\cols\two
The full version of this paper provides a brief intuition for two dominant 
approaches to watermarking language models, as well as proofs that many 
existing schemes are AEB-robust. In \secref{sec:CGZ}, we show how a zero-bit scheme from \cite{EPRINT:ChrGunZam23} can be cast in our block-by-block language, and that it is adaptively AEB-robust.
\fi

\subsection{Undetectable watermarks \cite{EPRINT:ChrGunZam23}}\label{sec:CGZ}
The zero-bit watermarking scheme of Christ, Gunn and Zamir \cite{EPRINT:ChrGunZam23} 
is easily cast as a block-by-block scheme, with all the provable properties needed 
to invoke our constructions: undetectability, soundness, completeness, and robustness.
The robustness guarantee is called \emph{$b(\ell)$-substring completeness}. The 
construction $\wat$ from \cite[Algorithms 5-6]{EPRINT:ChrGunZam23} 
is a $\bigl(\frac{8}{\ln 2}\lambda\sqrt{\ell}\bigr)$-substring 
complete watermarking scheme \cite[Theorem 8]{EPRINT:ChrGunZam23}, assuming one-way functions exist. Robustness is 
guaranteed to hold for generations with enough empirical entropy, with the amount 
required depending on the length of the generation.

\begin{definition}[Empirical entropy \cite{EPRINT:ChrGunZam23}]
    \label{def:empirical-entropy}
    For strings $\tau,\prompt \in \tokSet^*$ and model $\GenModel$, 
    we define the \emph{empirical entropy of $\tau$ with respect to $\GenModel$ and $\prompt$} 
    as $\eH(\tau;\prompt) := -\log \Pr[\GenModel(\prompt)_{\le|\tau|}=\tau]$.
\end{definition}

\begin{definition}[Substring completeness \cite{EPRINT:ChrGunZam23}]\label{def:CGZ-robustness}
A watermarking scheme $\wat$ is \emph{$b(\ell)$-substring complete} if for every 
prompt $\prompt$ and security parameter $\lambda$:

\ifx\cols\one
\begin{align*}
    \Pr_{\substack{\sk \gets \Setup(1^\lambda) \\ T \gets \Wat_\sk(\prompt)}}\biggl[
    \exists \text{ length-$\ell$ substring } \tau \in T ~:
    ~ \underbrace{\eH(\tau;\prompt) \ge b(\ell)}_{\text{enough entropy}}
    \text{ and } \underbrace{\Detect_\sk(\tau) = 0}_{\text{detection fails}}  \biggr]
    < \negl(\lambda).
\end{align*}
\fi
\ifx\cols\two
\begin{align*}
    \Pr_{\substack{\sk \gets \Setup(1^\lambda) \\ T \gets \Wat_\sk(\prompt)}}&\biggl[
    \exists \text{ length-$\ell$ substring } \tau \in T ~:\\
    &~ \underbrace{\eH(\tau;\prompt) \ge b(\ell)}_{\text{enough entropy}}
    \text{ and } \underbrace{\Detect_\sk(\tau) = 0}_{\text{detection fails}}  \biggr]\\
    &< \negl(\lambda).
\end{align*}
\fi

\end{definition}
To detect a watermark on input $\hT$, the construction $\Detect_\sk(\hT)$ 
outputs the OR of $\Detect_\sk(\htau)$ for all substrings $\htau \in \hT$. 
By substring completeness, $\Detect_\sk(\hT) = 1$ if there exists a substring 
$\tau\in T$ that satisfies the following two conditions: 
\begin{enumerate}[\label=(i)]
    \item $\eH(\tau;\prompt) \ge b(|\tau|)$.
    \item There exists a substring $\htau \in \hT$ for which $\htau = \tau$.
\end{enumerate}

The watermarking scheme $\wat$ described above is naturally viewed as a 
block-by-block scheme. Condition (i) defines the blocks:\ifx\cols\two\\\fi
$\isBlock_{\sub{CGZ}}(\tau;\prompt) = \1\bigl(\eH(\tau;\prompt)\ge b(|\tau|)\bigr)$.
Condition (ii) tells us that the binary relation $\bumpeq$ on strings 
$\tau$ and $\htau$ is string equality.  Let $\RCGZ$ be the AEB-robustness condition
induced by the function $\isBlock_{\sub{CGZ}}$ and the string equality relation, 
according to \defref{def:R-k}.

\begin{clm}\label{clm:CGZ-R-na-robust}
Let $\wat$ be the $b(\ell)$-substring complete watermarking scheme from 
\cite{EPRINT:ChrGunZam23}. Then $\wat$ is non-adaptively $\RCGZ$-robust. 
\end{clm}
\begin{proof}
    Fix a prompt $\prompt \in \bits^*$ of length $|\prompt| \leq \poly(\lambda)$ and 
    efficient adversary $\A$. To show that $\wat$ is non-adaptively $\RCGZ$-robust, 
    it suffices to show the following:
    \[
        \Pr
        \left[ \
        \Detect_\sk(\hT) = 1 ~\mid~ \RCGZ(\lambda, Q, T, \hat{T}) = 1
        \right] \geq 1 - \negl(\lambda),
    \]
    where $\sk \gets \Setup(1^\lambda)$, $T \gets \Wat_\sk(\prompt)$, and
    $\hT \gets \A(1^\lambda, T)$.

    By definition of $\RCGZ$, there exist substrings 
    $\tau \in T, \htau \in \hT$ such that 
    (i) $\isBlock_\sf{CGZ}(\tau; \prompt) = 1$, and (ii) $\htau = \tau$.
    Because $\wat$ is $b(\ell)$-substring complete,
    $\Detect_\sk(\htau) = \Detect_\sk(\tau) = 1$ with high probability. 
    By construction $\Detect(\hT) = 1$ with high probability, so $\wat$ is 
    non-adaptively $\RCGZ$-robust.
\end{proof}

\subsubsection{$\wat$ is adaptively robust}\label{ap:adaptive}
We now prove that $\wat$ is
in fact adaptively robust (strengthening \cite[Theorem 5]{EPRINT:ChrGunZam23}).
We suspect some other constructions from prior work are adaptively robust, but do not
work out the details of extending their existing proofs.

In \cite[Theorem 5]{EPRINT:ChrGunZam23}, the authors replace the pseudorandom function (PRF) with a random oracle and
prove the desired watermarking properties over the randomness of the oracle.
They then transition from their random oracle construction to the
PRF construction to achieve the desired results.
To extend the proof, we recall some of the notation of oracle (zero-bit)
watermarking schemes. Notice we do not need a $\Setup$ function, because the random
oracle will serve as the secret key.

\begin{definition}[Oracle watermarks]\label{def:oracle-wat}
An \emph{oracle watermarking scheme} $\wat^{\calO}$ is pair of algorithms
$(\Wat^{\calO},\Detect^{\calO})$ with access to a random oracle $\calO$.
\begin{itemize}
    \item $\Wat^{\calO}(\prompt) \to T$ is a randomized algorithm that
    takes a prompt $\prompt$ as input and outputs a string
    $T\in \tokSet^{*}$.
    \item $\Detect^{\calO}(T) \to b$ is a deterministic
    algorithm that takes a string
    $T\in \tokSet^{*}$ as input and outputs a bit $b \in \bits$.
\end{itemize}
We define additional properties of watermarking schemes
(block-by-block, soundness, undetectability, AEB-robustness)
as before, with respect to an adversary that does not know the random oracle.
\end{definition}

Our goal is to prove that the oracle watermarking scheme of
\cite{EPRINT:ChrGunZam23} is adaptively robust, which will imply that their
PRF-based scheme is adaptively robust (assuming the security of the PRF). For our proof, we
only need the following property to hold.

\begin{definition}[Unpredictable oracle watermarks]\label{def:unpredictable}
An oracle watermarking scheme $\wat^{\calO}=(\Wat^{\calO},\Detect^{\calO})$
is \emph{unpredictable} if for every $\A^{\calO}$
issuing most $\poly(\lambda)$ queries, $X_\prompt$
and $X_{\A}$ are disjoint with high probability,
in the probability experiment defined by
\begin{itemize}
    \item $\prompt \gets \A^{\calO}$ with $X_{\A} = \{\text{Query inputs made by }\A^{\calO}\}$
    \item $X_\prompt = \{\text{Query inputs made by }\Wat^{\calO}(\prompt)\}$.
\end{itemize}
\end{definition}

Notice that the counterexample scheme $\wat'$
from \secref{sec:non-adaptive-to-adaptive} is ``predictable.'' An
algorithm $\A$ could use $\Wat'^\calO$ to generate a watermarked output $T$, since it knows both
$\calO$ and $\Wat'$. Then while running $\Detect$, $\Wat'^{\calO}(T)$ would
query the same oracle point with high probability.
Fortunately, however, the scheme from \cite{EPRINT:ChrGunZam23} does
satisfy \defref{def:unpredictable}.

\begin{lemma}\label{lem:CGZ-direct}
The oracle watermarking scheme $\wat^{\calO}=(\Wat^{\calO},\Detect^{\calO})$
from \cite{EPRINT:ChrGunZam23} is unpredictable.
\end{lemma}
\begin{proof}
The unpredictability can be seen almost immediately from the prior work's
construction, by inspecting the proofs of \cite[Theorems 5-6]{EPRINT:ChrGunZam23}. 
Each $x \in X_\prompt$
is a block of generated text with at least $\lambda$-bits of
empirical entropy (\defref{def:empirical-entropy}). 
Hence, for all $x' \in X_\A$ we have $\Pr[x = x'] \le 2^{-\lambda}$.
Because $|X_{\A}| \le \poly(\lambda)$ and $|X_{\prompt}|\le \poly(\lambda)$, 
the probability of a collision is negligible.
\end{proof}

\begin{lemma}[Adaptivity from unpredictable watermarks]\label{lem:direct-to-adapt}
Let $\wat^{\calO}=(\Wat^{\calO},\Detect^{\calO})$ be a block-by-block zero-bit oracle
watermarking scheme that is unpredictable. If $\Msg$ is non-adaptively $R_1$-robust, then
it is adaptively $R_1$-robust.
\end{lemma}
\begin{proof}
For any $R$, non-adaptive $R$-robustness for any fixed prompt implies non-adaptive
$R$-robustness for any distribution over prompts. The latter implies adaptive
$R$-robustness for any adversary $\A_1$ making only a \emph{single query} to its oracle 
$\Wat^\calO(\cdot)$,
as the queried prompt is sampled from a fixed distribution independent of $\calO$.
It remains to show that adaptive $R_1$-robustness for single query implies
adaptive $R_1$-robustness for any $q = \poly(\lambda)$ queries.

Let $\A^{\Wat^\calO(\cdot)}$ be an adaptive-robustness adversary
 making $q$ queries for which $\Pr[\text{\fail}] \ge f$,
for $q = \poly(\lambda)$.  
We define $\A_1^{\Wat^{\calO}(\cdot)}$ making a single
query as follows:

\begin{enumerate}
    \item Sample a random function $\calO'$.\footnote{To make this efficient, the function $\calO'$ can be lazily sampled in response to $\A$'s queries.}
    \item Sample $j^* \getsr [q]$ uniformly at random.
    \item Run $\A$, responding to its oracle query $Q_i$ as follows:
    \begin{enumerate}
        \item For $i \neq j^*$, return $T_i \gets \Wat^{\calO'}(Q_i)$.
        \item For $i = j^*$, issue challenge query $Q_{j^*}$ and return
        $T_{j^*} \gets \Wat^{\calO}(Q_{j^*})$.
    \end{enumerate}
    \item When $\A$ outputs $\hT$, output $\hT$.
\end{enumerate}
\newcommand{\collide}{\mathtt{COLL}}
Notice that the view of $\A$ is identical between
this game and the adaptive robustness game (Definition~\ref{def-robustness-zero-adaptive}) unless there is a collision ---
some input in common queried to
$\calO'$ and to $\calO$. Define $\collide$ as the event that there exists an oracle input that was queried more than once
across the entire game. In particular, if $\collide$ does not occur, then there
is no input in common between those queried to $\calO$ and to $\calO'$.

Consider a simplified game, where all of $\A$'s queries are answered with the
same random oracle $\calO$. In such a game
$\Pr_{\calO}[\collide] \le \negl(\lambda)$. Otherwise, we could construct $\B^{\calO}$
in the unpredictability game (Def.~\ref{def:unpredictable}) that simulates $\Wat^{\calO}$ for $\A$ and
submits one of $\A$'s queries uniformly at random. Then, we can let $i\le j$ be the
first indices of $\A$'s queries which query the oracle on the same input $x$.
With probability at least $1/q$, $\B$ will submit the $j$th query and therefore will
have win the unpredictability game with probability at least $\Pr_{\calO}[\collide] / q$.

In the argument in the previous paragraph, there is only one random oracle $\calO$. However,
when $\A_1$ simulates $\A$, there are two random oracles: $\A_1$ simulates $\calO'$ to answer all but one of
$\A$'s queries, and $\calO$ is used for query $j^*$. Even so, the argument still works.
Before the first repeated random oracle input, all random oracle outputs are uniform and independent
bit-strings in both games. And therefore, the view of $\A$ is the same in both games, as is $\B$'s success probability.

The argument above allows us to focus just on the case when $\collide$ does not occur and
bound with respect to the event
\[
    \mathcal{E} := \left(\Detect^{\calO}(\hT) = 0\right)  \land \neg \collide.
\]
By law of total probability, the definition of $f$, and the argument above:
\begin{align*}
    f &= \Pr_{\calO}\left[R_1\bigl(\lambda, (Q_i)_i, (T_i)_i, \hT\bigr)=1
    ~\land~ \Detect^{\calO}(\hT) = 0\right]\\
    &\le \Pr[\collide] + \Pr_{j^*,\calO,\calO'}\left[R_1\bigl(\lambda, (Q_i)_i, (T_i)_i, \hT\bigr)=1
    ~\land~ \mathcal{E} \right].
\end{align*}
Notice that this inequality uses the fact that we can switch sample spaces from
the game with one random oracle to the game with two random oracles as long as we
condition on the event $\neg \collide$ in both games.

To complete the proof, we recall $R_1(\lambda, (Q_i)_i, (T_i)_i, \hT) \allowbreak= 1$ only if a substring of $\hT$ approximates a block in $\cup_{i} \Blocks(T_i;Q_i)$. 
Hence if
$R_1(\lambda, (Q_i)_i, (T_i)_i, \hT) = 1$, there exists $i^* \in [q]$ such
that $R_1(\lambda,Q_{i^*}, T_{i^*}, \hT) = 1$. By construction, $j^* = i^*$
with probability $1/q$. Therefore,
\begin{align*}
    f 
    &\le 
    q\cdot \Pr_{j^*,\calO,\calO'}\left[R_1\bigl(\lambda, Q_{j^*}, T_{j^*}, \hT\bigr)=1
    ~\land~ \mathcal{E} \right] + \negl(\lambda) \\
    &< \negl(\lambda).
\end{align*}
The last inequality follows from the fact that $\wat^{\calO}$ is non-adaptively
robust and $\A_1$ issues a single query.
\end{proof}

\begin{clm}\label{clm:CGZ-R-a-robust}
    Let $\wat$ be the $b(\ell)$-substring complete watermarking scheme from 
    \cite{EPRINT:ChrGunZam23}. Then $\wat$ is adaptively $\RCGZ$-robust. 
\end{clm}
\begin{proof}
    The result follows from \clmref{clm:CGZ-R-na-robust},
    \lemref{lem:CGZ-direct}, and \lemref{lem:direct-to-adapt}.
\end{proof}

\ifx\cols\one
\subsection{Watermarking from pseudorandom codes \cite{EPRINT:ChrGun24}}\label{sec:CG}
The watermarking schemes of Christ and Gunn \cite{EPRINT:ChrGun24} can also be 
cast as block-by-block schemes. We will focus on their secret-key schemes here, but also note that \cite{EPRINT:ChrGun24} propose variants with public
marking and public detection.
The work defines pseudorandom
error-correcting codes (PRCs), and use PRCs to derandomize the language model. 
A zero-bit PRC is a triple of randomized, efficient algorithms
$\PRC = (\sf{KeyGen}, \sf{Encode}, \sf{Decode})$. $\sf{KeyGen}$ generates 
a secret key $\sk$. $\sf{Encode}_\sk(1)$ generates codewords $c$ of length $n$
that are pseudorandom to any efficient adversary without $\sk$. $\sf{Decode}_\sk(T)$ 
attempts to decode a string $T$. The PRC is 
\emph{robust against a channel $\channel$} if
\begin{enumerate}[\label=(a)]
    \item For any fixed $T$ independent of $\sk$, $\sf{Decode}_\sk(T) = \bot$ with high probability over $\sk$.
    \item $\sf{Decode}_\sk(\channel(\sf{Encode_\sk(1)}) = 1$ with high probability over $\channel$.
\end{enumerate}
Christ and Gunn construct PRCs that are robust to 
$p$-bounded channels, assuming either
$2^{O(\sqrt{n})}$-hardness of Learning Parity with Noise (LPN)
or polynomial hardness of LPN and the planted XOR problem at low density.

\begin{definition}[$p$-Bounded channels]
    For any $p \geq 0$, a length-preserving channel $\channel: \Sigma^* \to \Sigma^*$ 
    is \emph{$p$-bounded} if there exists a negligible function $\negl$ such that 
    for all $n \in \N$, $\Pr_{x \gets \bits^n}[\Delta(\channel(x), x) > p] < \negl(n)$,
    where $\Delta$ is the normalized Hamming distance. 
\end{definition}

The watermarking robustness guarantee of \cite[Definition 13]{EPRINT:ChrGun24}
is called \emph{substring robustness against a channel $\channel$}. Informally,
substring robustness guarantees that the watermark will be detected even from a 
(sufficiently entropic) cropped string that has been corrupted by $\channel$. 
Substring robustness generalizes substring completeness 
(\defref{def:CGZ-robustness}): a scheme that is substring robust against the 
identity channel $\mathcal{I}(T) = T$ is substring complete. In the next 
section, we will show that the zero-bit watermarking scheme 
$\CGscheme = (\sf{KeyGen}, \Wat, \Detect)$ from 
\cite[Construction 7]{EPRINT:ChrGun24} is a block-by-block scheme, when 
instantiated with a zero-bit pseudorandom code $\PRC$ that is robust to certain 
$p$-bounded channels. We first state two robustness results 
in the language of \cite{EPRINT:ChrGun24}:

\begin{lemma}[Lemma 22, \cite{EPRINT:ChrGun24}]\label{lem:CG-sub-complete}
    Let $\varepsilon > 0$ be any constant. If $\PRC$ is a zero-bit 
    PRC with block length $n$ that is robust to any 
    $(1/2 - \varepsilon)$-bounded channel, then $\CGscheme$ is 
    $(4\sqrt{\varepsilon} \cdot L + 2\sqrt{2}\cdot n)$-substring complete.
\end{lemma}

\begin{lemma}[Lemma 23, \cite{EPRINT:ChrGun24}]\label{lem:CG-sub-robust}
    Let $\varepsilon, \delta > 0$ be any constants. If $\PRC$ is a zero-bit 
    PRC with block length $n$ that is robust to any 
    $(1/2 - \varepsilon\cdot \delta)$-bounded channel, then $\CGscheme$ is 
    $(4\sqrt{\varepsilon} \cdot L + 2\sqrt{2}\cdot n)$-substring robust against
    $\sf{BSC}_{1/2 - \delta}$, the binary symmetric channel with error rate 
    $1/2 - \delta$.
\end{lemma}

\subsubsection{$\RCG$-robustness of $\CGscheme$}\label{CG-embedding}
We begin by describing how $\CGscheme$ embeds and detects watermarks. 
Let $n$ be the length of PRC codewords and $L^*$ be an upper bound on the length token sequences generated by $\GenModel$.\footnote{
    \cite{EPRINT:ChrGun24} (and \cite{EPRINT:ChrGunZam23}) base their constructions on a language model 
    that uses a binary token set, but show that this assumption can be made without a loss of generality.
}
The secret key watermarking skey is $(\sk, a)$
where $\sk$ is the PRC key, and $a = (a_1, \dots, a_{\lceil L^* / n\rceil})$ where each $a_i \in \bits^{n}$ is a one-time pad. To 
generate watermarked text $T$, $\Wat_{(\sk, a)}$ samples an $n$-bit string 
$x_1 \gets \PRC.\sf{Encode}_\sk(1) \oplus a_1$. The scheme then uses $x_1$ to sample the first 
length-$n$ substring $\tau_1 \in T$. Crucially, $\tau_1$ will be a noised 
version of $x_1$, 
where the amount of noise is inversely proportional to the empirical entropy 
of $\tau_1$. This procedure is iterated, producing a final output is 
$T = \tau_1 \| \dots \|\tau_r \| \sigma$ where each $\tau_i$ is correlated with $x_i$. To detect, 
$\sf{Detect}_{(\sk, a)}(T)$ outputs the OR 
of $\PRC.\sf{Decode}_\sk(\tau \oplus a_i)$ for all substrings $\tau \in T$, $a_i \in a$. By the 
robustness of $\PRC$, for all $i \in [r]$, $\PRC.\sf{Decode}_\sk(\tau_i \oplus a_i) = 1$ with high probability.

\lemref{lem:CG-sub-complete} establishes that $\CGscheme$ is $b(\ell)$-substring
complete so long as $\PRC$ is sufficiently robust. As in \apref{sec:CGZ},
this means that $\Detect_{(\sk, a)}(\hat{T}) = 1$ with high probability
if there exists a substring $\tau$ that satisfies the same conditions (i) and 
(ii) from \secref{sec:CGZ}. So one can show that $\CGscheme$ is non-adaptively 
$R^\sf{CGZ}_1$-robust via the same reasoning as in the proof of \clmref{clm:CGZ-R-na-robust}.

\lemref{lem:CG-sub-robust} establishes that $\CGscheme$ is $b(\ell)$-substring 
robust provided the underlying $\PRC$ is robust to even noisier channels. 
To translate this guarantee into our language, we can keep condition (i) exactly 
the same as before. However, unlike requiring some substring $\htau \in \hT$ to 
be \emph{identical} to $\tau \in T$ as in condition (ii), we want to relax 
our binary string relation, putting strings $\htau$ and $\tau$ in relation only if 
every length-$n$ substring 
$\tau^* \in \tau$ is no more than $(\varepsilon \cdot (1 - \delta))$-far from some 
length-$n$ substring $\hat{\tau}^* \in \hat{\tau}$, in the normalized Hamming 
distance $\Delta$. This is expressed by our (asymmetric) relation $\bumpeq_\sf{CG}$: 
\ifx\cols\one
\begin{align*}
    \htau \bumpeq_\sf{CG} \tau \iff \forall \text{ length-$n$ substrings } \tau^* \in \tau,
    \exists \text{ length-$n$ substring }\htau^* \in \htau~:~
    \Delta(\htau^*, \tau^*) \leq \varepsilon \cdot (1 - \delta).
\end{align*}
\fi
\ifx\cols\two
\begin{align*}
    \htau &\bumpeq_\sf{CG} \tau \iff \forall \text{ length-$n$ substrings } \tau^* \in \tau,\\
    &
    \exists \text{ length-$n$ substring }\htau^* \in \htau~:~
    \Delta(\htau^*, \tau^*) \leq \varepsilon \cdot (1 - \delta).
\end{align*}
\fi

Now let $\RCG$ be the AEB-robustness condition induced by the function 
$\isBlock_\sf{CG}$ (defined like $\isBlock_\sf{CGZ}$ but for a different entropy
requirement $b(\ell) = 4\sqrt{\varepsilon}\cdot \ell + 2\sqrt{2}\cdot n$) and the 
relation $\bumpeq_\sf{CG}$, according to \defref{def:R-k}. The following
claim implies that $\CGscheme$ is a block-by-block scheme when 
instantiated from a pseudorandom code $\PRC$ with block length $n$ 
that is robust to every $(1/2 - \varepsilon \cdot \delta)$-bounded channel.

\begin{clm}\label{clm:CG-R-robust}
    Let $\varepsilon, \delta > 0$ be constants. If $\PRC$ is a zero-bit PRC of 
    block length $n$ that is robust to every 
    $(1 / 2 - \varepsilon \cdot \delta)$-bounded channel, then $\CGscheme$ is 
    non-adaptively $\RCG$-robust.
\end{clm}
\begin{proof}
    Fix a prompt $\prompt \in \bits^*$ of length $|\prompt| \leq \poly(\lambda)$ 
    and efficient adversary $\A$. To show that $\CGscheme$ is non-adaptively 
    $\RCG$-robust, it suffices to show the following:
    \[
        \Pr
        \biggl[ \
        \Detect_{(\sk, a)}(\hT) = 1 ~\mid~ \RCG(\lambda, Q, T, \hat{T}) = 1
        \biggr] \geq 1 - \negl(\lambda).
    \]
    where $(\sk, a) \gets \Setup(1^\lambda)$, $T \gets \Wat_{(\sk, a)}(\prompt)$, and
    $\hT \gets \A(1^\lambda, T)$.
    As discussed above, the watermarking scheme noisily embeds strings $x_i$ 
    into $T$, where each $x_i \gets \PRC.\sf{Encode}_\sk(1) \oplus a_i$. We are 
    conditioning on $\RCG(\lambda, Q, T, \hat{T}) = 1$. By definition of $\RCG$, 
    there exist substrings $\htau \in \hT, \tau \in T$ such that (i) 
    $\isBlock_\sf{CG}(\tau; \prompt) = 1$ and (ii) $\htau \bumpeq_\sf{CG} \tau$.
    By (i) there exists some $x^* \in (x_i)_i$ and some length-$n$ substring 
    $\tau^*$ of $\tau$ such that $\Delta(\tau^*, x) \leq 1/2 - \varepsilon$ with 
    high probability.\footnote{
        Condition (i) implies a lower bound on the empirical entropy of $\tau$. 
        We can then use the reasoning in the proof of \cite[Lemma 22]{EPRINT:ChrGun24} 
        to show the existence of some substring $\tau^* \in \tau$ of length $n$ with
        enough empirical entropy to apply \cite[Lemma 21]{EPRINT:ChrGun24}, 
        which provides the desired high-probability guarantee on the 
        $\Delta$-distance between $\tau^*$ and $x$.
    }
    By (ii) and the definition of $\bumpeq_\sf{CG}$ there exists some length-$n$ 
    substring $\htau^*$ of  $\htau$ such that 
    $\Delta(\htau^*, \tau^*) \leq \varepsilon \cdot (1 - \delta)$.
    By the triangle inequality, (i) and (ii) together give us 
    $\Delta(\htau^*, x^*) \leq \frac 1 2 - \varepsilon \cdot \delta$ with high 
    probability. 
    
    We know by hypothesis that $\PRC$ is robust to every \emph{fixed}
    $(\frac 1 2 - \varepsilon \cdot \delta)$-bounded channel, i.e., bounded channels that
    are independent of the PRC key $\sk$. By construction, $x^* = c^* \oplus a^*$
    for some PRC codeword $c^*$ and one-time pad $a^*$. Let $\channel$ be the channel 
    that takes $c^*$ to $\htau^* \oplus a^*$ (that is, the output of $\PRC.\mathsf{Encode}$ to the input of $\PRC.\mathsf{Decode}$). We know that $\channel$ is 
    $(\frac 1 2 - \varepsilon \cdot \delta)$-bounded because
    $\Delta(\htau^* \oplus a^*, c^*) = \Delta(\htau^*, x^*) \leq \frac 1 2 - \varepsilon \cdot \delta$.
    It remains to show that $\channel$ is independent of $\sk$. 
    
    To produce an output, $\channel$ noisily embeds $x^* = c^* \oplus a^*$ into the generated string $\tau^*$ and then adversarially modifies $\tau^*$ to output $\htau^*$.
    Observe that neither the noisy embedding procedure nor the adversary depend on the PRC secret key $\sk$, except insofar as the input $x^*$ depends on $\sk$.
    Over the randomness of the one-time pad $a^*$, $x^*$ is a uniformly random string, independent of $\sk$.
    Thus errors induced by $\channel$ on $c^*$ are independent of $\sk$. 
    As a result, $\PRC$ is 
    robust to $\channel$ and so $\PRC.\sf{Decode}_{\sf{sk}}(\htau^*\oplus a^*) = 1$ with 
    high probability.
    By construction $\Detect_{(\sk, a)}(\hT) = 1$ with high probability, so $\CGscheme$
    is non-adaptively $\RCG$-robust.
\end{proof}

Christ and Gunn also build $L$-bit watermarking schemes, using $L$-bit PRCs 
that whose encoding functions take messages $m \in \bits^L$ rather than the 
single message $m \in \{1\}$. \lemref{lem:CG-sub-robust} applies identically 
for an $L$-bit watermarking scheme, so all of the robustness results 
described in this section extend to $L$-bit watermarking schemes.

\subsection{Publicly detectable watermarks \cite{EPRINT:FGJMMW23}}\label{ap:f}
The zero-bit watermarking scheme of Fairoze, Garg, Jha, Mahloujifar, Mahmoody, 
and Wang \cite{EPRINT:FGJMMW23} can also cast as a block-by-block scheme. It was
the first proposed publicly detectable scheme. In other words,
the $\Detect$ algorithm only
requires a public key, and knowledge of the public key does not undermine 
undetectability, soundness, completeness, nor robustness. When used in our 
constructions, the resulting $L$-bit and multi-user watermarking schemes are 
also publicly detectable.

Unlike \cite{EPRINT:ChrGunZam23}, the proofs of undetectability and computational 
efficiency in \cite{EPRINT:FGJMMW23} require assuming that each block of 
$\ell = \ell(\lambda)$ tokens produced by $\GenModel$ has at least $\lambda$ bits 
of min-entropy. The parameter $\ell$ is assumed to be known and is used in the 
construction.

\begin{assumption}[Assumption 2.1 of \cite{EPRINT:FGJMMW23}]\label{FGJ-assumption}
For any prompt $\prompt$ and string $T$, $\Pr[\GenModel(\prompt)_{1:\ell} = T] \le 2^{-\lambda}$.
\end{assumption}

\begin{definition}[$d$-Robustness, adapted from Definition 4.5 of \cite{EPRINT:FGJMMW23}]\label{def:robustness:FGJ}
A publicly-detectable watermarking scheme is $d$-robust if for every prompt $\prompt$, 
security parameter $\lambda$, and PPT $\A$,
\[
\Pr\left[
    \Detect_\pk(\hat{T}) = 0
    ~:~ 
\begin{array}{c} 
    (\sk,\pk)\getsr \Setup(1^\lambda) 
    \\ 
    T \gets \Wat_\sk(\prompt)
    \\
    \hat{T} \gets \A(\pk,T)
\end{array}\right]
\leq \negl(\lambda)
\]
where $\A$ is required to output $\hT$ that contains a substring $\htau$ 
of length at least $d$ that is also a substring of $T$.    
\end{definition}
Using a signature scheme with pseudorandom $\ell_{\sub{sig}}$-bit signatures, 
\cite{EPRINT:FGJMMW23} gives a $(2\ell(1+\ell_{\sub{sig}}))$-robust watermarking 
scheme, where $\ell$ is from the min-entropy assumption above.
To turn this into a block-by-block scheme, we take 
$\isBlock_{\sub{FGJ+}}(\tau;\prompt) = \1(|\tau| \ge d)$, and the binary relation 
$\bumpeq$ on strings $\htau$ and $\tau$ to be string equality. Let $\RF$ 
be the AEB-robustness condition induced by the function $\isBlock_{\sub{FGJ+}}$ and 
the string equality relation, according to Definition~\ref{def:block-by-block-embedding}. 
(Note the similarity to Sec~\ref{sec:CGZ}.)
The following claim implies that the scheme of \cite{EPRINT:FGJMMW23} is a
block-by-block scheme, with $d = 2\ell(1+\ell_{\sub{sig}})$.

\begin{clm}\label{clm-F-R-robust}
    Under Assumption~\ref{FGJ-assumption}, if $\wat$ is an undetectable, 
    $d$-robust scheme then it is a non-adaptively $\RF$-robust scheme.
\end{clm}
\begin{proof}
    Fix a prompt $\prompt \in \tokSet^*$ of length $|\prompt| \leq \poly(\lambda)$ 
    and a PPT adversary $\A$. To show that $\wat$ is non-adaptively $\RF$-robust, 
    it suffices to show the following:
    \[
        \Pr
        \left[ \
        \Detect_\sk(\hT) = 1 ~\mid~ \RF(\lambda, Q, T, \hT) = 1
        \right] \geq 1 - \negl(\lambda).
    \]
    where $\sk \gets \Setup(1^\lambda)$, $T \gets \Wat_\sk(\prompt)$, and
    $\hT \gets \A(1^\lambda, T)$.

    By definition of $\RF$, there exist substrings $\tau \in T, \htau \in \hT$ such that 
    (i) $\isBlock_\sf{FGJ+}(\tau; \prompt) = 1$, and (ii) $\htau = \tau$. 
    Using Assumption~\ref{FGJ-assumption} and the fact that $\wat$ is $d$-robust,
    $\Detect_\sk(\hT) = 1$ with high probability, so $\wat$ is non-adaptively $\RF$-robust. 
\end{proof}

\subsection{A $p$-value for watermarks \cite{ARXIV:KTHL23}}\label{ap:kthl}
Slight variations of the zero-bit watermarking schemes of Kuditipudi, Thickstun, 
Hashimoto, and Liang \cite{ARXIV:KTHL23} can be cast as block-by-block schemes. 
In particular, the Inverse Transform Sampling (ITS) scheme from \cite[Section 2.3]{ARXIV:KTHL23} can be modified to be complete, 
sound, and robust to token substitutions.\footnote{
    It is unclear how to modify the ITS scheme to be robust to token insertions or 
    deletions, because the $k^k$ term in the statement of \cite[Lemma 2.6]{ARXIV:KTHL23}
    dominates the inverse exponential term.
}
But it is not undetectable, so our main constructions cannot be applied to it. More 
specifically, the ``distortion-free'' property achieved by the ITS scheme is 
essentially a single-query form of undetectability. Across many queries, however, 
the outputs will be highly correlated, unlike truly undetectable schemes.

At a high level, the ITS scheme $\wat = (\Setup, \Wat, \Detect)$ works as follows. 
As sketched in \apref{sec:intuition}, $\Wat_\sk(\prompt)$ outputs a generation $T$ 
by using a decoder function $\Gamma:  \keyspace \times \Delta(\tokSet) \to \tokSet$ 
to select successive tokens. The decoder $\Gamma$ uses an element of $\sk$ and 
the distribution $p_\prompt$ to deterministically select a token.
Detection relies on a test statistic $\phi: \tokSet^* \times \keyspace^* \to \mathbb{R}$
that is designed so $\phi(T, \sk)$ is small whenever $T \gets \Wat_\sk(\prompt)$.
Instead of outputting a bit, $\Detect_\sk(\hT)$ samples $s$ secret keys $\sk'$ 
independently and computes $\hp = (K + 1)/(s + 1)$, where $K$ is the number of 
times $\phi(\hT, \sk') \leq \phi(\hT, \sk)$. The fraction $\hp$ is a $p$-value relative 
to the null hypothesis that $\hT$ is not watermarked.

Our definitions require that detection only errs with negligible probability. 
Running $\Detect_\sk$ as described above would require exponential runtime to produce
negligibly small $p$-values. Instead, we can simplify $\wat$ by 
detecting watermarks only when the test statistic $\phi(\hT, \sk)$ is below some 
threshold $\thresh$:
\begin{equation}\label{simplified-detect}
    \widetilde{\Detect}_\sk(\hT) = \1\Biggl(\phi(\hT, \sk) < \thresh = - \sqrt{\frac{|\hT|\lambda}{8}}\Biggr).
\end{equation}
To achieve robustness, $\phi(\hT, \sk)$ is constructed to be small so long as 
$\hT$ is ``close enough'' to some watermarked text $T$. This is done by using an
``alignment cost'' function $d: \tokSet^* \times \keyspace^*$, where 
$\phi$ returns the minimum cost over all alignments of candidate text $\hT$ with 
substrings $\sigma \in \sk$. Below we state a robustness guarantee for the ITS scheme 
in the language of \cite{ARXIV:KTHL23}. The guarantee will depend on the observed 
token probabilities of verbatim outputs of $\Wat$, which is called the 
\emph{watermark potential}.

\begin{definition}[Watermark potential]\label{def:watermark-potential}
    Given some prompt $\prompt$, the \emph{watermark potential} 
    $\alpha: \tokSet^m \to \mathbb{R}$ of some text $T = \tau_1 \dots \tau_{m}$ 
    relative to the distribution $p = p_\prompt$ is 
    \[
        \alpha(T) := \frac 1 m \sum_{i=1}^{m}\Bigl( 1 - p(\tau_i \mid \prompt \| T_{<i})\Bigr).
    \]
    Furthermore, we define $\hat{\alpha}: \tokSet^m \times \tokSet^m \to \mathbb{R}$
    as
    \ifx\cols\one
    \begin{align*}
        \hat{\alpha}(T, \hT) = \frac 1 m \sum_{\{i: \tau_i = \htau_i\}} \
        \Bigl(1 - p(\tau_i \mid \prompt \| T_{< i})\Bigr) \
        - \frac 1 m \sum_{\{i: \tau_i \neq \htau_i\}}\Biggl(\frac{1}{|\tokSet| - 1}\Biggr).
    \end{align*}
    \fi
    \ifx\cols\two
    \begin{align*}
        \hat{\alpha}(T, \hT) &= \frac 1 m \sum_{\{i: \tau_i = \htau_i\}} \
        \Bigl(1 - p(\tau_i \mid \prompt \| T_{< i})\Bigr) \\ 
        &- \frac 1 m \sum_{\{i: \tau_i \neq \htau_i\}}\Biggl(\frac{1}{|\tokSet| - 1}\Biggr).
    \end{align*}
    \fi
\end{definition}

Note that for any $T$, $0 \leq \alpha(T) \leq \frac{|\tokSet| - 1}{|\tokSet|}$. 
The robustness guarantee from \cite{ARXIV:KTHL23} implies
that, even if an adversary substitutes many tokens in some watermarked 
text $T$ to create $\hT$, the expected $p$-value computed in $\Detect_\sk(\hT)$ will be small,
so long as the untouched tokens have sufficient watermark potential.
\begin{lemma}[Lemma 2.5, \cite{ARXIV:KTHL23}]\label{lem:KTHL-robust}
    Let $n, m \in \N$ with $n \geq m$, where $m$ is the length of the generation and 
    $n$ is the length of the secret key. Use the decoder $\Gamma$ from 
    \cite[Line (1)]{ARXIV:KTHL23}, alignment cost $d$ from 
    \cite[Line (2)]{ARXIV:KTHL23}, and $\phi$ from \cite[Algorithm 3]{ARXIV:KTHL23}
    with $k = m$. Let $\sk, \sk' \sim \keyspace^n$, with 
    $T = \Wat_\sk(m, p, \Gamma)$. Let $\hT \in \tokSet^m$ be conditionally 
    independent of $\sk$ and $\sk'$ given $T$.
    Then almost surely
    \[
        \Pr
        \left[ \
        \phi(\hT, \sk') \leq \phi(\hT, \sk) ~\mid~ T, \hT
        \right] \leq 2n\exp(-kC_0^2\hat{\alpha}(T, \hT)^2/2),
    \]
    where $C_0 = 1/12 + o_{|\tokSet|}(1)$ is a constant.
\end{lemma}

To describe our modified version of $\wat$ as a block-by-block scheme, it suffices 
to build a robustness condition $\RK$ that only holds when $T$ has sufficient watermark 
potential and $\hT$ is no more than $\delta$-far from $T$ in the normalized Hamming
distance $\Delta$. Let $N := |\tokSet|$ be the size of the token set of the language model. 
For any $0 \leq \delta < 1$, we define 
\[
    \isBlock_\sf{KTHL}(T; \prompt) = \1\Bigl(\alpha(T) \geq \thresh_\delta \Bigr)
\]
where
\[
    \thresh_\delta := \Biggl(\frac{1}{C_0\sqrt{2}}\cdot \sqrt{\frac{\lambda}{|T|}} \ 
    + \frac{\delta}{N - 1}\Biggr)\Biggl(1 - \delta\Biggr) + \delta
\]
and $C_0$ is the constant from \lemref{lem:KTHL-robust}. Note that the lower bound 
enforced by $\isBlock_\sf{KTHL}$ is only satisfiable when $|T| \in \Omega(\lambda)$.
We then define $\hT \bumpeq_\sf{KTHL} T \iff \Delta(\hT, T) \leq \delta$. 
Let $\RK$ be the robustness condition induced by $\isBlock_\sf{KTHL}$ and 
$\bumpeq_\sf{KTHL}$.
Whenever $\RK(\lambda, \prompt, T, \hT) = 1$, we can derive 
the following lower bound on $\hat{\alpha}(T, \hT)$.

\begin{lemma}\label{lem:lower-bound-alpha-hat}
    For any $\lambda \in \N$, $\prompt \in \tokSet^*$, and $T, \hT \in \tokSet^m$ 
    satisfying $\RK(\lambda, \prompt, T, \hT) = 1$, we have
    \begin{align*}
        \hat{\alpha}(T, \hT) \geq \frac{1}{C_0 \sqrt{2}}\cdot\sqrt{\frac{\lambda}{m}}.
    \end{align*}
\end{lemma}
\begin{proof}
    Consider any set of random variables $\mathcal{X} = \{X_1, X_2, \dots, X_m\}$,
    where each $X_i \in [0, 1]$. Let $\mu = \frac 1 m \sum_{i \in [m]} X_i$ and let 
    $\mu_{1-\delta}$ be the mean of any subset of $\mathcal{X}$ of size 
    $\lfloor m(1-\delta)\rfloor $. Then we have 
    \begin{align*}
        \mu_{1-\delta} \geq \frac{(\sum_{i \in [m]}X_i) - m\delta}{\lfloor m(1-\delta)\rfloor } \
        \geq \frac{\mu - \delta}{1 - \delta}.
    \end{align*}
    Taking $X_i = (1 - p(\tau_i \mid \prompt \| T_{<i}))$ so $\mu \geq \thresh_\delta$,
    we can lower bound $\hat{\alpha}(T, \hT)$ by
    \begin{align*}
        \hat{\alpha}(T, \hT) &\geq \frac{\thresh_\delta - \delta}{1 - \delta} - \frac{\delta}{N-1} \\
        &= \frac{1}{C_0 \sqrt{2}}\cdot\sqrt{\frac{\lambda}{m}},
    \end{align*}
    where the final equality holds by plugging in our value of $\thresh_\delta$.
    \qedhere
\end{proof}
The following claim implies that our modified version of 
the ITS watermarking scheme from \cite{ARXIV:KTHL23} is a block-by-block scheme.
\begin{clm}\label{clm-KTHL-R-robust}
    Let $\wat$ be the ITS watermarking scheme from \cite[Section 2.3]{ARXIV:KTHL23},
    using the decoder $\Gamma$ from \cite[Line (1)]{ARXIV:KTHL23}, alignment cost $d$
    from \cite[Line (2)]{ARXIV:KTHL23}, and $\phi$ from \cite[Algorithm 3]{ARXIV:KTHL23}
    with $k = m$. Define $\widetilde{\wat}$ to be $\wat$ except using the 
    $\widetilde{\Detect}$ function from \eqref{simplified-detect}. Then $\widetilde{\wat}$ is 
    sound and non-adaptively $\RK$-robust.
\end{clm}
\begin{proof}
    We will start by showing that $\widetilde{\wat}$ is sound. Fix a string 
    $T \in \tokSet^*$ of length $|T| = m \leq \poly(\lambda)$. We want to show 
    that $\widetilde{\Detect}_\sk(T)$ returns $1$ with negligible probability
    over secret keys $\sk$ of length $n$. Recall that $\widetilde{\Detect}_\sk(T)$ returns 
    $1$ if $\phi(T, \sk) < \thresh$, where 
    \[ 
        \phi(T, \sk) = \min_{\substack{\sigma \in \sk \\ |\sigma| = m}}\{d(T, \sigma)\}.
    \]
    Pick an arbitrary $\sigma \in \sk$. Then $\widetilde{\wat}$ is sound by the 
    following. Note that the second and third inequalities are not immediate, but 
    follow from the proofs of Lemmas 2.3 and 2.4 in \cite{ARXIV:KTHL23}.
    \ifx\cols\one
    \begin{align*}
        \Pr \left[\widetilde{\Detect}_\sk(T) = 1 \right] &= \Pr \left[ \phi(T, \sk) < \thresh\right] \\
        &\leq n \Pr \left[d(T, \sigma) < \thresh \right] \ 
        \quad\quad\quad\quad\quad\quad\quad\quad \text{(Union bound)} \\
        &\leq n \Pr \left[\Ex[d(T, \sigma)] - d(T, \sigma) > -\thresh \right] \quad\quad \text{($\Ex_\sk[d(T, \sigma)] = 0$)}\\
        &\leq 2n \exp\Bigl(-\frac{2(-\thresh)^2}{m(1/2)^2}\Bigr) \
        \quad\quad\quad\quad\quad\quad\quad \text{(Hoeffding's bound)} \\
        &= 2n\exp(-\lambda).
    \end{align*}
    \fi
    \ifx\cols\two
    \begin{align*}
        \Pr \left[\widetilde{\Detect}_\sk(T) = 1 \right] &= \Pr \left[ \phi(T, \sk) < \thresh\right] \\
        &\leq n \Pr \left[d(T, \sigma) < \thresh \right] \\
        & \quad\quad\quad\quad\quad\quad \text{(Union bound)} \\
        &\leq n \Pr \left[\Ex[d(T, \sigma)] - d(T, \sigma) > -\thresh \right] \\
        & \quad\quad\quad\quad\quad\quad \text{($\Ex_\sk[d(T, \sigma)] = 0$)} \\
        &\leq 2n \exp\Bigl(-\frac{2(-\thresh)^2}{m(1/2)^2}\Bigr) \\
        & \quad\quad\quad\quad\quad\quad \text{(Hoeffding's bound)} \\
        &= 2n\exp(-\lambda).
    \end{align*}
    \fi
    Next, fix a prompt $\prompt \in \tokSet^*$ of length $|\prompt| \leq \poly(\lambda)$
    and choose an efficient adversary $\A$. 
    To show that $\widetilde{\wat}$ is non-adaptively 
    $\RK$-robust, it suffices to show that
    $\widetilde{\Detect}_\sk(\hT) = 1$ with overwhelming probability, conditioned on 
    $\RK(\lambda, \prompt, T, \hT) = 1$, where $\sk \gets \Setup(1^\lambda)$, 
    $T \gets \Wat_\sk(m, p_\prompt, \Gamma)$, and $\hT \gets \A(1^\lambda, T)$. 
    Let $\sigma$ be the length-$m$ substring of $\sk$ used by $\Gamma$ to 
    generate $T$. Then we have 
    \begin{align*}
        \Pr \left[ \widetilde{\Detect}_\sk(\hT) = 1 \right] &= \
        \Pr \left[ \phi(\hT, \sk) < \thresh \right] \\
        &\geq \Pr \left[d(\hT, \sigma) < \thresh \right] \\
        &= 1 - \Pr \left[d(\hT, \sigma) \geq \thresh \right]
    \end{align*}
    All that is left is to show that $\Pr \left[d(\hT, \sigma) \geq \thresh \right]$
    is negligible. By \cite[Lemma 2.3, Observation B.1]{ARXIV:KTHL23}, we 
    have that $\Ex[d(\hT, \sigma)] = -mC_0\hat{\alpha}(T, \hT)$. Since we are 
    conditioning on the fact that $T$ and  $\hT$ pass the $\RK$-robustness condition, 
    \lemref{lem:lower-bound-alpha-hat} implies $\Ex[d(\hT, \sigma)] \leq 2\thresh$.
    Then applying Hoeffding's bound completes the proof.
    \ifx\cols\one
    \begin{align*}
        \Pr \left[d(\hT, \sigma) \geq \thresh \right] &= \ 
        \Pr \left[d(\hT, \sigma) - \Ex[d(\hT, \sigma)] \geq \thresh - \Ex[d(\hT, \sigma)] \right] \\ 
        &\leq \Pr \left[d(\hT, \sigma) - \Ex[d(\hT, \sigma)] \ 
        \geq -\thresh \right] \quad \text{ ($\thresh - \Ex[d(\hT, \sigma)] \geq -\thresh$) }\\ 
        &\leq \exp\Bigl(-\frac{2(-\thresh)^2}{m(1/2)^2}\Bigr) \\
        &< \exp(-\lambda).
    \end{align*}
    \fi
    \ifx\cols\two
    \begin{align*}
        \Pr \left[d(\hT, \sigma) \geq \thresh \right] &= \ 
        \Pr \left[d(\hT, \sigma) - \Ex[d(\hT, \sigma)] \geq \thresh - \Ex[d(\hT, \sigma)] \right] \\ 
        &\leq \Pr \left[d(\hT, \sigma) - \Ex[d(\hT, \sigma)] \ 
        \geq -\thresh \right] \\
        & \quad\quad\quad\quad\quad\quad \text{ ($\thresh - \Ex[d(\hT, \sigma)] \geq -\thresh$) }\\
        &\leq \exp\Bigl(-\frac{2(-\thresh)^2}{m(1/2)^2}\Bigr) \\
        &< \exp(-\lambda).\qedhere
    \end{align*}
    \fi
\end{proof}

\subsection{Green list / red list schemes \cite{ICML:KGWKMG23,ARXIV:ZALW23}}

The zero-bit watermarking scheme of Kirchenbauer, et al.\ \cite{ICML:KGWKMG23} is 
not undetectable, nor does it appear to enjoy the sort of provable soundness and
robustness guarantees required to apply our constructions. A heuristic version of 
our main construction applied to this scheme may work well in practice, though 
empirical analysis is well beyond our present scope.

We give details below, focusing on the simplified variant in \cite{ARXIV:ZALW23}. 
{You may safely skip the rest of this subsection.} We include it mainly to aid 
readers interested in understanding \cite{ARXIV:ZALW23}.

The core idea in the construction is to randomly partition the token set $\tokSet$ 
into a \emph{green list} $\green$ and \emph{red list} 
$\red = \tokSet \setminus \green$, and preferentially sample tokens from $\green$.\footnote{%
    The difference between \cite{ARXIV:ZALW23} and the original construction
    of \cite{ICML:KGWKMG23} is that in the original, the green and red lists
    change as a function of the preceding tokens.}
Note that by changing the distribution over tokens, the scheme is not undetectable.
The detection algorithm performs a hypothesis test on the fraction of green tokens in a text, declaring the text marked if the fraction is significantly greater than some expected threshold (i.e., rejecting the null hypothesis that the text is not marked).

Empirically and heuristically, the scheme appears well suited to a block-by-block 
interpretation. Detecting the watermark requires a text $\hT$ to contain a long 
enough substring that was (close to a substring) produced by the watermarked model.
For example, see \cite{kirchenbauer2023reliability} which refines the original 
$\Detect$ algorithm by testing every substring of the input text (among other 
improvements), and empirically analyzes robustness to paraphrasing and copy-paste attacks. 

Unfortunately, it does not appear that the scheme provides the sort of provable guarantees we need to view it as a block-by-block scheme. In brief, the scheme does not appear to simultaneously enjoy both non-trivial soundness (low false positives) for all strings and non-trivial completeness (low false negatives) for unmodified outputs of the watermarked model. Either type of error can be bounded (even negligibly small), at the cost of destroying the provable guarantee on the other. Generically turning this scheme into a block-by-block scheme seems to require simultaneously bounding both types of errors.

The construction is parameterized by constants $\gamma\in (0,1)$ and $\delta>0$. The parameter $\gamma$ governs the size of the green set: $|\green| = \gamma |\tokSet|$. The parameter $\delta$ governs the amount that the watermarked model is biased towards green tokens, as described below. The secret key is the green set: $\sk = \green$. The $\Detect$ algorithm computes a $z$-score and compares it to some threshold $\thresh$ (which may depend on the input $T$). That is, $\Detect_\sk(T) := \1\left(z_\sk(T) > \thresh(T)\right)$, where $\thresh(T)$ is a threshold and $z_\sk(T):=(\sum_{i=1}^{|T|}\1(\tau_i \in \green) - \gamma|T|)/\sqrt{|T|\gamma(1-\gamma)}$.

In particular, one would need to set the threshold $\thresh(T)$ below to satisfy both Theorems~\ref{thm:zalw-soundness} and~\ref{thm:zalw-completeness}. Consider the typical case of $\gamma=1/2$. Theorem~\ref{thm:zalw-soundness} requires $\thresh(T) \ge 64\lambda{\Cmax(T)}/{\sqrt{|T|}}$. If  $\Cmax(T) > |T|/64\lambda$, then $\thresh(T) > \sqrt{T}$ is needed.\footnote{%
    Observe that $\Cmax(T)/|T|$ is the frequency of the most common token in $T$. For natural language, $\Cmax(T) = \Omega(T)$ is typical. For example, about 7\% of the words in the Brown Corpus are ``the''. The condition $\Cmax(T) > |T|/64\lambda$ only requires that there exists a token with frequency $1/64\lambda$ in $T$.}
But Theorem~\ref{thm:zalw-completeness} requires $\thresh(T) < \sqrt{T}$, as $\kappa < 1$ and $\gamma = 1/2$. One cannot have both, regardless of $\lambda$ and the error rates $\alpha, \beta$.

Given $\GenModel$, prompt $\prompt$, and string $T$, let $p_i$ be the probability 
distribution over token $i$ in the output of $\GenModel(\prompt)$ conditioned on $T_{1:i-1}$. 
\[
p_i(\tau ~|~ \prompt\|T_{<i}):= \Pr[\GenModel(\prompt\|T_{<i})_1 = \tau].
\]
The marked model $\Wat$ samples the next token $\tau$ with probability 
$p_i^\green(t | Q\|T_{1:i})$, defined as:
\[
p^\green_i(\tau~|~ \prompt\|T_{<i}) \propto \exp(\delta\cdot \1(\tau \in \green))\cdot p(\tau~|~ \prompt\|T_{<i}).
\]
In other words, $p^\green$ is defined by  upweighting the probabilities of $\tau \in \green$ by a factor of $e^\delta$ and the distribution is renormalized (equivalently, adding $\delta$ to the logits and computing the soft-max).

The soundess guarantee of  \cite{ARXIV:ZALW23} is given in terms of two functions of a string $T \in \tokSet^*$. These functions $\Cmax$ and $\Vmax$ both take values in $[0,|T|]$. The function $\Cmax$ is the more important function for our purposes: it counts the number of occurences of the most frequent token in a string $T$.
\ifx\cols\one
\[
\Cmax(T):=\max_{\tau \in \tokSet}\sum_{i \in |T|}\1(\tau_i = \tau) \quad
\Vmax(T) := \frac{1}{|T|} \sum_{\tau \in \tokSet} \biggl( \sum_{i \in |T|} \1(\tau_i = \tau) \biggr)^2
\]
\fi
\ifx\cols\two
\begin{align*}
    \Cmax(T)&:=\max_{\tau \in \tokSet}\sum_{i \in |T|}\1(\tau_i = \tau) \\
    \Vmax(T) &:= \frac{1}{|T|} \sum_{\tau \in \tokSet} \biggl( \sum_{i \in |T|} \1(\tau_i = \tau) \biggr)^2
\end{align*}
\fi

\begin{theorem}[Soundness, Theorem C.4 of \cite{ARXIV:ZALW23}]\label{thm:zalw-soundness}
For any $T\in \tokSet^*$
\begin{equation}
\Pr_\sk\left[z_\sk(T) > \sqrt{\frac{64\log(9/\alpha)\Vmax(T)}{1-\gamma}} + \frac{16\log(9/\alpha)\Cmax(T)}{\sqrt{|T|\gamma(1-\gamma)}} \right] < \alpha.
\end{equation}
In particular, taking $\thresh(T) \ge \sqrt{\frac{64\log(9)\lambda\Vmax(T)}{1-\gamma}} + \frac{16\log(9)\lambda\Cmax(T)}{\sqrt{|T|\gamma(1-\gamma)}}$, we get $\Pr[\Detect_\sk(T) = 1] < 2^{-\lambda}$.
\end{theorem}

Completeness requires two conditions on $\GenModel$ for the prompt $\prompt$ whose definitions we omit: on-average high entropy and on-average homophily (\cite[Assumptions C.9, C.12]{ARXIV:ZALW23}).

\begin{theorem}[Completeness, adapted from Theorem C.13 of \cite{ARXIV:ZALW23}]\label{thm:zalw-completeness}
    Fix $\GenModel$ and $\prompt$. Suppose that $\beta,\kappa \in (0,1)$ and $|T|$ satisfy the following, where $c_1$ and $c_2$ are some constants that depend on the parameters $\delta$ and $\gamma$.
\begin{itemize}
    \item $|T| \ge c_1\cdot\frac{\log(1/\beta)}{(1-\kappa)^2}$.
    \item $\GenModel$ has {$\beta$-on-average-homophily} for $\prompt$
    \item $\GenModel$ has $(\xi,\beta/3)$-on-average-high-entropy for $\prompt$, for $\xi = c_2\cdot\frac{1-\kappa}{\log^2(|T|/\beta)}$
\end{itemize}
Then
\begin{equation}
\Pr\left[z_\sk(T) < \frac{\kappa(e^\delta - 1)\sqrt{|T|\gamma(1-\gamma)}}{1+(e^\delta-1)\gamma}\right] \le \beta
\end{equation}
In particular,  $\Pr[\Detect_\sk(T) = 0] \le 2^{-\lambda}$ for threshold \ifx\cols\two\\\fi$\thresh(T) < \frac{\kappa(e^\delta - 1)\sqrt{|T|\gamma(1-\gamma)}}{1+(e^\delta-1)\gamma}$.
\end{theorem}

\fi

\ifx\cols\one
\section{Reference: definition variants}\label{ap:definitions}

\subsection{Zero-bit watermarking}

\begin{definition}[Undetectability -- zero-bit \cite{EPRINT:ChrGunZam23}]\label{def-zero-undetectability}
    A zero-bit watermarking scheme
    $\wat=(\Setup,\Wat,\Detect)$ for $\GenModel$ is
    \emph{undetectable} if for all efficient adversaries $\A$,
    \[
        \left|
        \Pr[\A^{\GenModel(\cdot)}(1^{\lambda}) = 1] -
        \Pr_{\sk \getsr \Setup(1^{\lambda})}[
            \A^{\Wat_{\sk}(\cdot)}(1^{\lambda}) = 1]
        \right|
    \]
    is at most $\negl(\lambda)$.
\end{definition}

\begin{definition}[Soundness -- zero-bit]\label{def-soundness}
    A zero-bit watermarking scheme $\wat=(\Setup,\Wat,\Detect)$
    is \emph{sound} if for all polynomials $\poly$ and all 
    strings $T\in \tokSet^*$ of length $|T|\le \poly(\lambda)$,
    \[
        \Pr_{\sk \getsr \Setup(1^{\lambda})}[\Detect_{\sk}(T) \neq 0]
        < \negl(\lambda).
    \]
\end{definition}

\noindent For all of the robustness definitions below, completeness can be 
defined by including the extra clause $\hT = T$ (or $\hT \in (T_i)_i$ in the 
adaptive setting).

\begin{definition}[$R$-Robust\textcolor{blue}{/Complete} detection -- zero-bit, non-adaptive]\label{def-robustness-zero-nonadaptive}
    A zero-bit watermarking scheme $\wat =(\Setup,\Wat,\Detect)$
    is \emph{non-adaptively $R$-robustly\textcolor{blue}{/completely} detectable} with respect 
    to the robustness condition $R$ if for all efficient adversaries $\A$, all 
    polynomials $\poly$, and all prompts $\prompt \in \tokSet^*$ of length 
    $|\prompt| \le \poly(\lambda)$, the following event \fail occurs with negligible probability:

    \ifx\cols\one
    \begin{itemize}
        \item \textcolor{blue}{$\hT = T$, AND 
        \quad \texttt{// the adversary outputs $T$}}
        \item $R(\lambda,\prompt,T,\hT)  = 1$, AND
        \quad \texttt{// the robustness condition passes}
        \item $\Detect_\sk(\hT) = 0$
        \quad \texttt{// the mark is removed}
    \end{itemize}
    \fi
    \ifx\cols\two
    \begin{itemize}
        \item \textcolor{blue}{$\hT = T$, AND \\
        \quad \texttt{// the adversary outputs $T$}}
        \item $R(\lambda,\prompt,T,\hT)  = 1$, AND\\
        \quad \texttt{// the robustness condition passes}
        \item $\Detect_\sk(\hT) = 0$ \\
        \quad \texttt{// the mark is removed}
    \end{itemize}
    \fi
    
    in the probability experiment defined by
    \begin{itemize}
        \item $\sk \gets \Setup(1^\lambda)$
        \item $T \getsr \Wat_{\sk}(\prompt)$
        \item $\hT \gets \A(1^\lambda, T)$.
    \end{itemize}
\end{definition}

\begin{definition}[$R$-Robust\textcolor{blue}{/Complete} detection -- zero-bit, adaptive]\label{def-robustness-zero-adaptive}
    A zero-bit watermarking scheme $\wat=(\Setup,\allowbreak\Wat,\Detect)$
    is \emph{adaptively $R$-robustly\textcolor{blue}{/completely} detectable} with respect to the robustness condition
    $R$ if for all efficient adversaries $\A$, the following event \fail occurs 
    with negligible probability:

    \ifx\cols\one
    \begin{itemize}
        \item \textcolor{blue}{$\hT = T$, AND 
        \quad \texttt{// the adversary outputs $T$}}
        \item $R(\lambda,(\prompt_i)_i,(T_i)_i,\hT)  = 1$, AND
        \quad \texttt{// the robustness condition passes}
        \item $\Detect_\sk(\hT) = 0$
        \quad \texttt{// the mark is removed}
    \end{itemize}
    \fi
    \ifx\cols\two
    \begin{itemize}
        \item \textcolor{blue}{$\hT = T$, AND \\
        \quad \texttt{// the adversary outputs $T$}}
        \item $R(\lambda,(\prompt_i)_i,(T_i)_i,\hT)  = 1$, AND\\
        \quad \texttt{// the robustness condition passes}
        \item $\Detect_\sk(\hT) = 0$ \\
        \quad \texttt{// the mark is removed}
    \end{itemize}
    \fi
    
    in the probability experiment defined by
    \begin{itemize}
        \item $\sk \gets \Setup(1^\lambda)$
        \item $\hT \gets \A^{\Wat_{\sk}(\cdot)}(1^\lambda)$,
        denoting by $(Q_i)_i$ and $(T_i)_i$ the sequence of inputs and outputs of the oracle.
    \end{itemize}

    We also say a scheme satisfying this definition is \emph{$(\delta,R)$-robust}.
\end{definition}

\subsection{$L$-bit watermarking}

\begin{definition}[$(\delta, R)$-Robust\textcolor{blue}{/Complete} extraction -- $L$-bit, non-adaptive]\label{def-robustness-nonadaptive}
    An $L$-bit watermarking scheme $\Msg\allowbreak =(\MsgSetup,\Encode,\Extract)$
    is \emph{non-adaptively $(\delta, R)$-robustly\textcolor{blue}{/completely} extractable} with respect to the robustness condition
    $R$ if for all efficient adversaries $\A$, all messages $\msg \in \bits^L$, all
    polynomials $\poly$, and all prompts $\prompt \in \tokSet^*$ of length
    $|\prompt| \le \poly(\lambda)$, the following event \fail occurs with negligible probability:

    \ifx\cols\one
    \begin{itemize}
        \item \textcolor{blue}{$\hT = T$, AND 
        \quad \texttt{// the adversary outputs $T$}}
        \item $R(\lambda,\prompt,T,\hT)  = 1$, AND
        \quad \texttt{// the robustness condition passes}
        \item $\hat{\msg}\not\in B_{\delta}(m)$ 
        \quad \texttt{// the mark is corrupted}
    \end{itemize}
    \fi
    \ifx\cols\two
    \begin{itemize}
        \item \textcolor{blue}{$\hT = T$, AND \\
        \quad \texttt{// the adversary outputs $T$}}
        \item $R(\lambda,\prompt,T,\hT)  = 1$, AND\\
        \quad \texttt{// the robustness condition passes}
        \item $\hat{\msg}\not\in B_{\delta}(m)$ \\
        \quad \texttt{// the mark is corrupted}
    \end{itemize}
    \fi
    
    in the probability experiment define by
    \begin{itemize}
        \item $\sk \getsr \MsgSetup(1^{\lambda})$
        \item $T \getsr \Encode_{\sk}(\msg,\prompt)$
        \item $\hT \gets \A(1^{\lambda},T)$
        \item $\hat{\msg} \gets \Extract_{\sk}(\hT)$.
    \end{itemize}
\end{definition}

\subsection{Multi-user watermarking}\label{ap:multi-user-defs}
\begin{definition}[Undetectability -- multi-user]\label{def-multi-undetectability}
    Define the oracle $\GenModel'(u, \prompt) := \GenModel(\prompt)$. A multi-user
    watermarking scheme $\wat=(\Setup,\Wat,\Detect,\Trace)$ for $\GenModel$ is
    \emph{undetectable} if for all efficient adversaries $\A$,
    \[
        \left|
        \Pr[\A^{\GenModel'(\cdot, \cdot)}(1^{\lambda}) = 1] -
        \Pr_{\sk \getsr \Setup(1^{\lambda})}[\A^{\Wat_{\sk}(\cdot, \cdot)}(1^{\lambda}) = 1]
        \right|
    \]
    is at most $\negl(\lambda)$.
\end{definition}

\begin{definition}[Soundness -- multi-user]\label{def-multi-soundness}
    A multi-user watermarking scheme $\wat=(\Setup,\Wat,\Detect,\allowbreak\Trace)$ is
    \emph{sound} if for all polynomials $\poly$ and all strings
    $T\in \tokSet^*$ of length $|T|\le \poly(\lambda)$,
    \[
    \Pr_{\sk \getsr \Setup(1^{\lambda})}[\Detect_{\sk}(T) = 1]
    \le \negl(\lambda).
    \]
\end{definition}

\fi
\ifx\cols\two
\section{Reference: robustness definition variants}\label{ap:definitions}

For any robustness definition, completeness can be 
defined by including the extra clause $\hT = T$ (or $\hT \in (T_i)_i$ in the 
adaptive setting).

\begin{definition}[$R$-Robust\textcolor{blue}{/Complete} detection -- zero-bit, non-adaptive]\label{def-robustness-zero-nonadaptive}
    A zero-bit watermarking scheme $\wat =(\Setup,\Wat,\Detect)$
    is \emph{non-adaptively $R$-robustly\textcolor{blue}{/completely} detectable} with respect 
    to the robustness condition $R$ if for all efficient adversaries $\A$, all 
    polynomials $\poly$, and all prompts $\prompt \in \tokSet^*$ of length 
    $|\prompt| \le \poly(\lambda)$, the following event \fail occurs with negligible probability:

    \ifx\cols\one
    \begin{itemize}
        \item \textcolor{blue}{$\hT = T$, AND 
        \quad \texttt{// the adversary outputs $T$}}
        \item $R(\lambda,\prompt,T,\hT)  = 1$, AND
        \quad \texttt{// the robustness condition passes}
        \item $\Detect_\sk(\hT) = 0$
        \quad \texttt{// the mark is removed}
    \end{itemize}
    \fi
    \ifx\cols\two
    \begin{itemize}
        \item \textcolor{blue}{$\hT = T$, AND \\
        \quad \texttt{// the adversary outputs $T$}}
        \item $R(\lambda,\prompt,T,\hT)  = 1$, AND\\
        \quad \texttt{// the robustness condition passes}
        \item $\Detect_\sk(\hT) = 0$ \\
        \quad \texttt{// the mark is removed}
    \end{itemize}
    \fi
    
    in the probability experiment defined by
    \begin{itemize}
        \item $\sk \gets \Setup(1^\lambda)$
        \item $T \getsr \Wat_{\sk}(\prompt)$
        \item $\hT \gets \A(1^\lambda, T)$.
    \end{itemize}
\end{definition}

\begin{definition}[$R$-Robust\textcolor{blue}{/Complete} detection -- zero-bit, adaptive]\label{def-robustness-zero-adaptive}
    A zero-bit watermarking scheme $\wat=(\Setup,\allowbreak\Wat,\Detect)$
    is \emph{adaptively $R$-robustly\textcolor{blue}{/completely} detectable} with respect to the robustness condition
    $R$ if for all efficient adversaries $\A$, the following event \fail occurs 
    with negligible probability:

    \ifx\cols\one
    \begin{itemize}
        \item \textcolor{blue}{$\hT = T$, AND 
        \quad \texttt{// the adversary outputs $T$}}
        \item $R(\lambda,(\prompt_i)_i,(T_i)_i,\hT)  = 1$, AND
        \quad \texttt{// the robustness condition passes}
        \item $\Detect_\sk(\hT) = 0$
        \quad \texttt{// the mark is removed}
    \end{itemize}
    \fi
    \ifx\cols\two
    \begin{itemize}
        \item \textcolor{blue}{$\hT = T$, AND \\
        \quad \texttt{// the adversary outputs $T$}}
        \item $R(\lambda,(\prompt_i)_i,(T_i)_i,\hT)  = 1$, AND\\
        \quad \texttt{// the robustness condition passes}
        \item $\Detect_\sk(\hT) = 0$ \\
        \quad \texttt{// the mark is removed}
    \end{itemize}
    \fi
    
    in the probability experiment defined by
    \begin{itemize}
        \item $\sk \gets \Setup(1^\lambda)$
        \item $\hT \gets \A^{\Wat_{\sk}(\cdot)}(1^\lambda)$,
        denoting by $(Q_i)_i$ and $(T_i)_i$ the sequence of inputs and outputs of the oracle.
    \end{itemize}

    We also say a scheme satisfying this definition is \emph{$(\delta,R)$-robust}.
\end{definition}

\section{More properties of our multi-user scheme}\label{ap:multiuser-easy-props}

\begin{clm}[$\wat$ is consistent]\label{clm-consistent}
    Let $L, n, c > 1$ be integers and $0 \le \delta < 1$.
    Let $\Msg'$ be an $L$-bit watermarking scheme and
    $\FP$ be a fingerprinting code.
    Then the $\wat$ construction from
    \figref{fig-multiuser} is a consistent multi-user watermarking scheme.
\end{clm}

\begin{clm}[$\wat$ is undetectable]\label{clm-undetectable}
If $\GenModel$ is
prefix-specifiable and $\Msg'$ is an $L$-bit
watermarking scheme built from an undetectable zero-bit scheme, then the $\wat$ 
construction from \figref{fig-multiuser} is undetectable.
\end{clm}

\begin{clm}[$\wat$ is sound]\label{clm-sound}
    Let $L, n, c > 1$ be integers and $0\le \delta < 1$.
    Let $\Msg'$ be a sound $L$-bit watermarking scheme and
    $\FP$ be a fingerprinting code of length
    $L$ with parameters $(\lambda,n,c,\delta)$.
    Then the $\wat$ construction from \figref{fig-multiuser} is a sound
    multi-user watermarking scheme.
\end{clm}

\fi
\ifx\cols\one
\section{Proof of \lemref{lm-bnb}}\label{sec:lm-bnb-proof}

\begin{lemma}\label{lm-bnb-apdx}
    For $\lambda,L  \ge 1$ and $0\le \delta <1$, define 
    \[
        k^*(L,\delta) = \min\left\{L\cdot (\ln L + \lambda); \quad L\cdot \ln\left(\frac{1}{\delta - \sqrt{\frac{\lambda + \ln 2}{2L}}} \right) \right\}
    \]
    Then, after throwing $k\ge k^*(L,\delta)$ balls into $L$ bins, fewer than $\delta L$ bins are empty except with probability at most $e^{-\lambda}$.
\end{lemma}
\begin{proof}
    \newcommand{\Poiss}{\mathrm{Pois}}
    If $k \ge L(\ln L +\lambda)$ balls are thrown into $L$ bins, all bins are occupied except with probability at most $L(1 - \frac 1 L)^{L(\ln L + \lambda)} < Le^{-(\ln L + \lambda)} = e^{-\lambda}$. In this case, 0 bins are empty, and the claim holds. 
    
    Now suppose $L(\ln L +\lambda) > k\ge -L\ln\left({\delta - \sqrt{\frac{\lambda + \ln 2}{2L}}} \right)$. Then $\delta > \sqrt{\frac{\lambda + \ln 2}{2L}} > 0$.
    The analysis uses the Poisson approximation to balls and bins.
    Let $X = (X_1,\dots,X_L)$ be a multinomial random variable over $\mathbb{Z}_{\geq 0}^L$ where each $X_i$ denotes the number of balls in bin $i$.
    Let $Y = (Y_1, \dots, Y_L)$ where each $Y_i \sim \Poiss(k/L)$ i.i.d.\ is Poisson with mean $k / L$.
    Let $E = \{x \in \mathbb{N}_0^L : \sum_i \1(x_i = 0) > \delta L\}$ be the event that more than $\delta L$ bins are empty.
    
    Let $W_i = \1\{Y_i = 0\}$ and $W = \sum_{i=1}^L W_i$. $\Ex[W] = Le^{-k/L}$.
    Applying a Hoeffding bound,
    \begin{align*}
        \Pr_Y[E] 
        &= \Pr[W - \Ex[W] > L(\delta - e^{-k/L})]\\
        &\leq 
        \exp\Bigl(-2L(\delta-e^{-k/L})^2\Bigr)\\
        &\le \exp(-\lambda - \ln 2)
    \end{align*}
    Where the last inequality comes from our choice of $k$. Observe that the above requires $\delta - e^{-k/L} > 0$, which holds because $k> L \ln(1/\delta)$.
    
    The Poisson approximation gives $\Pr_X[E] \le 2\Pr_Y[E]$ for any event $E$ that is monotically decreasing with the number of balls $k$ \cite[Corollary 5.11]{MU05}. Throwing more balls only decreases the probability that more than $\delta L$ are empty. Hence $\Pr_X[E] \le e^{-\lambda}$.
\end{proof}

\section{Discussion of sub-uniform recovery}\label{ap:sub-uniform}
Our main $L$-bit watermarking scheme allows $\delta\cdot L$ bits of the embedded message to be erased. It is natural to hope that the erased bits are distributed uniformly at random.
But an adversary might want to erase some bits more than others. For example, the higher-order bits if the message is a timestamp.

Can we 
ensure that the adversary cannot influence which indices of the message are 
erased? In the proof of \lemref{lem-msg-robust}, we transition to a game 
(Hybrid 3) where the indices are erased uniformly at random. Perhaps 
surprisingly, however, this does not imply that our scheme extracts a message 
with uniformly random erasures.\footnote{
    The implication does hold if the adversary only ever queries a single message.
}
We present an informal argument to provide intuition as to why this implication 
cannot hold.

Suppose the adversary is using our $L$-bit scheme and can choose which messages 
are embedded into the text. Then they could generate watermarked outputs 
$T_0 \gets \Wat_\sk(m_0, \prompt_0)$ and $T_1 \gets \Wat_\sk(m_1, \prompt_1)$, 
where $m_0 = 0^L$ and $m_1 = 0\|1^{L-1}$. Suppose they edit $T_0$ and $T_1$, editing 
blocks uniformly at random, ultimately outputting some $\hT$ that contains enough 
of $T_0$ and $T_1$ such that every index in $m_0$ and $m_1$ is 
embedded once in expectation. In this case, there would likely be two blocks 
in $\hT$ generated using $k_{1, 0}$, none generated using $k_{1,1}$, and one 
generated using $k_{i,b}$ for every other $i \in [L], b \in \bits$. Suppose our 
$\Extract$ algorithm is better at extracting when two blocks generated with the 
same key are present than it is when only one block is present. It would then be 
more likely that $\Extract$ recovers the first bit of the embedded message than 
any other bit.

Intuitively, this may not seem like an issue. In fact, by making it easier to 
recover certain bits, the adversary appears to be helping the watermarker. But this
is not strictly better than having uniformly random erasures. To see why, we can 
model the adversary's behavior by first erasing bits uniformly at random and then 
allowing the adversary to \emph{choose} which bits become ``unerased.'' We call
the resulting distribution of the indices of erased bits ``sub-uniform.'' A scheme 
that is secure under uniformly random erasures is not necessarily secure against 
a sub-uniform adversary. For instance, the Tardos fingerprinting code \cite{ACM:Tar08} is 
secure against uniform erasures. If the watermarking scheme is embedding codewords 
from the Tardos code, the adversary may be able to reveal bits that make it 
harder for them to be traced.

To be clear, this adversary is exceptionally weak. If a fingerprinting code is 
robust to uniform erasures and has the (informal) property that revealing more bits 
always improves the chance of tracing to a guilty party, then it will also be robust 
to sub-uniform erasures. This is an intuitive property for a fingerprinting code to 
have, although not all satisfy it (and most do not prove it).

Before we go on to prove any results, we formally define sub-uniform distributions
in \defref{def-sub-uniform}. Informally, we can compare the definition to the above
game as follows. First, we choose a uniformly random subset $Y$, then an adversary
can pick any $X\subseteq Y$. This corresponds exactly to  erasing uniformly random
indices of a message (those in $Y$) and then allowing an adversary to
unerase indices of its choice (those in $Y\setminus X$). 
The indices that remain are those in $X$.

\begin{definition}\label{def-sub-uniform}
    For a universe $\calU$, integer $0 \le s \le |\calU|$, and random
    variable of a uniformly random subset of size $s$, $U_{s}\subseteq \calU$,
    we say a random variable $V$ supported on $2^{\calU}$ is \emph{$s$-sub-uniform}
    if there exists a %
    distribution (coupling) $W = (X,Y)$ over $2^{\calU} \times 2^{\calU}$,
    with marginal distributions $X=V$ and $Y=U_s$, such that $X\subseteq Y$ always.
\end{definition}

\lemref{lem-sub-uniform} formalizes the type of distribution that our erasures 
follow. So long as fingerprinting codes are robust with respect to sub-uniform 
erasures, they can be used in a black-box way with our $L$-bit watermarking 
construction in \figref{fig-embedding} to create multi-user watermarks.

\begin{lemma}\label{lem-sub-uniform}
    Suppose $\wat'$ is a block-by-block zero-bit watermarking scheme that is undetectable,
    sound, and $R_1$-robustly detectable. Let $\Msg = (\MsgSetup, \Wat, \Extract)$
    be the $L$-bit watermarking scheme from \figref{fig-embedding} using $\wat'$.

    Let $I_\bot(\hat{m}) = \{i : i \in [L], \hat{m}[i]=\bot\}$ and $0\le \delta < 1$.
    Then, for all efficient $\A$, the set $I_\bot(\hat{m})$ is computationally indistinguishable
    from a $\lfloor \delta L \rfloor$-sub-uniform random variable
    in the probability experiment defined by
    \begin{itemize}
        \item $\sk \getsr \MsgSetup(1^{\lambda})$
        \item $\hat{T} \gets \A^{\Wat_\sk(\cdot,\cdot)}(1^\lambda)$
        \item $\hat{\msg} \gets \Extract_{\sk}(\hat{T})$
    \end{itemize}
    and conditioned on $R_k$ being satisfied.
\end{lemma}
\begin{proof}[Proof outline]
The proof follows via the techniques used in the proof of \lemref{lem-msg-robust}.
Notice that the code of Hybrid 3 will return a message with uniformly randomly 
distributed erasures. If we consider running the same adversary with $\wat$ in 
Hybrid 3 and in the original game, then we can construct a distribution $W$ based
on this adversary. A slight modification of
the distribution of erasures in Hybrid 3 is
indistinguishable from a uniformly random subset of size $\lfloor \delta L \rfloor$.
Specifically, we can add uniformly random erasures to $\hat{\msg}$ until it has exactly
$\lfloor \delta L \rfloor$ entries set to $\bot$. Call the erasures in this
distribution $I^3_\bot(\hat{\msg})$.
Based on the arguments in Hybrids 1 and 2, we know that the original game will only 
have strictly fewer $\bot$ entries than in Hybrid 3. So the joint distribution
$W=(I_\bot(\hat{\msg}), I^3_\bot(\hat{m}))$, is indistinguishable from a
$\lfloor \delta L \rfloor$-sub-uniform joint distribution.
\end{proof}

\fi

\end{document}